%% file: main.tex
\documentclass[11pt,letterpaper]{article}

\usepackage[letterpaper, left=1in, right=1in, top=0.9in, bottom=0.9in]{geometry}
\usepackage[american]{babel}
\usepackage[normalem]{ulem}
\usepackage{amsmath, amssymb, cases, amsthm}
    \usepackage{thmtools}
\usepackage[shortlabels]{enumitem}
\usepackage{mdframed}
\usepackage{bbm}
\usepackage{bm}
\usepackage{microtype}
\usepackage{xcolor}
\usepackage{makecell}
\usepackage{mathtools}
\usepackage{algorithmic}
\usepackage[procnumbered,ruled,vlined,linesnumbered]{algorithm2e}
\usepackage{float}
\usepackage{varwidth}
\usepackage{modletter}
\usepackage{tcolorbox}
\usepackage{tikz}
\usetikzlibrary{positioning, fit, arrows.meta, backgrounds}

\newtcolorbox{construction}[2][]
{
	colframe = gray!50,
	colback  = gray!10,
	coltitle = gray!10!black,
	left*=0mm, 
	before skip = 10pt,
	after skip = 10pt,
	title    = \textbf{\space\space #2},
	#1,
}

\SetKwInput{KwData}{Input}
\SetKwInput{KwResult}{Output}
\SetKwInput{KwGlobalVar}{Global variables}

\usepackage[bookmarks,colorlinks,breaklinks]{hyperref}
\hypersetup{urlcolor=blue, colorlinks=true, citecolor=green!50!black, linkcolor=blue}
\usepackage[capitalize,nosort,nameinlink]{cleveref}

\declaretheorem[numberwithin=section,refname={Theorem,Theorems},Refname={Theorem,Theorems}]{theorem}

\declaretheorem[numberlike=theorem]{lemma}

\declaretheorem[numberlike=theorem]{corollary}
\declaretheorem[numberlike=theorem]{definition}
\declaretheorem[numberlike=theorem]{claim}
\declaretheorem[numberlike=theorem,style=remark]{remark}

\declaretheorem[numberlike=theorem, refname={Observation,Observations},Refname={Observation,Observations},name={Observation}]{observation}

\theoremstyle{definition}
\declaretheorem[numberlike=theorem]{assumption}

\newcommand{\eps}{\epsilon}
\newcommand{\poly}{\mathrm{poly}}
\newcommand{\polylog}{\mathrm{polylog}}
\newcommand{\pram}{$\mathsf{PRAM}$\xspace}
\newcommand{\mtram}{$\mathsf{TRAM}$\xspace}

\newcommand{\Start}{\mathrm{start}}
\newcommand{\End}{\mathrm{end}}
\newcommand{\Succ}{\mathrm{succ}}
\newcommand{\Prec}{\mathrm{prec}}
\newcommand{\low}{\mathrm{low}}
\newcommand{\high}{\mathrm{high}}
\newcommand{\inner}{\mathrm{inner}}
\newcommand{\init}{\mathrm{init}}
\newcommand{\id}{\chi}

\newcommand{\aug}{\mathrm{aug}}

\newcommand{\far}{\mathrm{far}}
\newcommand{\rem}{\mathrm{rem}}
\newcommand{\pre}{\mathrm{pre}}
\newcommand{\CT}[2]{\mathrm{CT}(#1,#2)}
\newcommand{\abs}[1]{\left\lvert #1 \right\rvert}
\newcommand{\wt}{\widetilde}
\newcommand{\dtres}{d_{\mathrm{TH}}}
\renewcommand{\big}{\mathrm{big}}
\newcommand{\alg}[1]{\mathtt{#1}}
\renewcommand{\outer}{\mathrm{outer}}
\newcommand{\agg}{\boxplus}
\newcommand{\bigagg}{\mathop{\displaystyle\agg}}

\DeclareMathOperator{\dist}{dist}

\begin{document}
\sloppy
\pagenumbering{roman}
\title{
Parallel Small Vertex Connectivity in Near-Linear Work and Polylogarithmic Depth
}
\author{
Yonggang Jiang\thanks{MPI-INF and Saarland University, Germany, {yjiang@mpi-inf.mpg.de}}\and
Changki Yun\thanks{Seoul National University, Seoul, South Korea, {tamref.yun@snu.ac.kr}}
}
\date{}
\maketitle

\input{abstract.tex}

\newpage

\tableofcontents

\newpage
\pagenumbering{arabic}
\input{intro.tex}

\input{overview}
\input{preliminaries}

\input{unbalanced.tex}
\input{localcuts}
\input{rounding}

\bibliographystyle{alpha}
\bibliography{refs}

\appendix*
\input{missingproofs}
\input{datastructure}
\input{reductiontoBMM}

\end{document}

%% file: abstract.tex
\begin{abstract}

We present a randomized parallel algorithm in the {\sf PRAM} model for $k$-vertex connectivity. Given an undirected simple graph, our algorithm either finds a set of fewer than $k$ vertices whose removal disconnects the graph or reports that no such set exists. The algorithm runs in $O(m \cdot \text{poly}(k, \log n))$ work and $O(\text{poly}(k, \log n))$ depth, which is nearly optimal for any $k = \text{poly}(\log n)$. Prior to our work, algorithms with near-linear work and polylogarithmic depth were known only for $k=3$ [Miller, Ramachandran, STOC'87]; for $k=4$, sequential algorithms achieving near-linear time were known [Forster, Nanongkai, Yang, Saranurak, Yingchareonthawornchai, SODA'20], but no algorithm with near-linear work could achieve even sublinear (on $n$) depth.

\end{abstract}

%% file: intro.tex
\section{Introduction}

This paper studies the parallel computation of vertex connectivity. The problem of computing the (global) vertex connectivity of a simple, undirected graph involves finding the minimum number of vertices whose removal disconnects the graph (or reduces it to a singleton). This problem, along with the closely related problem of computing (global) edge connectivity—which seeks the minimum number of edges whose removal disconnects the graph—has been a classic and fundamental topic that has attracted researchers' attention for the past fifty years.

In the sequential model, both problems are well understood\footnote{In this paper, we focus on randomized algorithms that are correct with high probability. The problem of deterministic vertex connectivity remains widely open.}. For edge connectivity, a long line of research spanning several decades \cite{FF56,GH61, ET75, Po73, KT86, Ma87, NI92A, HO94, SW97, frank94, Gab95, NI92, Mat93, Kar99, KS96} culminated in a near-linear time algorithm \cite{Karger00}. For vertex connectivity, although progress has been much slower \cite{Kleitman1969methods,Tarjan72,HopcroftT73,Podderyugin1973algorithm,EvenT75,BeckerDDHKKMNRW82,LinialLW88,KanevskyR91,CheriyanT91,NagamochiI92,HenzingerRG00,Gabow06,NanongkaiSY19,ForsterNYSY20}, recent advances have reduced the problem to a polylogarithmic number of maximum flow computations \cite{LiNPSY21}. Combined with an almost linear time maximum flow algorithm \cite{ChenKLPGS22}, these results have brought the running time for vertex connectivity close to almost linear time\footnote{We use ``nearly linear'' to denote $m\cdot \polylog(n)$ and ``almost linear'' to denote $m^{1+o(1)}$.}.

\paragraph{Parallel computing.} In this paper, we focus on parallel computing, an important paradigm that offers both deep theoretical insights and practical benefits for large-scale graph processing in modern applications~\cite{JaJa92}. To capture the parallelism of an algorithm, we use the classical \pram{}\footnote{There exist variants of the \pram{} model that differ in the memory access capabilities of different machines. However, these variants are equivalent up to polylogarithmic factors. Since this paper does not attempt to optimize these factors, we do not concern ourselves with the different model definitions.}
 model~\cite{FortuneWyllie78}, which evaluates an algorithm in terms of \emph{work} (the total number of unit operations executed) and \emph{depth} (the length of the longest chain of sequential dependencies).

An ideal parallel algorithm should be \emph{work-efficient}---that is, its work should be within subpolynomial factors\footnote{We use $n^{o(1)}$ to denote a subpolynomial factor.} of the best known sequential running time (which, in our case, is almost linear work). We say that a parallel algorithm is \emph{highly parallel}\footnote{These notations also appear in \cite{Fineman20}.} if it is work-efficient and its depth is subpolynomial. For example, the best-known parallel edge connectivity algorithms are highly parallel: by adapting techniques from the sequential algorithm~\cite{Karger00}, edge connectivity was solved with nearly linear work and polylogarithmic depth~\cite{GeissmannG21,Lopez-MartinezM21}. Moreover, parallelism is foundational for numerous other computational models, including distributed computing and streaming algorithms. In the case of edge connectivity, subsequent work has demonstrated near-optimal performance in these models~\cite{MukhopadhyayN20, DoryEMN21,0001Z22}.

\paragraph{Existing results for parallel vertex connectivity (\Cref{tab:vertexconnectivity_exact}).} 
Given the positive results for edge connectivity, one might expect the closely related problem of vertex connectivity to also admit a highly parallel algorithm. However, progress on vertex connectivity algorithms has been much slower than that for edge connectivity. For clarity, we define \emph{$k$-vertex connectivity} as the task of finding fewer than $k$ vertices whose removal disconnects the graph, or reporting that no such set exists (i.e., the graph is $k$-vertex connected)\footnote{Note that $k$-vertex connectivity is a special case of vertex connectivity, and vertex connectivity has the same parallel complexity (up to polylogarithmic factors) if $k$ is not restricted---this can be achieved by binary search on $k$.}.

\begin{table}[ht]
\centering
\begin{tabular}{|l|l|l|l|}
\hline
\textbf{Setting} & \textbf{Work} & \textbf{Depth} & \textbf{Reference} \\ \hline
\(k=1\) (Connectivity) & \(O(m\log n)\) & \(O(\log n)\) & \cite{ShiloachV81} \\ \hline
\(k=2\) (Biconnectivity) & \(O(m\log n)\) & \(O(\log n)\) & \cite{TarjanV85} \\ \hline
\(k=3\) (Triconnectivity) & \(O(m\log^2n)\) & \(O(\log^2 n)\) & \cite{MillerR87} \\ \hline
\(k=4\) (Four-Connectivity)& \(O(n^2 \log^2 n)\) & \(O(\log^2 n)\) & \cite{KanevskyR87} \\ \hline
Small $k$ & \(\tO{n^{2}k^{4}}\) & \(O(k^{2}\log n)\) & \cite{CheriyanT91,CheriyanKT93} \\ \hline
Independent of $k$ & \(n^{\omega+o(1)}\) & \(n^{o(1)}\) & \cite{LinialLW88,blikstad2025globalvsstvertex} \\ \hline
Sequential Algorithm: Extremely small $k$ & \(mk^{O(k^2)}\) & \(\Omega(n)\) & \cite{SaranurakY22,Korhonen25} \\ \hline
Sequential Algorithm: Small $k$ & \(\tO{mk^2}\) & \(\Omega(n)\) & \cite{ForsterNYSY20} \\ \hline
Sequential Algorithm: Independent of $k$ & \(m^{1+o(1)}\) & \(\Omega(n)\) & \cite{LiNPSY21} \\ \hline
\textbf{Our Result} & \(\tO{m\cdot \poly(k)}\) & \(\tO{\poly(k)}\) & \\ \hline
\end{tabular}
\caption{Summary of existing algorithms for parallel $k$-vertex connectivity and state-of-the-art results in the sequential setting. We do not include results on restricted graph classes, as our focus is on general graphs.}
\label{tab:vertexconnectivity_exact}
\end{table}

When $k=1$, the problem is equivalent to connectivity, for which a well-known highly parallel algorithm exists \cite{ShiloachV81}. The study of parallel $k$-vertex connectivity for $k>1$ dates back to Tarjan and Vishkin \cite{TarjanV85}, who provided a highly parallel algorithm for $k=2$. Subsequently, Miller and Ramachandran \cite{MillerR87} provided a highly parallel algorithm for $k=3$. Surprisingly, these remain the only work-efficient parallel algorithms known for small values of $k$; in fact, no highly parallel algorithm is known even for $k=4$. For $k=4$, the best-known parallel algorithm remains that of Kanevsky and Ramachandran \cite{KanevskyR87}, with a work bound of $\tO{n^2}$\footnote{Throughout the paper, we use $\tO{\cdot }$ to hide polylogarithmic factors.}.

For general $k$, existing algorithms present a trade-off: they are either not work-efficient or are highly sequential. In particular, if one allows the work to be as high as $n^2$, then an algorithm with depth polynomial in $k$ exists \cite{CheriyanT91,CheriyanKT93}. Alternatively, if one permits the work to be as large as that required for matrix multiplication, i.e., $n^{\omega}$, then one can achieve subpolynomial depth for any $k$ \cite{LinialLW88,blikstad2025globalvsstvertex}. Finally, if depth is not a concern, there exist work-efficient sequential algorithms \cite{ForsterNYSY20,LiNPSY21,SaranurakY22,Korhonen25}, although their depth is at least linear in $n$.

In summary, for any $k\ge 4$, no work-efficient parallel algorithm is known that achieves even sublinear (in $n$) depth, let alone a highly parallel one.

\paragraph{Our results.} 
We present a randomized \pram{} algorithm for $k$-vertex connectivity that runs in nearly linear work and polylogarithmic depth for any $k = \polylog(n)$.

\begin{theorem}[Informal version of \Cref{thm:main-detail}]\label{thm:main}
There exists a randomized \pram{} algorithm for $k$-vertex connectivity that uses $\tO{m\cdot \poly(k)}$ work and $\tO{\poly(k)}$ depth.
\end{theorem}

This is the first work-efficient algorithm achieving even sublinear (in $n$) depth for $k$-vertex connectivity when $k\ge 4$, and it demonstrates that the problem is highly parallelizable for any $k = n^{o(1)}$.

\subsection{Technical Overview}

\paragraph{Reachability barrier.} One might wonder whether a highly parallel algorithm exists for the $k$-vertex connectivity problem for {\em every $k$}. This intuition stems from the fact that such algorithms exist for edge connectivity and that $k$-vertex connectivity can be solved in almost linear time for every $k$ in the sequential setting. However, recent work \cite{blikstad2025globalvsstvertex} demonstrates significant barriers: any highly parallel vertex connectivity algorithm would imply a highly parallel algorithm for {\em dense (s-t) reachability}\footnote{Here, \emph{dense} means that the input graph has $\Theta(n^2)$ edges, and (s-t) reachability means determining whether a given vertex $s$ can reach another given vertex $t$ in a directed graph.}. 

Reachability is arguably the most fundamental problem in directed graphs, and the existence of a highly parallel algorithm for reachability remains a notorious long-standing open problem. Even for dense graphs, the best work-efficient parallel algorithm for reachability still has a polynomial depth of $n^{1/2+o(1)}$ \cite{Fineman20,LiuJS19}. Achieving subpolynomial (i.e., $n^{o(1)}$) depth currently requires work on the order of $\tO{n^{\omega}}$\footnote{$n^{\omega}$ denotes the work of the best-known algorithm for matrix multiplication under subpolynomial depth; note that $\omega$ remains far from $2$.}, which essentially relies on a folklore method based on repeatedly squaring the adjacency matrix.
Thus, achieving a highly parallel algorithm for the vertex connectivity problem, while theoretically possible, is not currently within reach.

\paragraph{Restricting $k$ to sublinear.}
Given the above barrier, it is natural to focus on smaller values of $k$, e.g., $k=n^\epsilon$ for any constant $\epsilon>0$. Unfortunately, by extending the proof in \cite{blikstad2025globalvsstvertex}, we show in this paper that for $k=n^{\eps}$, a highly parallel algorithm for $k$-vertex connectivity would imply a highly parallel algorithm for a \emph{direct sum version} of dense reachability. We refer interested readers to \Cref{sec:reduction} for a detailed discussion. In summary, it appears out of reach to design a highly parallel algorithm for $k$-vertex connectivity when $k$ is polynomial in $n$. In light of this barrier, our main theorem \Cref{thm:main} presents the best result (up to subpolynomial factors) achievable given the current understanding of these reachability barriars.

\paragraph{Restricting $k$ to subpolynomial.} Given the above barriers, it is natural to restrict our focus to $k = n^{o(1)}$. There are two work-efficient sequential algorithms \cite{ForsterNYSY20,LiNPSY21} for $k=n^{o(1)}$. {\em Can we adapt their ideas to design a highly parallel algorithm?} While parallel reachability is not formally a barrier in this case, it reemerges as an obstruction: \emph{both algorithms make at least one call to solve reachability as a subroutine}:

\begin{itemize}
    \item The almost linear-time algorithm \cite{LiNPSY21} reduces vertex connectivity to an \emph{exact} max flow computation; this reduction has also been implemented in a parallel model \cite{blikstad2025globalvsstvertex}. However, even in an undirected graph, exact max flow subsumes directed reachability.\footnote{Readers familiar with this area might notice that $k$-max flow (finding the flow value up to $k$) or $(1-\eps)$-approximate undirected max flow do not subsume reachability, and indeed are highly parallelizable, as we will explain later. However, the reduction in \cite{LiNPSY21} specifically requires an exact max flow algorithm with unbounded value because of the \emph{isolating cut} procedure. It remains a challenging and complex open problem to use approximate max flow to achieve a \emph{vertex} isolating cut (the edge version is known via fair cuts \cite{0006NPS23}, but not the vertex version).}
    \item The $\tO{mk^2}$ time algorithm \cite{ForsterNYSY20} requires $k$ iterations of finding an \emph{augmenting path} in a residual graph (which is directed), and thus necessitates at least one directed reachability call.
\end{itemize}

Therefore, attempting to make these algorithms highly parallel would again require solving parallel reachability.\footnote{For even smaller $k$, we note that the $mk^{O(k^2)}$ algorithms \cite{SaranurakY22,Korhonen25} does not use reachability, and have the potential to be parallelizable (although there are sequential procedures in their algorithms, and parallelizing them may not be trivial). Nonetheless, their techniques inherently result in work that is exponential in $k$, which is much worse than our work bound.}

\paragraph{Our Techniques: Avoiding Reachability.} Given the above discussion, one may ask: can we circumvent this barrier for $k = n^{o(1)}$, or does it inherently subsume reachability as in the case of $k=n^\epsilon$? Before our work, this question remained unclear. Our main result \Cref{thm:main} provides a positive answer.

Our starting point is the algorithm of \cite{ForsterNYSY20}. The idea in \cite{ForsterNYSY20} can be viewed as localizing the Ford-Fulkerson algorithm for $k$-max flow (i.e., computing the flow value up to $k$). A \emph{localized} algorithm, intuitively, explores only a small portion of the input graph. For example, a DFS procedure that stops after exploring a limited number of edges can be considered a localized DFS. The benefit of such localization is that the total work can be even \emph{sublinear} in the size of the graph.

Rather than localizing the inherently sequential Ford-Fulkerson method, we localize a \emph{parallel $k$-max flow algorithm}.\footnote{Importantly, we consider \emph{vertex-capacitated max flow}, so edge-capacitated algorithms are not applicable.} It is well known that the \emph{multiplicative weight update (MWU)} method can be used to solve max flow-related problems \cite{GargK07,BernsteinGS21} and, in particular, $k$-max flow in $\tO{m\cdot \poly(k)}$ work and $\tO{\poly(k)}$ depth, making it a promising candidate for our algorithm. This technique reduces the $k$-max flow problem to $\poly(k)$ rounds of computing \emph{approximate shortest paths} ({\sf ApxSP}), which can be solved in $\tO{m\cdot \poly(k)}$ work and $\tO{1}$ depth \cite{AndoniSZ20,li2020faster,RozhonGHZL22}.

Nonetheless, a significant challenge remains: how can one localize a parallel {\sf ApxSP} algorithm? Intuitively, this appears difficult because parallel {\sf ApxSP} algorithms typically process the entire graph, ensuring that as the source extends the shortest path tree into a region, the necessary information is already available there. This global dependency distinguishes them from DFS or Dijkstra’s algorithm, which are easier to localize due to their sequential exploration that allows the algorithm to halt exploration at an appropriate point.

We overcome this challenge by using data structures to maintain a global \emph{single-source distance sparsifier} specifically designed for the MWU framework with sublinear size. It then suffices to run a black-box {\sf ApxSP} algorithm on the sparsifier and map the results back to the original graph. Although mapping back to the original graph might incur significant work, this step can be localized by leveraging the global information provided by the sparsified graph. This approach achieves overall sublinear work and low depth. The details of this technique are presented in \Cref{sec:overview}.

%% file: overview.tex
\section{Overview}\label{sec:overview}

In this section, we present an overview of \Cref{thm:main}: solving $k$-vertex connectivity in $\tO{m\cdot \poly(k)}$ work and $\tO{\poly(k)}$ depth. In \Cref{subsec:overviewwarmup}, we explain how to use the multiplicative weight update method to solve $s$-$t$ $k$-vertex connectivity (adapted from \cite{GargK07,BernsteinGS21}). In \Cref{subsec:overviewframework}, we describe the framework for solving (global) $k$-vertex connectivity under the assumption of an efficient \emph{local cuts} subroutine. Finally, in \Cref{subsec:overviewlocalcuts}, we detail the local cut subroutine, which is the most technical component and our main technical contribution.

We first introduce some necessary definitions.

\paragraph{Basic Notations.} In this section, we consider undirected simple graphs $G=(V,E)$. The degree of a vertex $v$ is denoted by $\deg_G(v)$, and for a vertex set $A$, we define $\deg_G(A)=\sum_{v\in A}\deg_G(v)$. We use $V(G),E(G)$ to denote the vertex set and edge set of a graph $G$.

\paragraph{Paths and Vertex-Length.} For a path $P$, we denote its vertex set and edge set by $V(P)$ and $E(P)$, respectively. A path is called an \emph{$(s,t)$-path} if it starts at $s$ and ends at $t$. We say that $P$ is a \emph{non-trivial path} if $V(P)$ contains more than one vertex. Given a \emph{vertex-length} function $w:V\to\bbR$, the \emph{vertex-length} (or simply the length) of $P$ (with respect to $w$) is defined as
\[
w(P)=\sum_{v\in V(P)}w(v).
\]

\paragraph{Shortest Paths.} An $(s,t)$-shortest path (with respect to $w:V\to\bbR$) is an $(s,t)$-path with the minimum length, denoted by $\dist_{G,w}(s,t)$. An $(1+\eps)$-approximate $(s,t)$-shortest path is a path whose length lies between $\dist_{G,w}(s,t)$ and $(1+\eps)\cdot \dist_{G,w}(s,t)$. There is a \emph{deterministic} algorithm finding $(1+\eps)$-approximate $(s,t)$-shortest path in $\tO{m/\eps^2}$ work and $\tO{1}$ depth \cite{RozhonGHZL22}, although originally written for \emph{edge length}, it can be easily adjusted to vertex length (see \Cref{thm:sssp}).

\paragraph{Vertex Cuts.} A vertex cut $(L,S,R)$ is a partition of the vertex set $V$ into three sets $L,S,R$ such that there are no edges between $L$ and $R$. We refer to $|S|$ as the size of the vertex cut $(L,S,R)$. The vertex cut $(L,S,R)$ is called an $(s,t)$-vertex cut if $s\in L$ and $t\in R$. A graph $G$ is said to be $k$-vertex connected if it contains no vertex cut of size less than $k$. Following the convention, we sometimes also call $S$ as a vertex cut.

For simplicity, in this overview, we focus on the \emph{value} version of $k$-vertex connectivity, i.e., we are only interested in testing whether a graph is $k$-vertex connected.%
\footnote{Outputting a vertex cut, rather than merely testing $k$-vertex connectivity, requires converting a certificate (in our case, a fractional cut) into an (integral) cut. This conversion, referred to as \emph{rounding}, is explained in \Cref{sec:rounding}. Although sequential algorithms for rounding exist, to the best of our knowledge there is no highly parallel rounding algorithm; hence, our algorithm might have independent interest.}

\subsection{Warm-up: $s$-$t$ $k$-Vertex Connectivity Using Multiplicative Weight Updates}\label{subsec:overviewwarmup}

In this section, we present a warm-up algorithm that employs the multiplicative weight updates (MWU) framework to solve the following problem: given an undirected simple graph $G=(V,E)$, two vertices $s,t$, and an integer $k$, determine whether there exists an $(s,t)$-vertex cut of size less than $k$. We refer to this problem as \emph{$s$-$t$ $k$-vertex connectivity}. The algorithm runs in $\tO{m\cdot \poly(k)}$ work and $\tO{\poly(k)}$ depth. A formal statement and proof can be found in \Cref{lem:stvertexcut}.

\paragraph{MWU Framework for $s$-$t$ $k$-Vertex Connectivity.} The framework is based on the approach in \cite{GargK07}. Although it was originally designed for edge cuts, it has been observed that by using vertex lengths, the same framework can be adapted to find vertex cuts \cite{BernsteinGS21}. While we omit the detailed analysis and exact parameter settings here, the algorithm roughly proceeds as follows:

\begin{description}
    \item[Initialization.] Initialize a weight function $w:V\to\bbR$ by setting $w(s)=w(t)=0$, and $w(v)=1$ for every $v\in V\setminus\{s,t\}$. The algorithm then performs weight update steps for $\poly(k)$ rounds.
    \item[Approximate Shortest Path.] In each round, compute a $(1+\frac{1}{\poly(k)})$-approximate $(s,t)$-shortest path $P$ with respect to $w$ (using the algorithm in \cite{RozhonGHZL22}). The algorithm returns `yes' (indicating that there exists an $(s,t)$-vertex cut of size less than $k$) if $w(P)$ is sufficiently large---specifically, if 
    \[
    w(P) \ge \frac{\sum_{v\in V}w(v)}{k-0.5}.
    \]
    \item[Weight Updates.] Update the weight of every vertex $v$ on the path $P$ by setting
    \[
    w(v) \leftarrow \left(1+\frac{1}{\poly(k)}\right)\cdot w(v).
    \]
    \item[Termination.] If, after all $\poly(k)$ rounds, the algorithm has not returned `yes', then it returns `no'.
\end{description}

\paragraph{Intuition.} We now provide a brief intuition for the MWU framework. 

An $(s,t)$-path is said to be \emph{blocked} if its length with respect to $w$ is at least 
$\frac{\sum_{v\in V}w(v)}{k-1}$.
Observe that if there exists a vertex set $S$ of size $k-1$ whose removal disconnects $s$ from $t$, then assigning a weight of $1/(k-1)$ to each vertex in $S$ would block every $(s,t)$-path. The reverse direction can also be shown (see \Cref{lem:fractional_to_integral} for more details). 

In each round, the algorithm searches for an unblocked $(s,t)$-path and increases the weights on its vertices, thereby raising the likelihood that the path becomes blocked in subsequent rounds. If, during any round, the computed approximate $(s,t)$-shortest path has length at least $\frac{\sum_{v\in V}w(v)}{k-0.5}$ (note that the $k-0.5$ in the denominator compensates for the approximation loss), then every $(s,t)$-path is blocked, and the algorithm outputs `yes'. Otherwise, the classical MWU analysis implies that it is impossible to separate $s$ from $t$ by deleting $k-1$ vertices.

\paragraph{Summary.} Assuming that $k$ is polylogarithmic in $n$, the algorithm is highly parallel: it performs $\poly(k)$ rounds, with each round involving the computation of a $(1+\frac{1}{\poly(k)})$-approximate shortest path. According to \cite{RozhonGHZL22}, such a shortest path algorithm can be implemented in $\tO{m\cdot \poly(k)}$ work and $\tO{\poly(k)}$ depth \emph{deterministically}.

Note that the entire algorithm is \textbf{deterministic}. Thus, given $s,t\in V$, we can deterministically decide whether there exists an $(s,t)$-vertex cut of size less than $k$ in $\tO{m\cdot \poly(k)}$ work and $\tO{\poly(k)}$ depth.

\subsection{A Framework for Solving Vertex Connectivity via Local Cuts}\label{subsec:overviewframework}

Note that one can solve global $k$-vertex connectivity by applying the $s$-$t$ $k$-vertex connectivity algorithm to every pair of vertices $s,t\in V$. However, this naive approach is not work-efficient when using the $\tO{m\cdot \poly(k)}$-work algorithm described in \Cref{subsec:overviewwarmup}.

In this section, we briefly review the framework of \cite{ForsterNYSY20}, which reduces the global $k$-vertex connectivity problem to that of finding \emph{local cuts}.

Throughout this section, we denote by $(L,S,R)$ a vertex cut of $G$ of size less than $k$ satisfying $\deg_G(L)\le \deg_G(R)$ (if $\deg_G(L)>\deg_G(R)$, swapping $L$ and $R$ yields a cut of the same size). Given that such a cut exists, our goal is to certify its existence.

\paragraph{Balanced case: $\deg_G(L)=\Omega(m)$.} Suppose that $\deg_G(L)=\Omega(m)$ (and consequently $\deg_G(R)\ge\deg_G(L)\ge\Omega(m)$). In this case, we independently sample a pair of vertices $(s,t)$ with probabilities proportional to their degrees, i.e.,
\begin{equation}\label{eq:probability}
    \Pr[s \text{ is chosen}] = \frac{\deg_G(s)}{\deg_G(V)} \quad \text{and} \quad \Pr[t \text{ is chosen}] = \frac{\deg_G(t)}{\deg_G(V)}.
\end{equation}

We then run the algorithm from \Cref{subsec:overviewwarmup} for $s$-$t$ $k$-vertex connectivity and output `yes' if an $(s,t)$-vertex cut of size less than $k$ is found.

This algorithm exhibits a one-sided error: if it outputs `yes', then a vertex cut of $G$ of size less than $k$ certainly exists (because the algorithm described in \Cref{subsec:overviewwarmup} is deterministic). Conversely, a `no' output only indicates that the sampled pair did not consist of one vertex from $L$ and one from $R$. Since the error probability is bounded away from $1$, repeating the procedure for $\tO{1}$ rounds suffices to boost the overall success probability to high probability.

\paragraph{Unbalanced case: Reduction to local cuts.} Now suppose that $\deg_G(L)$ is polynomially smaller than $m$, say $\deg_G(L)=\Theta(m^{0.9})$. Under the assumption that $k=\polylog(n)$, it follows that $\deg_G(R)=\Omega(m)$. (Note: the only edges not adjacent to $L$ or $R$ are inside $S$, which can be at most $k^2$ of them.) This scenario captures the hardest case; the other cases (when $\deg_G(L)=\omega(m^{0.9})$ or $\deg_G(L)=o(m^{0.9})$) can be handled similarly.

The sampling technique from the balanced case does not work here: to sample a vertex pair $(s,t)$ with $s\in L$ and $t\in R$, one needs roughly $m^{0.1}$ independent samples to achieve a constant probability of success. If we were to run an $\tO{m\cdot \poly(k)}$-work algorithm for each sampled pair, the total work would become $\tO{m^{1.1}\cdot \poly(k)}$, which is far from our target of work-efficiency.

To reduce the work, our idea is to solve $s$-$t$ $k$-vertex connectivity in \emph{sublinear (in $m$)} work, more precisely in roughly $\tO{m^{0.9}\cdot \poly(k)}$ work. Then, taking $m^{0.1}$ samples, the total work becomes
\[
m^{0.1}\cdot \tO{m^{0.9}\cdot \poly(k)} = \tO{m\cdot \poly(k)},
\]
which meets our efficiency goal.

In general, solving $s$-$t$ $k$-vertex connectivity in sublinear work is not possible since it requires reading the entire input graph. However, this becomes feasible when our goal is only to certify the existence of an $(s,t)$-vertex cut of size less than $k$ with $\deg_G(L)=\Theta(m^{0.9})$. For example, when $k=1$ (so that $S=\emptyset$), one can trivially determine $L$ in $O(m^{0.9})$ work by performing a BFS from $s$—although this may incur large depth (a challenge we address in the next section). Moreover, if such an $L$ does not exist, the algorithm can stop after exploring more than $m^{0.9}$ edges. This example provides intuition for why we can hope for an algorithm that operates in sublinear work.

Motivated by this observation, we define $(L,S,R)$ to be an \emph{$(s,\mu,t)$-local cut} if $|L|=\Theta(\mu)$ and $s\in L$, $t\in R$. Consequently, to solve parallel $k$-vertex connectivity, it suffices to prove the following lemma.

\begin{lemma}[Informal version of \Cref{lem:localvertexcut}]\label{lem:overviewlocalcut}
    Given $G=(V,E)$, $s,t\in V$, and integers $k,\mu$, decide whether there is an $(s,\mu,t)$-local cut of size less than $k$ in $\tO{\mu\cdot \poly(k)}$ work and $\tO{\poly(k)}$ depth.
\end{lemma}

\paragraph{Summary.} Assuming the correctness of \Cref{lem:overviewlocalcut}, we obtain an algorithm for $k$-vertex connectivity that runs in $\tO{m\cdot \poly(k)}$ work and $\tO{\poly(k)}$ depth by following these steps:

\begin{description}
    \item[Step 1: Guess the size of $\deg_G(L)$.] Iterate over candidate values of $\mu$, where $\mu$ ranges over powers of $2$.
    \item[Step 2: Sample vertex pairs.] For each guessed $\mu$, independently sample $m/\mu$ pairs of vertices according to their degrees (see \Cref{eq:probability}).
    \item[Step 3: Local cuts.] For every sampled pair $(s,t)$, invoke the local cuts procedure described in \Cref{lem:overviewlocalcut} to certify the existence of an $(s,\mu,t)$-local cut, and output `yes' if such a cut is certified.
    \item[Ending.] If none of the iterations returns `yes', then output `no'.
\end{description}

A formal description of this framework is provided in \Cref{sec:parallelvertexconnectivity}. Our main technical contributions and challenges lie in proving \Cref{lem:overviewlocalcut}, which we detail in the next section.

\subsection{Parallel Local Cuts}\label{subsec:overviewlocalcuts}

In this section, we provide an overview of \Cref{sec:localcuts} and an informal proof of \Cref{lem:overviewlocalcut}. 
Remember that we assume the existence of a $(s,\mu,t)$-local cut denoted by $(L,S,R)$, meaning $s\in L,t\in R$ and $\deg_G(L)=\Theta(\mu)$. Our goal is to certify the existence of $(L,S,R)$.

Throughout this section, we make the following simplifying assumption:
\begin{assumption}\label{assumption}
    The graph has maximum degree $O(\mu)$.
\end{assumption}
Removing \Cref{assumption} requires a more involved modification of the algorithm by splitting the graph into high-degree and low-degree parts. This modification does not affect the high-level idea, so we omit it from the overview. Intuitively, \Cref{assumption} is reasonable because we are guaranteed that $\deg_G(L)=\Theta(\mu)$, which implies that every vertex $u\in L$ satisfies $\deg_G(u)=O(\mu)$. Thus, if we only focus on $L$, we need only consider vertices with degree bounded by $O(\mu)$.

\paragraph{A Naive Approach.} One could employ the MWU framework described in \Cref{subsec:overviewwarmup}. Recall that this framework requires $\poly(k)$ rounds, each involving the computation of a $(1+\frac{1}{\poly(k)})$-approximate $(s,t)$-shortest path with respect to a weight function $w:V\to\bbR$. A bottleneck arises, however: even a single such approximate shortest path computation takes $\tO{m\cdot \poly(k)}$ work and $\tO{1}$ depth, which is far from our target of $\tO{\mu\cdot \poly(k)}$ work when $\mu \ll m$. The goal of this section is to reduce the work per round to $\tO{\mu\cdot\poly(k)}$.

For simplicity, throughout this section, whenever we refer to an \emph{approximate shortest path}, we mean a $(1+\frac{1}{\poly(k)})$-approximate shortest path.

\subsubsection{(Single-Source) Length Sparsifier} 

Our approach is to maintain an \emph{$s$-source length sparsifier} $\hG{}$ together with an associated weight function $\hat{w}:V(\hG{})\to \bbR$, which satisfy the following properties:
\begin{itemize}
    \item \emph{(Graph contraction)} $\hG{}$ is obtained by contracting vertex subsets of $G$. Specifically, vertices in $V(\hG{})$ correspond to disjoint subsets (or \emph{cluster}) of $V(G)$, and the union of these clusters is $V(G)$ (i.e., clusters form a \emph{partition} of $V(G)$). An edge exists between two vertices in $V(\hG{})$ if and only if there is an edge in $E(G)$ connecting the corresponding clusters.
    \item \emph{(Length preserving)} There is a natural many-to-one mapping from any path $P$ in $G$ to a path $\hP{}$ in $\hG{}$, obtained by contracting consecutive vertices in $P$ that lie in the same cluster (denoted by $\hP{}=\hG{}(P)$). We guarantee that for every non-trivial simple path $P$ starting from $s$,
    \[
    \Bigl(1-\frac{1}{\poly(n)}\Bigr)\cdot w(P)\le \hat{w}(\hG{}(P))\le w(P),
    \]
    \item \emph{(Sparsity)} The graph $\hG{}$ contains at most $\tO{\mu\cdot \poly(k)}$ edges.
\end{itemize}

We will later explain how to construct such a length sparsifier and weight function; for now, assume that $\hG{}$ is provided for free.

\paragraph{Finding approximate shortest paths in $\hG{}$.} Notice that the length of a path $P$ in $G$ and that of its contracted image $\hG{}(P)$ differ only by a factor of $(1+\frac{1}{\poly(n)})$, which is negligible compared to the $(1+\frac{1}{\poly(k)})$-approximation factor of the approximate shortest path algorithm. Suppose that $s$ lies in a cluster $C_s$ and $t$ in a cluster $C_t$, and let $\hP{}$ be an approximate shortest path from $C_s$ to $C_t$ in $\hG{}$. Then any $(s,t)$-path $P$ in $G$ satisfying $\hG{}(P)=\hP{}$ is an approximate shortest path from $s$ to $t$, by virtue of the length preserving property of $\hG{}$.

Thus, we first compute a $(C_s,C_t)$-shortest path, denoted by $\hP{}_{s,t}$, in $\hG{}$. Since $\hG{}$ is sparse, we can compute $\hP{}_{s,t}$ in $\tO{\mu\cdot \poly(k)}$ work and $\tO{1}$ depth using the black-box approximate shortest path algorithm from \cite{RozhonGHZL22}.

After obtaining $\hP{}_{s,t}$, the next step is to find a corresponding path $P_{s,t}$ in $G$ such that $\hG{}(P_{s,t})=\hP{}_{s,t}$, so that we can perform weight updates accordingly.

\paragraph{An inefficient naive approach.} Suppose $\hP{}_{s,t}$ has the form
\[
\hP{}_{s,t}=(C_1,C_2,\dots,C_{z-1},C_z),
\]
where each $C_i$ is a cluster of $G$, with $s\in C_1$ and $t\in C_z$. A corresponding path $P_{s,t}$ in $G$ can be constructed in $\tO{\sum_{i\in[z]}\deg_G(C_i)}$ work and $\tO{1}$ depth by finding, within each cluster $C_i$, a subpath such that the concatenation of these subpaths forms an $(s,t)$-path in $G$.\footnote{We omit the details regarding the mapping from vertices in $\hG{}$ back to clusters in $G$, and how to efficiently find an undirected path in a cluster. These issues can be resolved using appropriate data structures \Cref{lem:dynamicspanningforest}.} However, since $\sum_{i\in[z]}\deg_G(C_i)$ can be as large as $m$, a more efficient algorithm for finding $P_{s,t}$ is necessary.

At first glance, this large work appears unavoidable since it could be the case that any path connecting $s$ and $t$ in $G$ has length at least $\Omega(m)$. Therefore, rather than explicitly finding a full $(s,t)$-path, we employ the idea of \emph{random terminating}, as introduced in \cite{ForsterNYSY20}.

\subsubsection{Random Termination and Modified Weight Updates}\label{subsubsec:randomtermination}
As explained in the previous section, we do not aim to compute a complete $(s,t)$-path in $G$. Instead, we compute an $(s,r)$-path $P_{s,r}$ in $G$ that satisfies:
\begin{enumerate}[label=(\roman*)]
    \item $r\in R$, and
    \item $P_{s,r}$ is an approximate $(s,t)$-shortest path.
\end{enumerate}

\paragraph{Modification of Weight Updates.} Although the standard multiplicative weight update (MWU) framework is designed to update the weights along an approximate $(s,t)$-shortest path in each iteration, a slight modification allows us to update the weights along an $(s,r)$-shortest path for some $r\in R$. Intuitively, increasing the weights along an $(s,r)$-shortest path (which must pass through $S$) also increases the chance of blocking paths from $s$ to $t$, since both $r$ and $t$ lie in $R$. For a formal justification that this modification does not affect the MWU proof, see \Cref{lem:mwu}. To summarize, the modified MWU framework proceeds as follows:
\begin{enumerate}[label=(\roman*)]
    \item Find a $(s,r)$-shortest path $P_{s,r}$ from $s$ to some node $r\in R$.
    \item Increase the weights of the vertices in $P_{s,r}$.
\end{enumerate}

\paragraph{Finding $r$ and $P_{s,r}$.} The term \emph{random termination} refers to choosing the vertex $r$ at random. We select $r$ and construct $P_{s,r}$ as follows:
\begin{enumerate}[\arabic*]
    \item Consider the approximate shortest path $\hP{}_{s,t}=(C_1,C_2,\dots,C_{z-1},C_z)$ in $\hG{}$ from the previous section. Use binary search to identify an index $i^*$ such that
    \[
    \sum_{j\le i^*}\deg_G(C_j)=\Theta(\mu \cdot \poly(k)).
    \]
    \textbf{Remark.} Note that such an index $i^*$ might not exist if some cluster $C_i$ has an exceptionally high degree. For simplicity, we assume that such an $i^*$ exists. To remove this assumption, one must split out a small part of the first exceptionally high-degree cluster using the appropriate data structure (see \Cref{lem:dynamicspanningforest}).
    
    \item Let $C^*=\bigcup_{j\le i^*}C_j$. Sample a vertex $r$ from $C^*$ with probability proportional to its degree, i.e.,
    \[
    \Pr[r\text{ is sampled}] = \frac{\deg_G(r)}{\sum_{v\in C^*}\deg_G(v)}.
    \]
    
    \item Find a path $P_{s,r}$ in $G[C^*]$ from $s$ to $r$ such that $\hG(P_{s,r})$ is a prefix of $\hP_{s,t}$.
\end{enumerate}

\paragraph{Complexity.} The first step requires fast computation of $\deg_G(C_i)$, which is achievable using the data structure described in \Cref{lem:dynamicspanningforest}. Since the second and third steps are confined to $G[C^*]$, the total work is reduced to $\tO{\deg_G(C^*)}=\tO{\mu\cdot \poly(k)}$, which meets our efficiency goal.

We note that we omit many details about the data structure in \Cref{lem:dynamicspanningforest}, as it is a modification of previous works and we do not claim novelty here. The data structure requires an initialization with $\tO{m}$ work; therefore, our local cuts algorithm is not strictly sublinear in $m$. However, after initialization, all updates and queries operate within the desired work bound of $\tO{\mu\cdot \poly(k)}$.  Since the initialization is shared across all local cuts computations in the framework described in \Cref{subsec:overviewframework}, the $\tO{m}$ work is not an issue.

\paragraph{Correctness: (i) $r\in R$, (ii) $P_{s,r}$ is an approximate $(s,t)$-shortest path.} We now argue that $P_{s,r}$ satisfies the two desired conditions:
\begin{enumerate}[label=(\roman*)]
    \item \emph{($r\in R$)} Notice that $\deg_G(L\cup S)\le O(k\cdot \mu)$, since under our bounded-degree \Cref{assumption} we have $\deg_G(L)=\Theta(\mu)$ and $\deg_G(S)=O(k\cdot \mu)$. Therefore, when sampling $r$ from $C^*$, where $\deg_G(C^*)=\Theta(\mu \cdot \poly(k))$, the probability that $r\in R$ is at least $1-\frac{1}{\poly(k)}$ (we will explain later why this probability suffices).
    \item \emph{($P_{s,r}$ is an approximate shortest path from $s$ to $r$)} Since $\hP{}_{s,t}$ is an approximate shortest path in $\hG{}$ and $\hG(P_{s,r})$ corresponds to a prefix of $\hP{}_{s,t}$, it follows that $\hG(P_{s,r})$ is also an approximate shortest path.\footnote{This might not hold in general, but \cite{RozhonGHZL22} outputs an approximate shortest path tree, meaning that any prefix of an approximate shortest path is itself an approximate shortest path.} Moreover, by the length preserving property of $\hG{}$, $P_{s,r}$ is an approximate shortest path in $G$.
\end{enumerate}

This completes the construction of $P_{s,r}$.

\paragraph{Overall Algorithm.} In each of the $\poly(k)$ rounds of the algorithm, we determine $r$ and $P_{s,r}$ and perform weight updates along $P_{s,r}$. Given that the success probability (i.e., the event that $r\in R$) in each round is at least $1-\frac{1}{\poly(k)}$, the overall success probability is bounded below by a positive constant. As explained in \Cref{subsec:overviewwarmup}, our final $k$-vertex connectivity algorithm exhibits one-sided error, so a constant success probability suffices.

The following observation is important for the next section.
\begin{observation}\label{observation}
    $\deg_G(V(P))\le O(\mu\cdot \poly(k))$.
\end{observation}

This follows from the fact that $V(P)\subseteq C^*$ and $\deg_G(C^*)=\Theta(\mu\cdot \poly(k))$.

\subsubsection{Constructing $\hG$}\label{subsubsec:sparsifier}

It remains to show how to construct an $s$-source length sparsifier $\hG$ and a weight function $\hat{w}$ from $G$ and $w$ in each round of computation. Recall that our goal is to contract clusters (i.e., form a partition of $V$) to approximately preserve the lengths of all non-trivial simple paths starting from $s$.

We claim that once the following conditions are satisfied, then $\hG$ and its associated weight function (denoted here as $\hat{w}$) approximately preserve the lengths of all non-trivial paths starting from $s$. We call a vertex $v$ with $w(v)=1$ a \emph{trivial} vertex, and all other vertices \emph{non-trivial}.

\begin{enumerate}[label=(\arabic*)]
    \item Every non-trivial path in $G$ starting at $s$ has length (with respect to $w$) at least $\poly(n)$.
    \item Every cluster $C$ contains either (i) a single non-trivial vertex $v$, in which case we set $\hat{w}(C)=w(v)$, or (ii) only trivial vertices, in which case we set $\hat{w}(C)=1$. 
\end{enumerate}

\paragraph{Conditions (1) and (2) imply length preservation.}  
To see why these conditions yield (approximate) length preservation, consider a non-trivial simple path $P$ starting from $s$. Contracting consecutive vertices in $P$ that lie in the same cluster can decrease the path length by at most $n$ (because every vertex being contracted is trivial and contributes a length of $1$ with respect to $w$). Since $w(P)\ge \poly(n)$, we have
\[
\Bigl(1-\frac{1}{\poly(n)}\Bigr)\cdot w(P) \le w(P)-n \le \hat{w}(\hG(P)) \le w(P).
\]

\paragraph{Adjusting initial weights to satisfy condition (1).}  
To ensure that condition (1) holds, we modify the initialization of the weight function $w$ as follows:
\[
w(u) =
\begin{cases}
    0 & \text{if } u = s, \\
    \poly(n) & \text{if } u \in N_{G}(s), \\
    1 & \text{otherwise},
\end{cases}
\]
where $N_G(s)$ denotes the set of all neighbors of $s$ in $G$. Although the default MWU initialization assigns a uniform weight of $1$ to every vertex (except $s$), the MWU analysis shows that assigning an initial weight of $\poly(n)$ to some vertices incurs only an additional $O(\log n)$ cost in the overall complexity. Setting $w(u)=\poly(n)$ for all $u\in N_G(s)$ guarantees that every non-trivial path in $G$ starting from $s$ has length at least $\poly(n)$, as every such path must pass through a vertex in $N_G(s)$. Since weights only increase during the algorithm, condition (1) is maintained in all rounds.

\paragraph{Satisfying condition (2) by contracting trivial vertices.}  
To enforce condition (2), we construct $\hG$ by contracting all connected components formed solely by trivial vertices, using the weight function $w$ initialized as above. We omit the detailed procedure for efficiently contracting these components; this can be accomplished using appropriate data structures (see \Cref{lem:dynamicspanningforest}).

\paragraph{$\hG$ satisfies the sparsity property.}  
It remains to show that $\hG$ is sparse, meaning it contains at most $\tO{\mu\cdot \poly(k)}$ edges. By the construction of $\hG$, every edge in $\hG$ must have at least one endpoint corresponding to a non-trivial vertex—otherwise, the edge would have been contracted. Thus, we have
\[
|E(\hG)| \le \sum_{\substack{v \text{ is non-trivial}}} \deg_G(v).
\]
Initially, only vertices in $\{s\}\cup N_G(s)$ are non-trivial. Since $\{s\}\cup N_G(s)\subseteq L\cup S$ (with $s\in L$), and given that $\deg_G(L)=\Theta(\mu)$ (because $(L,S,R)$ is a $(\mu,s,t)$-local cut) and \Cref{assumption} implies $\deg_G(S)=O(k\cdot \mu)$, we obtain
\[
\deg_G(\{s\}\cup N_G(s)) \le O(k\cdot \mu).
\]
In each round, after performing weight updates, some trivial vertices may become non-trivial. However, since we update weights only along the path $P_{s,r}$ (as constructed in the previous section) and by \Cref{observation} the total increase in $\sum_{v\text{ is non-trivial}}\deg_G(v)$ is at most $\tO{\mu\cdot \poly(k)}$ per round, the total number of edges in $\hG$ remains $\tO{\mu\cdot \poly(k)}$.

%% file: preliminaries.tex
\section{Preliminaries} \label{sec:prelims}
\paragraph{Basic assumption and terminologies.} Without specific mention, a graph is an undirected simple graph. We use $n$ to denote the number of vertices and $m$ to denote the number of edges of the input graph. We write $\tO{f}=O(f\cdot\log^cn)$ for some constant $c$. We define $[z]=\{1,2,\dots,z\}$. For a real function $f:X\to\bbR$, we write $f\mid_{X'}$ for some $X'\subseteq X$ as the restriction of $f$ on $X'$. We write $f(X')=\sum_{x\in X'}f(x)$. 

\paragraph{Basic graph terminologies.} 

An (undirected) graph is denoted by $G=(V,E)$ where $V$ is the vertex (or node) set and $E$ is the edge set. $N_G(u)$ is the set of neighbors of $u$ defined by $\{v\in V\mid (u,v)\in E\}$. We omit $G$ in the subscript if $G$ is clear from the context. We also define $N[u]=N(u)\cup\{u\}$. In terms of edges, we say an edge $e$ is \emph{incident} to a vertex $v$ if $v$ is one of endpoints of $e$. The set of incident edges to a vertex $v$ is denoted by the notation $\delta(v) \subseteq E$. Given two disjoint vertex set $A, B \subseteq V$, define $E(A,B)=\{(x,y)\in E\mid x\in A,y\in B\}$. 
A \emph{tree} $T\subseteq E$ is a connected acyclic subgraph of $G$. We use $V(T)$ to denote the vertex set of $T$. This notation can be extended: for a subgraph $E'\subset E$, we use $V(E')$ to denote the vertices in $V$ that are adjacent to at least one edge in $E'$. For a set $S$, we denote the \emph{induced subgraph} of $G$ from $S$ by $G[S]$, where $V(G[S]) = S$ and $E(G[S])$ is the set of edges which have both endpoints in $S$.

\paragraph{Degree and volume.} The \emph{degree} of a vertex $u$ is defined as $\deg_G(v)=\abs{N(u)}$. The \emph{minimum degree} of $G$ is defined as $\min_{v\in V}\deg_G(v)$. For a vertex set $S$, we define the \emph{volume} of $S$ as $\deg_G(S)=\sum_{v\in S}\deg_G(v)$.

\paragraph{Vertex lengths and shortest paths.} A path $P$ in $G$ is a sequence of vertices $(v_1,...,v_k)$ such that $(v_i,v_{i+1})\in E$ for every $i\in[k-1]$. A simple path is a path that all vertices are distinct. We define $V(P)=\{v_1,...,v_k\}$, $E(P)=\{(v_i,v_{i+1})\mid i\in[k-1]\}$, and $\Start(P)=v_1,\End(P)=v_k$. We define $\Prec_P(v_i)=v_{i-1}$ and $\Succ_P(v_i)=v_{i+1}$. 
For two vertices $v_i,v_j\in V(P)$ with $i\le j$, we define $P[v_i,v_j]$ as the subpath $\{v_i,v_{i+1},...,v_j\}$. We define $P\circ (v_{k+1},v_{k+2},...)$ as the concatenated path $(v_1,...,v_k,v_{k+1},v_{k+2}...)$

For an undirected graph $G=(V,E)$, a \emph{length function} of $G$ is a function $\ell:V\to\mathbb{R}_{\ge0}$. We define the length of the path $P = (v_{1}, \cdots, v_{k})$ of $G$ with respect to $\ell$ as $\ell(P) = \sum_{i=1}^{k} \ell(v_{i})$.
For vertices $s, t \in V$, A path $P$ is an \emph{$(s, t)$-path} if $\Start(P)=s$ and $\End(P)=t$. For two vertex sets $S,T\subseteq V$, a path $P$ is an \emph{$(S,T)$-path} if it is an $(s,t)$-path for some $s\in S,t\in T$. 
Given a vertex length $\ell$, $P$ is an \emph{$(s, t)$-shortest path (with respect to $\ell$)} if $P$ has the smallest vertex length among all $(s, t)$-paths. $P$ is an \emph{$(S, T)$-shortest path (with respect to $\ell$)} if $P$ has the smallest vertex length among all $(S,T)$-paths.
We define $\dist_{G,\ell}(s,t)$ as the length of the $(s,t)$-shortest path, and $\dist_G(S,T)$ the length of the $(S,T)$-shortest path. 
When $G$ or $\ell$ is clear from the contest, we omit $G$ or $\ell$ from the subscript. For a real number $\alpha \ge 1$, $P$ is an \emph{$\alpha$-approximate $(s, t)$-shortest path (with respect to $\ell$)} if $\ell(P)\le \alpha\cdot \dist_{G,\ell}(s,t)$, $P$ is an \emph{$\alpha$-approximate $(S, T)$-shortest path} if $\ell(P)\le \alpha\cdot \dist_{G,\ell}(S,T)$.

\paragraph{Vertex cut.} A \emph{vertex cut} (or simply \emph{cut} in this paper) is a partition of $V$ denoted by $(L,S,R)$ where (1) $L \neq \emptyset, R \neq \emptyset, E(L,R) = \emptyset$  or (2) $\abs{S} \geq n-1$. The size of a vertex cut $(L,S,R)$ is defined as $|S|$. The \emph{minimum} vertex cut refers to the vertex cut with the smallest size. %

We refer to a vertex cut $(L, S, R)$ as a $(u, v)$-vertex cut if $u \in L$ and $v \in R$. For two vertex sets $A,B\subseteq V$, we say $(L,S,R)$ is an $(A,B)$-vertex cut if $A\subseteq L,B\subseteq R$. In this case, we say $S$ is a $(u,v)$-separator or $(A,B)$-separator. $S$ is called a \emph{separator} (or a \emph{vertex cut}) if $S$ is a $(u,v)$-separator for some $u,v\in V$. We use $\kappa_G(u,v)$ to denote the size of the smallest $u,v$-separator in $G$, and $\kappa(G)$ to denote the smallest separator in $G$. When there are no $(u,v)$-separator or no separator in $G$, we define $\kappa_G(u,v)$ or $\kappa(G)$ as $n-1$. In this case, any set of $n - 1$ vertices is considered a separator.

\paragraph{Fractional vertex cut.} Given a graph $G = (V, E)$, a \emph{fractional separator (or a fractional vertex cut)} is a length function $C: V\to \mathbb{R}_{\ge 0}$ such that there exists $s,t\in V$ such that $C(s)=C(t)=0$ and $\dist_{G,C}(s,t)\ge 1$. In this case, we say $C$ is a $(s,t)$-fractional separator (or $(s,t)$-fractional vertex cut). For a set $T \subseteq V$, we say that $C$ is an $(s, T)$-fractional separator if it is an $(s, t)$-fractional separator for every $t \in T$.

We use $\sum_{v\in V}C(v)$ to denote the \emph{size} of the fractional vertex cut $C$.

The following folklore lemma shows the connection between fractional and (integral) vertex cuts.

\begin{restatable}[Proof in \Cref{sec:rounding}]{lemma}{fractionalcut}
    Given a graph $G = (V, E)$ with a positive integer $k > 1$ and two distinct vertices $s, t \in V$, there is an $(s, t)$-vertex cut with size $< k$ if and only if there is a fractional $(s, t)$-vertex cut with size $\le k - 0.5$. \label{lem:fractional_to_integral}
\end{restatable}

\paragraph{\pram{} model.} In \pram{} model, we have a set of processors and a shared memory space. Time is divided into discrete time slots. In each time slot, each processor can independently read and write on the shared memory space or do other unit operations\footnote{There are different variations of \pram{}, differed by the solutions to concurrent write and read. However, the complexity differed by $\polylog(n)$ factors for different models. Since in this paper $\polylog(n)$ is considered small, we do not concern ourselves with the specific model.}.
The input is given initially on the shared memory space, and the processors are required to jointly compute a specific problem given the input. The complexity is measured by \emph{work} and \emph{depth}, where work is measured by the \textbf{total} amount of unit operations performed by all the processes, and depth is measured by the time consumed before the output is generated. 

In this paper, we consider Monte-Carlo randomized algorithms that succeed with high probability (i.e., with probability $1/n^c$ for an arbitrarily large constant $c$).

\paragraph{Single-source shortest path.} Given an undirected graph $G=(V,E)$, a vertex length $\ell$ and a vertex $s$, and spanning tree $T$ of $G$ is a \emph{$s$-source $\alpha$-approximate shortest path tree (on $G$ with respect to $\ell$)} if every $(s,t)$-path on the tree (for some $t\in V$) is an $\alpha$-approximate $(s,t)$-shortest path on $G$ with respect to $\ell$.

The following theorem provides a deterministic \pram{} algorithm for the approximate single-source shortest path problem, which will be used as an important building block in our algorithm. It is based on \cite{RozhonGHZL22}, although \cite{RozhonGHZL22} only considers edge length, we will show in \Cref{sec:misingproofs} that vertex length is not harder.

\begin{restatable}[Proof in \Cref{sec:misingproofs}]{theorem}{sssp}
    There is a deterministic \pram{} algorithm $(T,d)\leftarrow \alg{SSSP}(G,\ell,s,\eps)$ that, given a connected undirected graph $G = (V, E)$, a vertex length function $\ell$, a source vertex $s\in V$ and an approximation factor $\eps$, outputs a $(1+\eps)$-approximate $s$-source shortest path tree and a distance function $d(\cdot)=\dist_T(s,\cdot)$ in $\tO{m/\eps^2}$ work and $\tO{1}$ depth.\label{thm:sssp}

\end{restatable}

\paragraph{Useful Inequalities}
Let \(X_{1}, \dots, X_{n}\) be independent Bernoulli random variables, and define their sum as \(X = X_{1} + \dots + X_{n}\) with expectation \(\mu = \bbE[X]\). By \emph{Chernoff's inequality}, for any \(\delta > 0\), it holds that
\begin{align}\label{eq:chernoff}
\Pr[X \ge (1+\delta)\mu] \le e^{-\frac{\delta^{2}\mu}{2+\delta}}.
\end{align}
Setting \(\delta = \Theta(\log n)\) ensures that \(X = \tO{\mu}\) with high probability, which is enough for our cases.

%% file: unbalanced.tex
\section{A Framework for Parallel Vertex Connectivity}\label{sec:parallelvertexconnectivity}

In this section, we present and prove our main theorem.

\begin{theorem}\label{thm:main-detail}
    There is a randomized \pram{} algorithm that, given an undirected graph $G=(V,E)$ and an integer $k\ge 1$, outputs a vertex cut $S$ or $\bot$, such that
    \begin{itemize}
        \item \emph{(Correctness)} If $\kappa(G)< k$, the algorithm outputs a vertex cut $S$ of size less than $k$ with high probability. If $\kappa(G)\ge k$, the algorithm always outputs $\bot$.
        \item \emph{(Complexity)} The algorithm runs in $\tO{mk^{12}}$ work $\tO{k^3}$ depth with high probability. 
    \end{itemize}
\end{theorem}

\subsection{Algorithmic Components}

In this section, we will define several algorithmic components of our algorithm for \Cref{thm:main-detail}, which will be proved in later sections.

An important building block of \Cref{thm:main-detail} is the following \emph{sensitivity spanning forest} data structure. The data structure is based on the Batch-Parallel Euler Tour Tree \cite{tseng2019batch} against an oblivious adversary, and we describe only the operations required for our algorithm.
 The data structure is formally proved in \Cref{sec:datastructure}. 

\begin{restatable}[Sensitivity Spanning Forest]{lemma}{decrementalspanningforest}\label{lem:dynamicspanningforest}
There exists a randomized \emph{sensitivity spanning forest} data structure \(\mathcal{D}\) against oblivious adversaries that supports the following operations in the \pram{} model.

\begin{itemize}
    \item \(\alg{Init}(G,\sigma)\): Given an undirected graph \(G=(V,E)\) and a weight function \(\sigma:V\to\mathbb{R}^{+}\) satisfying \(\sigma_v\ge 1\) for every \(v\in V\), this operation takes $\tO{m}$ random bits, and accordingly initializes the data structure using \(\tilde{O}(m)\) work and \(\tO{1}\) depth with high probability.

    \item \(\alg{Fail}(E')\): Given a set of edges \(E'\subseteq E\), which is \emph{independent from the choice of random bits of \(\alg{Init}\)}, this operation marks the edges in \(E'\) as \emph{failed}, updating the graph to \(\tilde{G}=G-E'\), where \(G\) is the graph provided in \(\alg{Init}\). It initializes a spanning forest \(F_{\cD}\) of $\tG$ with a set of identifiers \(\cI_{\cD}\), assigning each connected component \(K\) in \(F_{\cD}\) a unique identifier \(\id(K) \in \cI_{\cD}\). For each identifier \(\chi \in \cI_{\cD}\), denote by \(\cD(\id)\) the set of vertices in the connected component corresponding to \(\chi\). The operation satisfies:
    \begin{itemize}
        \item \emph{(Correct against oblivious adversary.)} If $E'$ is independent of the random bits used in $\alg{Init}$, with high probability, \(F_{\cD}\) is a maximal spanning forest of \(\tilde{G}\). i.e. connected components of $F_{\cD}$ is also a connected component of $\tG$. In particular, the correctness is always guaranteed if $E' = \emptyset$.
        \item \emph{(Complexity)} The operation runs in \(\tO{\abs{E'}}\) work and \(\tO{1}\) depth with high probability.
    \end{itemize}
\end{itemize}

After \(\alg{Fail}(E')\) is executed, the data structure \(\cD\) supports the following queries. Assuming the correctness of \(\alg{Fail}(E')\), these queries always return correct results.

\begin{itemize}
    \item \(\alg{ID}(U)\): Given a subset \(U\subseteq V\) of vertices, this query returns a list of identifiers \((\id_u)_{u\in U}\), where each \(\id_u\) corresponds to the connected component in \(F_{\cD}\) containing vertex \(u\). The query runs in \(\tO{\abs{U}}\) work and \(\tO{1}\) depth with high probability.

    \item \(\alg{Sum}(I)\): Given a set of identifiers \(I\subseteq \mathcal{I}_{\mathcal{D}}\), this query returns a list \((\bar{\sigma}_{\id})_{\id\in I}\), where \(\bar{\sigma}_{\id}=\sigma(\mathcal{D}(\id))\). The query runs in \(\tO{\abs{I}}\) work and \(\tO{1}\) depth with high probability.

    \item \(\alg{Components}(I)\): Given a set of identifiers \(I\subseteq \mathcal{I}_{\mathcal{D}}\), this query returns the connected components \((\mathcal{D}(\id))_{\id\in I}\). The query runs in \(\tO{\abs{I}}\) work and \(\tO{1}\) depth with high probability.

    \item \(\alg{Tree}(\id,x,q)\): Given an identifier \(\id\in \mathcal{I}_{\mathcal{D}}\), a vertex \(x\in \mathcal{D}(\id)\), and a parameter \(q\) satisfying \(\sigma(\mathcal{D}(\id))\ge 2q\), this query returns a tree \(T\subseteq \tilde{E}\) with \(V(T)\subseteq \mathcal{D}(\id)\) and \(x\in V(T)\). The operation runs in \(\tO{q}\) work and \(\tO{1}\) depth with high probability. Moreover, the returned tree \(T\) satisfies one of the following conditions:
    \begin{enumerate}
        \item \(q\le \sigma(V(T))\le 2q\), or
        \item \(\sigma(V(T))<q\) and there exists a vertex \(v\in \mathcal{D}(\id)\), adjacent to \(V(T)\) in \(\tilde{E}\), satisfying \(\sigma(v)>q\).
    \end{enumerate}
\end{itemize}
\end{restatable}

Another part of our algorithm is $\alg{LocalCuts}$. We first define local cuts as follows.

\begin{definition}[Local cuts]\label{def:localcuts}
    Given an undirected graph $G=(V,E)$, a vertex $u\in V$ and an integer $\mu\ge 1$, a(n) (integral) vertex cut $(L,S,R)$ is called a \emph{$(u,\mu)$-local (vertex) cut}, if $u\in L$ and $\mu\le \deg_G(L)\le 2\mu$. 
\end{definition}

Intuitively, $\alg{LocalCuts}$ takes as input a vertex $x$ and parameter $\mu$, and returns an $(x,\mu)$-local vertex cut $(L,S,R)$ such that $L\subseteq V_{\inner}$. Each vertex in $V_{\inner}$ needs to have a low degree to guarantee a small complexity.

We define a \emph{non-zero representation} of a function $C:V\to\bbR_{\ge 0}$ as only storing the non-zero entries of $C$. In this way, it is possible for an algorithm to output a cut in sublinear work.

\begin{restatable}[Parallel Local Cuts]{lemma}{localcuts}\label{lem:localvertexcut}
There exists a randomized \pram{} algorithm \[(C, \tau) \gets \alg{LocalCuts}(G,k,x,\mu,V_{\inner},(\cD^{(i)})_{i \in [r]})\] given an undirected graph \(G\), integers \(k,\mu \ge 1\), a vertex \(x\in V\), a vertex set \(V_{\inner} \subseteq V\), and $r = \lceil 400k^{3} \ln n\rceil$ independent sensitivity spanning forest data structures \(\cD^{(1)}, \cdots, \cD^{(r)}\), returning $\bot$ or a non-zero representation of a function \(C:V\to\bbR_{\ge 0}\) with a vertex $\tau \in V$. The algorithm has the following guarantees.
\begin{enumerate}
    \item \emph{(Guarantees for the inputs)} Every vertex in \(V_{\inner}\) has degree at most \(5\mu\), \(N_G[V_{\inner}]\neq V\), the subgraph \(G[V_{\inner}]\) is connected, and each \(\cD^{(i)}\) is independently initialized on the graph \(G[V_{\inner}]\) by calling \(\cD^{(i)}.\alg{Init}(G[V_{\inner}],\deg_G)\).
    \item \emph{(Correctness)} If there is an \((x,\mu)\)-local cut \((L,S,R)\) of size less than \(k\) with \(L\subseteq V_{\inner}\), the algorithm returns \((C, \tau)\) with probability at least \(0.01\). Moreover, $C$ is a fractional $(x, \tau)$ cut of size at most $k - 0.5$ with high probability.
    \item \emph{(Complexity)} The algorithm runs in \(\tO{k^{12}\mu}\) work and \(\tO{k^3}\) depth with high probability. In particular, \(C\) has $\tO{k^7\mu}$ nonzero entries with high probability when it's returned.
\end{enumerate}
\end{restatable}

\Cref{lem:localvertexcut} is the most technical part of our paper, and we will prove it in the next section. 

Our algorithm will find a fractional cut. We will use the following lemma to turn a fractional cut into an integral cut. We defer the proof to \Cref{sec:rounding}.

\begin{restatable}[Integral $(s, t)$-cut]{lemma}{rounding}\label{lem:rounding}
    There exists a randomized \pram{} algorithm \[\alg{IntegralSTCut}(G, k, s, t)\]
    Given an undirected graph $G = (V, E)$, a positive integer $k$, two distinct vertices $s, t \in V$, returning a $(s, t)$-vertex cut $S$ of size at most $k$ or $\bot$, such that
    \begin{itemize}
        \item \emph{(Correctness.)} If $\kappa(s, t) < k$, the algorithm returns $S$ with high probability. If $\kappa(s, t) \ge k$, the algorithm always returns $\bot$.
        \item \emph{(Complexity.)} The algorithm runs in $\tO{mk^5}$ work and $\tO{k^3}$ depth.
    \end{itemize}
\end{restatable}

\subsection{Main Algorithm and Analysis (Proof of \Cref{thm:main-detail})}\label{subsec:main-detail}

\paragraph{Algorithm.} We now present the algorithm corresponding to \Cref{thm:main-detail}. If $n \le 100k^{2}$, run the algorithm of \cite{CheriyanKT93} and terminate. For the remainder of the algorithm, we assume $n > 100k^{2}$. Repeat the following procedure \(\Theta(\log n)\) times, each consisting of \(\lfloor \log m\rfloor + 1\) iterations. For every iteration \(i = 0, \dots, \lfloor \log m \rfloor\), set \(\mu = 2^i\) and perform the following five steps.

\begin{description}
    \item[Step 1.] Sample a vertex \(t\) from \(V\), where each vertex \(x\) is chosen with probability \(\deg_{G}(x)/2m\).

    \item[Step 2.] Compute the vertex sets:
        \[V_{\high}:=\{v\in V-N[t]\mid \deg_G(v)> 5\mu\}\]
        \[V_{\low}:=\{v\in V-N[t]\mid \deg_G(v)\le 5\mu\}\]
    
    \item[Step 3.] For each connected component \(K\) in \(G[V_{\low}]\), initialize $r$ independent instances of the sensitivity spanning forest data structure \((\cD_{K}^{(i)})_{i\in[r]}\) from \Cref{lem:dynamicspanningforest} by calling \[\cD_K^{(i)}.\alg{Init}(G[K],\deg_G|_{K})\]
    for each $i = 1, \cdots, r$ in parallel.

    \item[Step 4.] For each vertex \(v \in V_{\low}\) satisfying \(\deg_{G}(v) \le 2\mu\), sample \(v\) independently with probability \(\deg_{G}(v)/2\mu\). Denote the sampled vertex set by \(X\).

    \item[Step 5.] For each vertex \(x\in X\), let \(K_x\) denote the connected component of \(G[V_{\low}]\) containing \(x\). Execute the following (using \Cref{lem:localvertexcut}):
    \[(C_{x}, \tau_{x}) \leftarrow \alg{LocalCuts}(G,k,x,\mu,K_x,(\cD_{K_x}^{(i)})_{i \in [r]})\]
\end{description}

\begin{remark}
When running \(\alg{LocalCuts}(G,k,x,\mu,K_x,(\cD_{K_x}^{(i)})_{i \in [r]})\) in parallel for different \(x\in X\), it might require concurrent modifications to the initialized data structure \(\cD_{K_x}^{(i)}\) (stored in memory \(\cM_{\init}\)). To avoid conflicts, each algorithm reads \(\cM_{\init}\) normally, allowing concurrent reads, and when a write is necessary, it creates a separate space for modified entries. Subsequent reads of modified entries are redirected to this new space, ensuring independent parallel execution without memory conflicts.
\end{remark}

\paragraph{Finding an integral cut.} If all of the $\alg{LocalCuts}(G,k,x,\mu,K_x,(\cD_{K_x}^{(i)})_{i \in [r]})$ return $\bot$, the algorithm returns $\bot$. Otherwise, select any output pair $(C_{x}, \tau_{x})$, and

\[S\leftarrow \alg{IntegralSTCut}(G, k, x,\tau_{x}).\]

If \(S\) is an integral \((x, \tau_x)\)-vertex cut, return \(S\); otherwise, return \(\bot\).

\medskip

We now prove \Cref{thm:main-detail}, assuming \Cref{lem:dynamicspanningforest,lem:localvertexcut,lem:rounding}.

\paragraph{Correctness.}
Suppose first that $\kappa(G) \ge k$. Then, by \Cref{lem:fractional_to_integral}, there exists no fractional $(s,t)$-cut of size at most $k - 0.5$ for any $s,t \in V$. Consequently, the algorithm always returns $\bot$ by the correctness guarantee of \Cref{lem:rounding}.

Now suppose $\kappa(G) < k$, so there exists a vertex cut $(L, S, R)$ with $|S| < k$. Without loss of generality, assume $\deg_G(L) \le \deg_G(R)$. Since $\deg_G(L) \le m/2$, there exists an integer $0 \le i \le \lfloor \log m \rfloor$ such that $\mu = 2^i$ satisfies $\mu \le \deg_G(L) \le 2\mu$.

We claim that with high probability, the algorithm selects a vertex $x \in L$ such that $(L, S, R)$ forms a valid $(x, \mu)$-local vertex cut.

In Step 1, the vertex $t$ is chosen with probability proportional to degree. Since $\deg_G(R)$ is lower bounded by \[
2m = \deg_{G}(L)+\deg_{G}(R)+\abs{E(S, L)} + \abs{E(S, R)} + \abs{E(G[S])} \le 4\deg_{G}(R) + k^{2}/2.
\] 
The first inequality comes from the fact that $\abs{E(S, L)} \le \deg_{G}(L)$ and $\abs{E(S, R)} \le \deg_{G}(R)$, and again $\deg_{G}(L) \le \deg_{G}(R)$. As $G$ is a simple graph, so is $G[S]$ so $\abs{E(G[S])}$ is at most $(k-1)(k-2)/2$. one deduce that $\deg_{G}(R) \ge m/2-k^{2}/8 \ge m/3$. Hence, the probability that $t \in R$ is at least $1/6$. In this case, $N[t] \subseteq R \cup S$, so $L \cap N[t] = \emptyset$, implying $L \subseteq V_{\low}$.

In Step 4, each vertex $x \in L$ is included in $X$ independently with probability $\deg_G(x)/2\mu$. The probability that no vertex from $L$ is selected is:
\[
\prod_{x \in L}\left(1 - \frac{\deg_G(x)}{2\mu}\right) \le \exp\left(-\frac{\deg_G(L)}{2\mu}\right) \le e^{-1/2} < \frac{2}{3}.
\]
Hence, with probability at least $1 - e^{-1/2} > 1/3$, we have $X \cap L \neq \emptyset$.

Now fix a vertex $x \in X \cap L$. By \Cref{lem:localvertexcut}, the algorithm successfully returns a fractional cut $(C_x, \tau_x)$ with probability at least $0.01$. Therefore, the probability that a single repetition of the algorithm produces a valid fractional cut is at least:
\[
\frac{1}{6} \cdot \frac{1}{3} \cdot 0.01 = \frac{1}{1800}.
\]
Repeating the entire process $1800c \log n$ times boosts the success probability to at least $1 - 1/n^c$ for any constant $c$, by standard amplification.

By the correctness guarantee of \Cref{lem:localvertexcut}, the returned function $C_x$ is a valid fractional $(x, \tau_x)$-cut of size at most $k - 0.5$ with high probability. Then, by \Cref{lem:fractional_to_integral}, it follows that $\kappa(x, \tau_x) < k$, and \Cref{lem:rounding} ensures the algorithm returns an integral $(x, \tau_x)$-vertex cut $S$ of size at most $k - 1$ with high probability.

Finally, observe that the inputs to each call to $\alg{LocalCuts}$ are independent of the random bits used to initialize each $\mathcal{D}^{(i)}$. Hence, the oblivious adversary condition in \Cref{lem:dynamicspanningforest} is satisfied. This will be further detailed in \Cref{rmk:oblivious}.

This concludes the proof of correctness for \Cref{thm:main-detail}.

\paragraph{Complexity.}
If $n \le 100k^{2}$, the algorithm of \cite{CheriyanKT93} runs in $\tO{k^{8}}$ work and $\tO{k^2}$ depth.

Assume $n > 100k^{2}$. We analyze a single iteration within one repetition. Step 1 samples a vertex $t$ in $O(m)$ work and $\tO{1}$ depth. Step 2 computes $V_{\high}$ and $V_{\low}$ also in $O(m)$ work and $\tO{1}$ depth. Identifying connected components in $G[V_{\low}]$ requires $\tO{m}$ work and $O(\log n)$ depth using parallel algorithms for spanning forests (e.g., \cite{awerbuch1987new}).

In Step 3, initializing $r = O(k^3 \log n)$ sensitivity spanning forest structures across all components requires total work $\tO{mk^3}$ and $\tO{1}$ depth by \Cref{lem:dynamicspanningforest}.

In Step 4, sampling vertices into $X$ requires $O(m)$ work and $\tO{1}$ depth. The expected size of $X$ is $O(m / \mu)$, and by Chernoff's inequality (\Cref{eq:chernoff}), its size is $\tO{m/\mu}$ with high probability.

In Step 5, each call to $\alg{LocalCuts}$ runs in $\tO{k^{12} \mu}$ work and $\tO{k^3}$ depth by \Cref{lem:localvertexcut}. Across all sampled $x \in X$, total work is:
\[
\tO{\abs{X} \cdot k^{12} \mu} = \tO{mk^{12}}, \quad \text{depth } \tO{k^3}.
\]

Since there are $O(\log m) = O(\log n)$ iterations per repetition, and $O(\log n)$ repetitions, the overall complexity is $\tO{mk^{12}}$ work and $\tO{k^3}$ depth with high probability.

Finally, converting a fractional cut into an integral one using \Cref{lem:rounding} takes $\tO{mk^5}$ work and $\tO{k^3}$ depth.

This completes the complexity analysis for \Cref{thm:main-detail}.

%% file: localcuts.tex
\section{Parallel Local Cuts (Proof of \Cref{lem:localvertexcut})}\label{sec:localcuts}

In this section, we will prove \Cref{lem:localvertexcut}, restated as follows.

\localcuts*

\subsection{Algorithm Description}\label{subsec:algorithmlocalcuts}

Before giving the algorithm for \Cref{lem:localvertexcut}, let us first define some necessary notations. For convenience, we write

\[H=G[V_{\inner}]\]

as a primary working graph in this section. The algorithm consists of maintaining \emph{weight functions} and \emph{trivial vertices} defined as follows.

\begin{definition}[Weight Functions \& Trivial vertices]
    A weight function on $G = (V, E)$ is a function $w : V \to \mathbb{R}_{\ge 0}$. A vertex $v \in V$ is called \emph{trivial (with respect to $w$)} if $w(v) = 1$; otherwise, it is \emph{non-trivial}.
\end{definition}

The algorithm maintains only the weights of non-trivial vertices. This compact representation is referred to as the \emph{non-trivial representation} of $w$.

\paragraph{Subroutine \(\alg{Contract}\).}
This subroutine abstracts away trivial vertices by contracting them into connected components, and isolates the behavior of non-trivial vertices which encode meaningful weight information.

Given a graph \(G = (V, E)\), a designated subset of vertices \(V_{\inner} \subseteq V\), and a sensitivity spanning forest data structure \(\cD\) (satisfying the input conditions of \Cref{lem:localvertexcut}), along with a weight function \(w: V \to \bbR_{\ge 0}\) given in its non-trivial representation, the subroutine
\[
(G', \ell) \leftarrow \alg{Contract}(G, V_{\inner}, \cD, w)
\]
produces a contracted graph \(G' = (V', E')\) and a vertex length function \(\ell: V' \to \bbR_{\ge 0}\), constructed as follows.

\begin{enumerate}
\item Let \(H = G[V_{\inner}]\) be the subgraph induced by \(V_{\inner}\). Define the set of vertices and edges adjacent to non-trivial vertices.
    \[
    V_{\ne1} := \bigcup_{w(u) \ne 1} N_H[u], \qquad
    E_{\ne1} := \bigcup_{w(u) \ne 1} \delta_H(u).
    \]
    This collects all neighbors and incident edges of non-trivial vertices. The subgraph formed by \(H - E_{\ne1}\) is composed entirely of trivial vertices.
    
    \item Apply the sensitivity data structure update
    \[
    \cD.\alg{Fail}(E_{\ne1})
    \]
    to remove the influence of these non-trivial edges, and retrieve the connected component identifiers for all relevant vertices:
    \[
    (\id_u)_{u \in V_{\ne1}} \leftarrow \cD.\alg{ID}(V_{\ne1}).
    \]

\item Define the contracted graph \(G' = (V', E')\) as:
    \[
    V' := \{\id_u \mid u \in V_{\not=1} \}, \qquad
    E' := \{ (\id_u, \id_v) \mid (u,v) \in E_{\not=1} \}.
    \]
    Here, we treat \(E'\) as a simple set of unordered pairs (i.e., without self-loops or parallel edges).

\item Finally, define the vertex length function \(\ell: V' \to \bbR_{\ge 0}\) by
    \[
    \ell(\id_u) := w(u) \quad \text{for each } u \in V_{\not=1}.
    \]
    This definition is well-defined because any two vertices \(u, v\in V_{\not=1}\) with \(\id_u = \id_v\) must be trivial (i.e., \(w(u) = w(v) = 1\)). Non-trivial weights are thus never collapsed.
    \end{enumerate}

The resulting pair \((G', \ell)\) encodes the non-trivial structure of the graph while aggregating trivial parts into connected blocks, enabling efficient computations such as shortest paths or weight updates that respect the structure of the original graph.

The following definition defines $\CT{P}{G'}$ for a path $P$ on $G$ as a path on $G'$ after contracting trivial vertices.

\begin{definition}\label{def:pathcontract}
    Given a path $P$ in $G$, and a partition of $V(G)$ induced by vertex clusters in $G'$, define $\CT{P}{G'}$ as the path in $G'$ obtained by contracting each maximal subpath of $P$ contained in a single cluster to that cluster.
\end{definition}

The following lemma shows the important properties of $(G',\ell)$. We will prove it in the next section.

\begin{lemma}\label{lem:contract}
Let \(G, V_{\inner}, \cD\) satisfy the input conditions from \Cref{lem:localvertexcut}, and assume the correctness of \(\cD.\alg{Fail}\) as described in \Cref{lem:dynamicspanningforest}. Suppose also that \(V_{\inner}\) contains at least one vertex that is non-trivial with respect to \(w\). Then, after executing \((G'=(V',E'),\ell) \leftarrow \alg{Contract}(G,V_{\inner},\cD,w)\), the following properties hold:

\begin{enumerate}
    \item For each non-trivial vertex \(u \in V_{\inner}\) with respect to \(w\), the component \(\{u\} = \cD(\id_u)\) is a singleton connected component in the graph \(H - E_{\not=1}\). This ensures that \(\ell\) is well-defined.
    \item Every vertex \(\id \in V'\) corresponds exactly to a connected component \(\cD(\id)\) of the graph \(H - E_{\not=1}\). Additionally, an edge \((\id_1, \id_2) \in E'\) indicates the existence of at least one edge between the respective components \(\cD(\id_1)\) and \(\cD(\id_2)\) in the graph \(H\).
    \item The resulting contracted graph \(G'\) is connected.
    \item If $P$ is a simple path in $G$ and $\ell$ is the vertex length function on $G'$, \[
w(P) - n \le \ell(\CT{P}{G'}) \le w(P),
\]
\end{enumerate}
\end{lemma}

\medskip

\paragraph{Subroutine \(\alg{Tree}\).} The subroutine \[T \leftarrow \alg{Tree}(K, \tilde{\sigma}, v, q)\]

isolates the $\alg{Tree}$ function of the sensitivity spanning forest for convenience.

Let \(K\) be a connected graph, and let \(\tilde{\sigma} : V(K) \to \mathbb{R}_{\ge 1}\) be a vertex weight function such that \(\tilde{\sigma}(V(K)) \ge 2q\). The subroutine proceeds as follows:

\begin{enumerate}
    \item Initialize a sensitivity spanning forest $\tilde{\cD}$ by
    \[
    \tilde{\cD}.\alg{Init}(K, \tilde{\sigma}).
    \]
    \item Select an arbitrary vertex \(u \in V(K)\), and retrieve the component identifier
    \[
    \id \leftarrow \tilde{\cD}.\alg{ID}(u).
    \]
    \item Call the failure update without deletions and then query the tree on $\tilde{\cD}$.
    \[
    \tilde{\cD}.\alg{Fail}(\emptyset), \quad T \leftarrow \tilde{\cD}.\alg{Tree}(\id, v, q).
    \]
\end{enumerate}

The tree \(T\) is returned. The correctness of this procedure follows from \Cref{lem:dynamicspanningforest}, which ensures that the vacuous \(\tilde{\cD}.\alg{Fail}(\emptyset)\) preserves the spanning tree of \(K\).

\begin{observation}\label{obs:tree}
Suppose \(K\) is a connected graph and \(\tilde{\sigma}\) is a vertex weight function such that \(\tilde{\sigma}(v) \ge 1\) for every \(v \in V(K)\), and \(\tilde{\sigma}(V(K)) \ge 2q\). Then the output \(T \leftarrow \alg{Tree}(K, \tilde{\sigma}, v, q)\) satisfies:
\begin{itemize}
    \item \(V(T) \subseteq V(K)\), and \(v \in V(T)\);
    \item moreover, one of the following holds:
    \begin{enumerate}
        \item \(q \le \tilde{\sigma}(V(T)) \le 2q\);
        \item \(\tilde{\sigma}(V(T)) < q\), and there exists a vertex \(u \in \tilde{\cD}(\id)\) with \(\tilde{\sigma}(u) > q\), such that \(u\) is adjacent to some vertex in \(V(T)\) in \(K\).
    \end{enumerate}
\end{itemize}
\end{observation}

\medskip

\paragraph{Algorithm $\alg{LocalCut}$.} We now describe the algorithm \(\texttt{LocalCuts}(G,k,x,\mu,V_{\inner},(\cD^{(i)})_{i \in [r]})\). Throughout the description, we assume that all calls to \(\alg{Fail}\) on the sensitivity spanning forest structures \(\cD^{(i)}\) are correct. We refer to this assumption as \emph{all the SSF operations are correct}. This assumption holds with high probability, from \Cref{lem:dynamicspanningforest}. If the SSF operations fail, the algorithm may have undefined behavior which is acceptable in our probabilistic setting.

\begin{description}
    \item[Initialization.] First, we check whether the following inequality holds:
    \begin{equation} \label{eq:deg_condition} 
        \deg_G(N_H(x)) > 2\mu + (k-1)\cdot (5\mu)
    \end{equation}
    If it does, the algorithm returns \(\bot\) immediately and halts.

    Otherwise, initialize the weight function \(w^{(1)}: V \rightarrow \mathbb{R}_{\ge 0}\) as:
    \[
        w^{(1)}(u) =
        \begin{cases}
            0 & \text{if } u = x, \\
            n^3 & \text{if } u \in N_{G}(x), \\
            1 & \text{otherwise}.
        \end{cases}
    \]

    This initialization ensures that only \(N_{G}[x]\) nontrivial weights.

    Set the precision parameter $\eps$ as
    \[
        \eps = \frac{1}{10k}.
    \]
\end{description}

The algorithm proceeds for
\[
    r = \left\lceil \frac{40k^2 \ln n}{\eps} \right\rceil = \left\lceil 400k^3 \ln n \right\rceil
\]
iterations. Note that \(r\) also determines the number of sensitivity spanning forest structures \(\cD^{(i)}\), one for each iteration. Each \(\cD^{(i)}\) is exclusively used during the iteration \(i\).

We now describe the steps of the algorithm for the \(i\)-th iteration, where \(i \in [r]\). Each iteration consists of five steps. If the algorithm does not return a fractional vertex cut in any of the \(r\) iterations, it returns \(\bot\) at the end.

\begin{description}
\item[Step 1 (SSSP on a contracted graph).]     
Invoke the contraction subroutine:
\[
\left(G^{(i)} = (V^{(i)}, E^{(i)}), \ell^{(i)}\right) \leftarrow \alg{Contract}(G, V_{\inner}, \cD^{(i)}, w^{(i)}).
\]
Let \(\id_x\) denote the identifier of the connected component containing \(x\). i.e.
\[
\id_x \leftarrow \cD^{(i)}.\alg{ID}(\{x\}).
\]
By \Cref{lem:contract}, the contracted graph \(G^{(i)}\) is connected. Thus, we can compute a \((1 + \eps)\)-approximate shortest path tree rooted at \(\id_x\) using:
\[
(T^{(i)}, \tilde{d}^{(i)}) \leftarrow \alg{SSSP}(G^{(i)}, \ell^{(i)}, \id_x, \eps).
\]

\item[Step 2 (Binary search for a small subgraph).]
Compute the degree-sum of each contracted vertex. Note that $V^{(i)} \subseteq \cI_{\cD^{(i)}}$ is a set of IDs.
\[
(\sigma^{\deg}(\id))_{\id \in V^{(i)}} \leftarrow \cD^{(i)}.\alg{Sum}(V^{(i)}).
\]
Recall that \(\cD^{(i)}\) was initialized with \(\sigma(v) = \deg_G(v)\) (see \Cref{lem:localvertexcut}), so assuming all SSF operations are correct, we have:
\[
\sigma^{\deg}(\id) = \deg_G(\cD^{(i)}(\id)).
\]

Define, for any distance threshold \(d \in \bbR_{\ge 0}\),
\[
V^{(i)}_{\le d} := \left\{ \id \in V^{(i)} \mid \tilde{d}^{(i)}(\id) \le d \right\}, \quad
V^{(i)}_{< d} := \left\{ \id \in V^{(i)} \mid \tilde{d}^{(i)}(\id) < d \right\},
\]
and their corresponding degree sums:
\[
\deg^{(i)}_{\le}(d) := \sum_{\id \in V^{(i)}_{\le d}} \sigma^{\deg}(\id), \quad
\deg^{(i)}_{<}(d) := \sum_{\id \in V^{(i)}_{< d}} \sigma^{\deg}(\id).
\]

Now define the threshold distance \(\dtres \in \bbR \cup \{+\infty\}\) as the largest value such that

    \[\deg^{(i)}_{<}(\dtres)\le 10{k r \mu}.\]

    The definition of $\dtres$ implies that either $\dtres=+\infty$, or there are some vertices in $V^{(i)}$ with distance exactly $\dtres$ so that $\deg^{(i)}_{\le}(\dtres)>10{k r \mu}$. 

    To find $\dtres$, we first sort vertices in $V^{(i)}$ according to $\tilde{d}$, then we binary search for $\dtres$ given that we can calculate $\deg^{(i)}_{<}(d)$ for any given $d$ by summing up $\sigma^{\deg}$.
    
    \item[Step 3 (find the shortest path going out).] Now switch to the original graph $G$. Compute 
    \[V^{(i)}_{H,<\dtres}\leftarrow \bigcup \cD^{(i)}.\alg{Components}\left(V^{(i)}_{<\dtres}\right)\]

    The following lemma shows the important property of $V^{(i)}_{H,<\dtres}$. We defer the proof to the next section.
    \begin{lemma}\label{lem:ViHdtress}
        Suppose all the SSF operations are correct, then $\deg_G(V^{(i)}_{H,<\dtres})\le 10{k r \mu}$ and $G[V^{(i)}_{H,<\dtres}]$ is connected.
    \end{lemma}
    \textit{On the graph $G$}, find

    \[V^{(i)}_{\outer}:=N_G(V^{(i)}_{H,<\dtres})\cap (V-V_{\inner})\]

    We define the \emph{augmented graph} $G^{(i)}_{\aug}=(V^{(i)}_{\aug},E^{(i)}_{\aug})$ as

    \[V^{(i)}_{\aug}:=V^{(i)}_{H,<\dtres}\cup V^{(i)}_{\outer}\]
    \[E^{(i)}_{\aug}:=\{(u,v)\in E\mid u\in V^{(i)}_{H,<\dtres},v\in V^{(i)}_\aug\}\]

    Or equivalently,
    \[E^{(i)}_{\aug}:=\bigcup_{v\in V^{(i)}_{H,<\dtres}}\delta_G(v)\]
    \[V^{(i)}_{\aug}:=V(E^{(i)}_{\aug})\]
    According to \Cref{lem:ViHdtress}, $G^{(i)}_\aug$ is connected. So we can run \Cref{thm:sssp} on the graph $G^{(i)}_\aug$ with the length function as the weight function $w^{(i)}$ to get an approximate distance. 

    \[(T^{(i)}_{\aug},\tilde{d}^{(i)}_\aug)\leftarrow \alg{SSSP}(G^{(i)}_\aug,x,w^{(i)}|_{V^{(i)}_{\aug}},\eps)\]

    Let $v^{(i)}_*$ be the vertex in $V^{(i)}_\outer$ with the smallest $d^{(i)}_\aug(v^{(i)}_*)$. 
    
    If $d^{(i)}_\aug(v^{(i)}_*)\le \dtres$, let $P^{(i)}$ be the path on $T^{(i)}_\aug$ from $x$ to $v^{(i)}_*$ and $\tilde{w}^{(i)}=w^{(i)}$. Otherwise, we will define $\tilde{w}^{(i)}$ and $P^{(i)}$ in the next step.

    \item[Step 4 (find the shortest inner path).] Suppose $d^{(i)}_\aug(v^{(i)}_*)> \dtres$. It follows that $\dtres\not=+\infty$. 
    
    Remember that by the definition of $\dtres$, we have
    \begin{equation}\label{eq1}
        \deg^{(i)}_{\le}(\dtres)=\sum_{\id\in V^{(i)}_{\le \dtres}}\sigma^{\deg}(\id)>10{k r \mu}
    \end{equation}
    
    The following lemma shows that $T^{(i)}[V^{(i)}_{\le\dtres}]$ is connected. We defer the proof to the next section.
    \begin{lemma}\label{lem:Tiledtressconnected}
        Suppose all the SSF operations are correct, then $T^{(i)}[V^{(i)}_{\le\dtres}]$ is connected.
    \end{lemma}
    According to \Cref{lem:Tiledtressconnected,eq1}, we can call (see \Cref{obs:tree})
    
    \begin{equation} \label{eq:hatT}\hat{T}^{(i)}\leftarrow \alg{Tree}(T^{(i)}[V^{(i)}_{\le\dtres}],\sigma^{\deg},\id_x,5{k r \mu})\end{equation}
    
    According to \Cref{obs:tree}, there are two cases. We will define the vertex subset $V^{(i)}_{H,\le\dtres}\subseteq V(H)$ differently in two cases. We can distinguish the two cases by computing $\sigma^{\deg}(V(\hat{T}^{(i)}))$.

    \begin{description}
        \item[Case 1.] Suppose $5{k r \mu}\le \sigma^{\deg}(V(\hat{T}^{(i)}))\le 10{k r \mu}$. We let
        \begin{equation} \label{eq:step4_c1}V^{(i)}_{H,\le\dtres}\leftarrow\left(\bigcup\cD^{(i)}.\alg{Components}(V(\hat{T}^{(i)}))\right)\cup V^{(i)}_{H,<\dtres}\end{equation}
        \item[Case 2.] Suppose $\sigma^{\deg}(V(\hat{T}^{(i)}))< 5{k r \mu}$ and there is a vertex $\id^{(i)}_{\big}\in V^{(i)}_{\le\dtres}$ with $\sigma^{\deg}(\id^{(i)}_{\big})>5{k r \mu}$ such that $\id^{(i)}_{\big}$ is adjacent to $\hat{T}^{(i)}$ in $T^{(i)}$. According to \Cref{lem:contract} (2), an edge connecting $\id^{(i)}_{\big}$ to $\hat{T}^{(i)}$ corresponds to an edge $(a^{(i)}_{\big},b^{(i)}_{\big})$ in the original graph $G$ such that $b^{(i)}_{\big}\in \cD(\id^{(i)}_{\big})$.
        
        We let
        \begin{equation}\label{eq:step4_c2}
            T^{(i)}_{\big}\leftarrow \cD^{(i)}.\alg{Tree}(\id^{(i)}_{\big},b^{(i)}_{\big},2{k r \mu}),
        \end{equation}

        and
        \begin{equation}\label{eq:step4_c2_s2}V^{(i)}_{H,\le\dtres}\leftarrow \left(\bigcup\cD^{(i)}.\alg{Components}(V(\hat{T}^{(i)}))\right)\cup V(T^{(i)}_{\big})\cup V^{(i)}_{H,<\dtres}.\end{equation}
        
    \end{description}

    After getting $V^{(i)}_{H,\le\dtres}$, we define the graph $G^{(i)}_{H,\le\dtres}$ as the subgraph of $H$ with edge set
    \[E^{(i)}_{H,\le\dtres}\leftarrow \bigcup_{v\in V^{(i)}_{H,\le\dtres}}\delta_G(v)\]

    The following lemma shows the important property of $G^{(i)}_{H,\le\dtres}$.
    \begin{lemma}\label{lem:ViHledtres}
        Suppose all the SSF operations are correct, then $G^{(i)}_{H,\le\dtres}$ is connected and we have
        
        \[2{k r \mu}\le \deg_G\left(V^{(i)}_{H,\le\dtres}\right)\le 20{k r \mu}.\]
    \end{lemma}

    Uniformly at random sample an edge from $E^{(i)}_{H,\le\dtres}$ and choose an arbitrary one of its endpoints in $V^{(i)}_{H,\le\dtres}$, we get a vertex $v^{(i)}_{H}$.

    Define a weight function $\tilde{w}^{(i)}$ restriction on $V(G^{(i)}_{H,\le\dtres})$ with the weight of $v^{(i)}_{H}$ set to be zero

    \[
        \tilde{w}^{(i)}(u) :=
        \begin{cases}
        w^{(i)}(u) & \text{if } u\in V(G^{(i)}_{H,\le\dtres})-\{v^{(i)}_{H}\} \\
        0 & \text{if }u=\{v^{(i)}_{H}\} \\
        \end{cases}
    \]
    
    According to \Cref{lem:ViHledtres}, we can run \Cref{thm:sssp} on the graph $G^{(i)}_{H,\le\dtres}$ with length function as the weight function $\tilde{w}^{(i)}$ to get an approximate distance.

    \[(T^{(i)}_{H,\le\dtres},\tilde{d}_{H,\le\dtres})\leftarrow \alg{SSSP}(G^{(i)}_{H,\le\dtres},x,\tilde{w}^{(i)},\eps)\]

    Let $P^{(i)}$ be the path on $T^{(i)}_{H,\le\dtres}$ from $x$ to $v^{(i)}_{H}$ excluding the last vertex $v^{(i)}_{H}$.

    \item[Step 5 (weight updates).] Let

    \[\tW^{(i)}=\sum_{\substack{v\in V\\ \tilde{w}^{(i)}(v)\not=1}}\tilde{w}(v)\]

    be the summation of the weights of all non-trivial vertices. 
    If
    \[\frac{\tilde{w}^{(i)}(P^{(i)})}{\tW^{(i)}}\ge \frac{1}{k-0.6}\]

    \textbf{return} all the non-zero values of
    \[
        C(u) :=
        \begin{cases}
        0 & \text{if } \tilde{w}^{(i)}(u)\le 1\\
        \frac{\tilde{w}^{(i)}(u)}{\tW^{(i)}}\cdot (k-0.6)\cdot (1+5\eps) & \text{if } \tilde{w}^{(i)}(u)>1 \\
        \end{cases}
    \]
    as the fractional vertex cut with a vertex $\tau$, which is chosen differently upon the cases. $\tau$ is an arbitrary vertex in $V-N_{G}[V_{\inner}]$ if $P^{(i)}$ is from the Step 3, or $\tau = v_{H}^{(i)}$ if $P^{(i)}$ is from Step 4. Notice that we only need a non-trivial representation of $\tilde{w}^{(i)}$ to output non-zero entries of $C$.

    Otherwise, update the weights for every $v\in V(P^{(i)})$ by

    \[w^{(i+1)}(v)\leftarrow (1+\eps)\cdot w^{(i)}(v)\]

\end{description}

\begin{remark} \label{rmk:oblivious}
    Note that the sensitivity data structure \(\cD^{(i)}\) is only accessed during iteration \(i\). In particular, the update operation \(\cD^{(i)}.\alg{Fail}\) is applied to an edge set that is independent of the random bits used in the initialization of \(\cD^{(i)}.\alg{Init}\). This independence guarantees the \emph{with high probability} correctness of SSF operations, as required for analysis against oblivious adversaries.
\end{remark}

\subsection{Missing Proofs}\label{subsec:localcutmissingproofs}
In this section, we will provide the missing proofs of the last section. We first show a basic fact about the algorithm.
\begin{claim}\label{cla:basicw}
    Throughout the algorithm, for every iteration, $w^{(i)}(x)=\ell^{(i)}(\id_x)=0$ for every iteration. For every vertex $v\in V-N_G[V_\inner]$, we have $w^{(i)}(v)=1$. $w^{(i)}$ is non-decreasing as $i$ increases, this implies that any path containing at least $2$ vertices starting from $x$ has length at least $n^3$ with respect to $w^{(i)}$. 
\end{claim}
\begin{proof}
    We have $w^{(1)}(x)=0$ in the first iteration, and every further iteration only applies a multiplicative factor to the weight of $x$. Moreover, since $x$ is a non-trivial vertex, according to \Cref{lem:contract}, $\{x\}$ is always a singleton connected component in the data structure $\cD^{(i)}$. So $\ell^{(i)}(\id_x)=w^{(i)}(x)=0$.

    For every vertex $v\in V-N_G[V_\inner]$, $v$ can never be included in $V(P^{(i)})$ according to the definition of $P^{(i)}$ in Step 3 and Step 4. Since we only update the weight of the vertices in $V(P^{(i)})$ in Step 5, we have $w^{(i)}(v)=w^{(1)}(v)=1$.

    It is easy to see that $w^{(i)}$ is non-decrease since the multiplicative factor $(1+\eps)\ge 1$. Moreover, since initially we have $w^{(1)}(v)=n^3$ for every $v\in N_G(x)$, every path that contain at least $2$ vertices starting from $x$ has length at least $n^3$ with respect to $w^{(i)}$. 
\end{proof}

\begin{proof}[Proof of \Cref{lem:contract}]
    We first show that $V'$ contains all identifiers of connected components of $H-E_{\not=1}$. According to the input guarantee of \Cref{lem:localvertexcut}, $G[V_{\inner}]=H$ is connected. Thus, if $K$ is a connected component of $H-E_{\not=1}$, either $K=V_{\inner}$, in which case $K$ has identifier $\id_u$ for an arbitrary $u\in V_{\inner}$ (notice that $V_{\not=1}$ is non-empty since $V_{\inner}$ contains at least one non-trivial vertex with respect to $w$); or $K\subset V_{\inner}$, in which case $K$ is adjacent to an edge in $E_{\not=1}$ on $G$, which implies that there is a vertex $u\in K\cap V_{\not=1}$, so $\id_u$ is the identifier of $K$ as a vertex in $V'$. 
    
    Notice that each edge $(\id_u,\id_v)$ of $G'$ represent an edge $(u,v)\in E_{\not=1}$. According to the definition of $E_{\not=1}$, we can assume $u$ is a non-trivial node. Thus, $\{u\}$ becomes a single connected component in $H-E_{\not=1}$ since all adjacent edges of $u$ are in $E_{\not=1}$. So $u,v$ are in different connected components. Moreover, edge connected component is adjacent to at least one edge in $E_{\not=1}$, so edges in $G'$ represent edges connecting different connected components of $H-E_{\not=1}$.
    
    Then we show that for every path $P$ on $H$, $\CT{P}{G'}$ is a path on $G'$. Let $(u,v)\in E(P)$, if $u,v$ are in the same connected component of $H-E_{\not=1}$, then $\id_u=\id_v$ and they corresponds to the same vertex $\id_u$ on $\CT{P}{G}$. If $u,v$ are in different connected components, then according to the definition of $E'$, $(\id_u,\id_v)$ is an edge in $G'$. Also, each vertices in $V(G')$ appears at most once from the given condition. Thus, $\CT{P}{G'}$ is a path on $G'$. 

    According to the input guarantees in \Cref{lem:localvertexcut}, we have that $G[V_{\inner}]=H$ is connected. This there is a $(s,t)$-path in $H$ for every $s,t\in V_{\inner}$. Since $\CT{P}{G'}$ is a path on $G'$ for any path $P$ on $H$ and every vertex in $G'$ corresponds to a connected component in $H-E_{\not=1}$, $G'$ is connected. 

    Moreover, the length has the upper bound
    \[\ell(\CT{P}{G'})=\sum_{\id \in V(\CT{P}{G'})}\ell(\id)\le \sum_{\id \in V(\CT{P}{G'}),u\in\cD(\id)}w(u)\le w(P)\]

    As for the lower bound, notice that if $w(u)>1$ for some $u\in V_{\inner}$, then all edge adjacent to $u$ is added to $E_{\not=1}$, so $\{u\}$ becomes a connected component of $H-E_{<1}$ and is corresponding to a vertex in $G'$. Thus, $w(u)>1$ implies that $\ell(\id_u)=w(u)$. $\ell(\id)<w(u)$ for some $u\in\cD^{(i)}(\id_u)$ only happens when $w(u)=1$. There are at most $n$ such $u$ since $P$ is a simple path, so the lemma follows.
    
\end{proof}

\begin{proof}[Proof of \Cref{lem:ViHdtress}]
    According to the definitions and the assumption that all the SSF operations are correct, we have
    \begin{align*}
        \deg_G(V^{(i)}_{H,<\dtres})&=\sum_{\id\in V^{(i)}_{<\dtres}}\deg_G(\cD^{(i)}(\id))\\
        &=\sum_{\id\in V^{(i)}_{<\dtres}}\sigma^{\deg}(\id)\\
        &=\deg^{(i)}_<(\dtres)\le 10{k r \mu}
    \end{align*}
    According to \Cref{lem:contract} (2), to show $G[V^{(i)}_{H,<\dtres}]$ is connected, it suffices to show $G^{(i)}[V^{(i)}_{<\dtres}]$ is connected. Notice that $\tilde{d}^{(i)}(\id_x)=0$ since $\ell^{(i)}(x)=0$ according to \Cref{cla:basicw}, and $\cD^{(i)}(\id_x)=\{x\}$, so $\id_x\in V^{(i)}_{<\dtres}$. We will show that every vertex in $V^{(i)}_{<\dtres}$ is reachable from $\id_x$ in $T^{(i)}[V^{(i)}_{<\dtres}]$. Suppose to the contrary, there is a vertex $\id \in V^{(i)}_{<\dtres}$ such that there is no path from $\id_x$ to $\id$ on $T^{(i)}[V^{(i)}_{<\dtres}]$. Let path $P$ be the path from $\id_x$ to $\id$ on $T^{(i)}$, there must exists a vertex $\id'\in V(P)$ such that $\id'\not\in V^{(i)}_{<\dtres}$. However, according to \Cref{thm:sssp}, $\tilde{d}^{(i)}(\id')$ is the length of the path from $\id_x$ to $\id'$ on the tree $T^{(i)}$, which is a subpath of the path from $\id_x$ to $v$ on the tree $T^{(i)}$. Since the lengths of vertices are non-negative, we get $\tilde{d}^{(i)}(\id')\le \tilde{d}^{(i)}(v)<\dtres$, contradicting the fact that $\id'\not\in V^{(i)}_{<\dtres}$.
\end{proof}

\begin{proof}[Proof of \Cref{lem:Tiledtressconnected}]
    The proof follows similarly to the previous proof. According to \Cref{cla:basicw}, $\ell(\id_x)=0$ and $\id_x\in V^{(i)}_{\le\dtres}$. We will show that every vertex in $V^{(i)}_{\le\dtres}$ is reachable from $\id_x$. Suppose to the contrary, there is a vertex $\id \in V^{(i)}_{\le\dtres}$ such that there is no path from $\id_x$ to $\id$ on $T^{(i)}[V^{(i)}_{\le\dtres}]$. Let path $P$ be the path from $\id_x$ to $\id$ on $T^{(i)}$, there must exists a vertex $\id'\in V(P)$ such that $\id'\not\in V^{(i)}_{\le\dtres}$. However, according to \Cref{thm:sssp}, $\tilde{d}^{(i)}(\id')$ is the length of the path from $\id_x$ to $\id'$ on the tree $T^{(i)}$, which is a subpath of the path from $\id_x$ to $\id$ on the tree $T^{(i)}$. Since the lengths of vertices are non-negative, we get $\tilde{d}^{(i)}(\id')\le \tilde{d}^{(i)}(v)\le \dtres$, contradicting the fact that $\id'\not\in V^{(i)}_{\le\dtres}$.
\end{proof}
    
    \begin{proof}[Proof of \Cref{lem:ViHledtres}]
        There are two cases for the definition of $V^{(i)}_{H,\le\dtres}$ in Step 4 of the algorithm.
        
        If it is case 1, we denote
        
        \[\tilde{V}^{(i)}_{H,\le\dtres}=\bigcup_{\id\in V(\hat{T}^{(i)})}\cD^{(i)}(\id)\]

        If it is case 2, we denote 
        \[\tilde{V}^{(i)}_{H,\le\dtres}=\bigcup_{\id\in V(\hat{T}^{(i)})\bigcup V(T^{(i)}_{\big})}\cD^{(i)}(\id)\]
        
        Notice that if all the SSF operations are correct, then according to the definition of $V^{(i)}_{H,\le\dtres}$, we have
        \[V^{(i)}_{H,\le\dtres}=\tV^{(i)}_{H,\le\dtres}\cup V^{(i)}_{H,<\dtres}\]
        
        Since $\hat{T}^{(i)}$ is a tree, $\hat{T}^{(i)}$ is connected. According to \Cref{lem:contract} (2), if it is Case 1, then $H[\tilde{V}^{(i)}_{H,\le\dtres}]$ is connected. According to \Cref{lem:ViHdtress}, the part $V^{(i)}_{H,<\dtres}$ is connected and shares a common vertex $x$ to $\tilde{V}^{(i)}_{H,\le\dtres}$, so $G^{(i)}_{H,\le\dtres}$ is connected. In Case 2, the only vertices added to $\tV^{(i)}_{H,\le\dtres}$ are from a tree of $H$, which has an edge connecting to $a^{(i)}_\big\in \tV^{(i)}_{H,\le\dtres}$, so $G^{(i)}_{H,\le\dtres}$ is also connected if it is Case 2. Next, we argue the size bound.
        
        If $\deg_G\left(V^{(i)}_{H,\le\dtres}\right)$ is from Case 1, i.e., if we have 
        \[5{k r \mu}\le \sigma^{\deg}(V(\hat{T}^{(i)}))\le 10{k r \mu}\]

        Then according to the definition of $\sigma^{\deg}$ and $V(\hat{T}^{(i)})$, we get
        \begin{align*}
            \deg_G\left(\tV^{(i)}_{H,\le\dtres}\right)&=\sum_{\id\in V(\hat{T}^{(i)})}\deg_G\left(\cD^{(i)}(\id)\right)\\
            &=\sum_{\id\in V(\hat{T}^{(i)})}\sigma^{\deg}(\id)\\
            &=\sigma^{\deg}(V(\hat{T}^{(i)}))\le 10{k r \mu}
        \end{align*}
        If $\deg_G\left(V^{(i)}_{H,\le\dtres}\right)$ is from Case 2, i.e., if we have 
        \[\sigma^{\deg}(V(\hat{T}^{(i)}))< 5{k r \mu}\]
        then 
        \begin{align*}
            \deg_G\left(\tV^{(i)}_{H,\le\dtres}\right)&=\left(\sum_{\id\in V(\hat{T}^{(i)})}\deg_G\left(\cD^{(i)}(\id)\right)\right)+\deg_G(V(T^{(i)}_{\big}))\\
                &=\sigma^{\deg}(V(\hat{T}^{(i)}))+\deg_G(V(T^{(i)}_{\big}))
        \end{align*}

        The first term is between $0$ and $5{k r \mu}$. To bound the second term, notice that from the input guarantee of \Cref{lem:localvertexcut}, we have $\deg_G(v)\le 5\mu $ for every $v\in V_{\inner}$. Also remember that for the data structure $\cD^{(i)}$ has the weight $\sigma=\deg_G(v)$. So $\sigma(v)>2{k r \mu}$ will never happen for any $v\in V_{\inner}$. Thus, when calling \Cref{eq:step4_c2}, according to \Cref{lem:dynamicspanningforest}, we must have

        \[2{k r \mu}\le \deg_G(V(T^{(i)}_{\big}))\le 4{k r \mu}\]

        So we get 
        \[2{k r \mu}\le \deg_G\left(\tV^{(i)}_{H,\le\dtres}\right)\le 9{k r \mu}\]

        According to \Cref{lem:ViHdtress}, we get
        \[2{k r \mu}\le \deg_G\left(V^{(i)}_{H,\le\dtres}\right)\le 20{k r \mu}\]
        
    \end{proof}

\subsection{Correctness}\label{subsec:localcutcorrectness}
In the section, we prove the correctness of \Cref{lem:localvertexcut}. We need to verify two things.

\begin{itemize}
    \item Assuming all the SSF operations are correct, every time the algorithm outputs a fractional cut $C$ with a vertex $\tau$, $C$ is a valid $(x, \tau)$-fractional cut of size at most $k-0.5$.
    \item If there is a $(x,\mu)$-local cut $(L,S,R)$ of size at most $k$ such that $L\subseteq V_{\inner}$, then the algorithm returns a fractional vertex cut $C$ of size at most $k - 0.5$ with probability at least $0.01$. 
\end{itemize}

Recall from our assumptions that $N_G[V_{\inner}]\neq V$, implying the existence of some vertex $t\in V - N_G[V_{\inner}]$. Define a new graph $G'$ by adding edges connecting the vertex $t$ to each vertex in $N_G(V_{\inner})$. The following lemma establishes the equivalence of fractional cuts between graphs $G$ and $G'$.

\begin{lemma}\label{lem:GandGprime}
    Every fractional cut in $G'$ is also a fractional cut in $G$. Moreover, if there is a $(x,\mu)$-local cut $(L,S,R)$ of size less than $k$ in $G$ with $L\subseteq V_{\inner}$, it is also a local cut for $G'$.
\end{lemma}

\begin{proof}
    First, suppose $C$ is a fractional cut in the graph $G'$. Since the edge set of $G$ is a subset of the edge set of $G'$, distances induced by the fractional cut $C$ cannot decrease when considering $G$ instead of $G'$. Therefore, the fractional cut $C$ remains valid in graph $G$.

    Next, consider a $(x,\mu)$-local cut $(L,S,R)$ in $G$ where $L\subseteq V_{\inner}$. The cut clearly remains a $(x,\mu)$-local cut in $G'$. Additionally, it is also valid in $G'$, because all newly introduced edges in $G'$ are adjacent only to vertex $t$, which by construction does not lie in $V - N_{G}[V_{\inner}] \subseteq  R$ and thus does not affect vertices in $L$ or $S$.
\end{proof}

Thus, it suffices to consider $G'$ instead of $G$. In what follows, when we say $G$, we mean $G'$. 

\paragraph{Correctness: the output is a valid cut.} Assume all the SSF operations are correct. Suppose the algorithm outputs $(C, \tau)$ in step 5 at some iteration $i$. Without loss of generality, we assume $\tau = t$ if $P^{(i)}$ is constructed in Step 3. Recall that $\tau = v_{H}^{(i)}$ when $P^{(i)}$ is constructed in Step 4. We will prove that $C$ is a valid vertex cut of size at most $k-0.6$ in $G$ assuming that all the SSF operations are correct. The proof is based on the following two important lemmas which show that $P^{(i)}$ is always an approximate shortest path.
    
\begin{lemma}\label{lem:step3shortestpath} Assume that all the SSF operations are correct.
    In every iteration, if $P^{(i)}$ is constructed in Step 3, then $P^{(i)}\circ(t)$ is a $(1+3\eps)$-approximate $(x,t)$-shortest path with respect to $w^{(i)}=\tilde{w}^{(i)}$ on $G$.
\end{lemma}

\begin{lemma}\label{lem:step4shortestpath} Assume that all the SSF operations are correct.
    In every iteration, if $P^{(i)}$ is constructed in Step 4, then $P^{(i)}\circ(v^{(i)}_H)$ is a $(1+2\eps)$-approximate $(x,v^{(i)}_H)$-shortest path with respect to $\tilde{w}^{(i)}$ on $G$.
\end{lemma}

Before we prove the two lemmas, we first show how to use them to prove the output is a valid cut of size $k-0.5$. Let $P^*=P^{(i)}\circ(\tau)$. We will show that $C$ is a fractional $(x,\tau)$-vertex cut. Notice that $C$ already satisfies $C(x)=C(\tau)=0$, for any cases of $\tau$; $t\not\in N[V_\inner]$, according to \Cref{cla:basicw}, we have $C(t)=0$; according to the definition of $\tilde{w}^{(i)}$ we have $C(v^{(i)}_H)=0$; we have $C(x)=0$ according to \Cref{cla:basicw}. Thus, $C$ satisfies the condition that $C(x)=C(\tau)=0$.

Next, we verify the condition of distance with respect to $C$. According to \Cref{lem:step3shortestpath,lem:step4shortestpath}, we have that

\[\dist_{G,\tilde{w}^{(i)}}(x, \tau)\ge (1-4\eps)\cdot \tilde{w}^{(i)}(P^*)\ge \frac{(1-4\eps)\cdot \tW^{(i)}}{k-0.6}\]

The last inequality is due to the return condition for Step 5. According to the definition of $C$, we get

\[\dist_{G,C}(x, \tau)\ge \left(\frac{(1-4\eps)\cdot \tW^{(i)}}{k-0.6}-n\right)\cdot \frac{k-0.6}{\tW^{(i)}}\cdot (1+5\eps)\ge 1\]

The first inequality contains two parts, one is the $-n$ term due to the fact that $C(u)$ decreases from $\tilde{w}^{(i)}(u)$ by $1$ for at most $n$ many trivial vertices, and the multiplicative $\frac{k-0.6}{\tW^{(i)}}\cdot (1+5\eps)$ factor is from the scaling for non-trivial vertices. The second inequality uses the fact that $\tW^{(i)}$ is larger than $n^3$. Now we proved that $C$ is a valid fractional cut. The size of the cut is upper bounded by

\[\sum_{u\in V}C(u)\le \sum_{\tilde{w}^{(i)}(u)\not=1}\frac{\tilde{w}^{(i)}(u)}{\tW^{(i)}}\cdot (k-0.6)\cdot (1+5\eps)\le k-0.5.\]

Now we proceed to prove \Cref{lem:step3shortestpath,lem:step4shortestpath}.
\begin{proof}[Proof of \Cref{lem:step3shortestpath}]
Let $E^{(i)}_{\not=1}$ be the edge set of adjacent edges to non-trivial nodes in the $i$-th iteration, i.e.,
    \[E^{(i)}_{\not=1}:=\bigcup_{w^{(i)}(u)\not=1}\delta_H(u).\]
    Notice that $\cD^{(i)}$ is maintaining connected components of 
    \[H^{(i)}_{\rem}:=H-E^{(i)}_{\not=1}.\]
    Let $\hP$ be an arbitrary simple $(x,t)$-path on $G$, we will show that $w^{(i)}(\hat{P})\ge (1-2\eps)\cdot w^{(i)}(P^{(i)}\circ(t))$.
    
    Denote the first vertex in the path $\hP$ that is not a vertex in $V^{(i)}_{H,<\dtres}$ by $v_{\far}$. Such a vertex must exist since $t\not\in V^{(i)}_{H,<\dtres}$. There are two cases.

    \paragraph{Case 1: $v_{\far}\in V_{\inner}$.} In this case, we have $v_{\far}\in V_{\inner}-V^{(i)}_{H,<\dtres}$. Let $K$ be the connected component of $H^{(i)}_{\rem}$ that contains $v_{\far}$. According to the construction of $V^{(i)}_{H,<\dtres}$, i.e.,

    \[V^{(i)}_{H,<\dtres}\leftarrow \bigcup \cD^{(i)}.\alg{Components}\left(V^{(i)}_{<\dtres}\right).\]

    and assuming all the SSF operations are correct, we have that $\id_{\cD^{(i)}}(K)\not\in V^{(i)}_{H,<\dtres}$. Thus, according to the definition of $V^{(i)}_{H,<\dtres}$, we have 
    \[\tilde{d}^{(i)}(\id_{\cD^{(i)}}(K))\ge \dtres.\]

    Since $\tilde{d}^{(i)}$ is an $(1+\eps)$-approximate distance function on $G^{(i)}$, every $(\id_{\cD^{(i)}}(\{x\}),\id_{\cD^{(i)}}(K))$-path has length at least $(1-\eps)\cdot \dtres$ with respect to $\ell^{(i)}$.
    
    Since $v_{\far}$ is the first vertex in $\hP$ such that $v_{\far}\not\in V^{(i)}_{H,<\dtres}$, the subpath $\hat{P}[x,v_{\far}]$ only contains vertices in $V_{\inner}$. 
    According to \Cref{lem:contract} (4), the corresponding path $\CT{\hP[x,v_{\far}]}{G^{(i)}}$ is a $(\id_{\cD^{(i)}}(\{x\}),\id_{\cD^{(i)}}(K))$-path in $G^{(i)}$ and satisfies that

    \[w^{(i)}(\hP)\ge w^{(i)}(\hP[x,v_{\far}])\ge \ell^{(i)}(\CT{\hP[x,v_{\far}]}{G^{(i)}})\ge (1-\eps)\cdot \dtres\]

    According to the last line of Step 3, $P^{(i)}$ is constructed in Step 3 only if $d^{(i)}_{\aug}(v^{(i)}_*)\le \dtres$, where $d^{(i)}_{\aug}(v^{(i)}_*)$ is the length of the path in $T^{(i)}_{\aug}$ from $x$ to $v^{(i)}_*$, which is equal to $w^{(i)}(P^{(i)}\circ(t))-1$ according to the definition of $P^{(i)}$ and the fact that $w^{(i)}(t)=1$ according to \Cref{cla:basicw}. Thus, we get
    \[w^{(i)}(\hP)\ge (1-\eps)\cdot \dtres - n\ge (1-\eps)\cdot (w^{(i)}(P^{(i)}\circ(t))-1)-n\ge (1-2\eps)\cdot w^{(i)}(P^{(i)}\circ(t))\]

    The last inequality is from the fact that $w^{(i)}(P^{(i)}\circ(t))\ge n^3$ according to \Cref{cla:basicw}.

    \paragraph{Case 2 : $v_{\far}\not\in V_{\inner}$.} In this case, since $v_{\far}$ is the first vertex on $\hP$ not in $V^{(i)}_{H,<\dtres}$, we have that $\Prec_{\hP}(v_{\far})\in V^{(i)}_{H,<\dtres}$. According to the definition of $V_{\outer}$, we have $v_{\far}\in V_{\outer}$, and the whole subpath $\hP[x,v_{\far}]$ is in $G^{(i)}_\aug$. According to the definition of $\tilde{d}^{(i)}_\aug$, we have that $\tilde{d}^{(i)}_{\aug}(v_\far)=w^{(i)}(P^{(i)}\circ(t))-1$ is a $(1+\eps)$-approximate distance in $G^{(i)}_\aug$. Thus, we get

    \[w^{(i)}(\hP)\ge (1-\eps)\cdot (w^{(i)}(P^{(i)}\circ(t))-1)-n\ge (1-2\eps)\cdot w^{(i)}(P^{(i)}\circ(t))\]

    Again, the last inequality is from the fact that $w^{(i)}(P^{(i)}\circ(t))\ge n^3$ according to \Cref{cla:basicw}.
\end{proof}

\begin{proof}[Proof of \Cref{lem:step4shortestpath}]
    Let $\hP$ be an arbitrary simple $(x,v^{(i)}_H)$-path on $G$. We will show that $\tilde{w}^{(i)}(\hat{P})\ge (1-2\eps)\cdot \tilde{w}^{(i)}(P^{(i)}\circ(v^{(i)}_H))$. First we relate the value $\dtres$ with $\tilde{w}^{(i)}(P^{(i)}\circ(v^{(i)}_H))$.
    \begin{lemma}\label{lem:dtresPi}
        $\dtres\ge (1-\eps)\cdot \tilde{w}^{(i)}(P^{(i)}\circ(v^{(i)}_H))-n-1.$
    \end{lemma}
    \begin{proof}
    According to the definition of  $\hat{T}^{(i)}$, it is a subtree of $T^{(i)}[V^{(i)}_{\le \dtres}]$, where $T^{(i)}[V^{(i)}_{\le \dtres}]$ is an approximate shortest path tree where every node $\id\in V^{(i)}_{\le \dtres}$ has approximate distance at most $\tilde{d}^{(i)}(\id)\le \dtres$ according to the definition of $V^{(i)}_{\le \dtres}$. According to the definition of $v^{(i)}_H$, there exists a neighborhood $v_{\pre}$ of $v^{(i)}_H$ in $H$ such that $v_{\pre}\in V^{(i)}_{H,\le\dtres}$. Let $\id_{\pre}$ be the identifier in $\cD^{(i)}$ that represent the connected component containing $v_{\pre}$. We have (if $\id^{(i)}_\big$ exists)
    
    \[\id_{\pre}\in V^{(i)}_{<\dtres}\cup\{\id^{(i)}_{\big}\}\cup V(\hat{T}^{(i)})\]
    
    We claim that there exists a path from $\id_x$ to $\id_{\pre}$ of length at most $\dtres+1$ with respect to $\ell^{(i)}$ using only vertices in $V^{(i)}_{<\dtres}\cup\{\id^{(i)}_{\big}\}\cup V(\hat{T}^{(i)})$: if $\id_{\pre}\in V(\hat{T}^{(i)})$ then there is a $(\id_x,\id_{\pre})$ path on $\hat{T}^{(i)}$; if $\id_{\pre}\in V^{(i)}_{<\dtres}$ then there is a path on $T^{(i)}[V^{(i)}_{<\dtres}]$ which is connected; if $\id_{\pre}=\id^{(i)}_{\big}$, notice that $\id^{(i)}_{\big}$ has $\ell^{(i)}$ length $1$ (only trivial vertices can form a big cluster) and is connected to $ V(\hat{T}^{(i)})$, so this reduces to the previous case.
    
    Denote such path as $P^*$. We construct an arbitrary path $P^*_H$ such that $\CT{P^*_H}{G^{(i)}}=P^*$ and $\Start(P^*_H)=x,\End(P^*_H)=v_{\pre}$, so that $V(P^*_H)\subseteq V^{(i)}_{H,\le\dtres}$. According to \Cref{lem:contract} (4), we get
    \begin{align*}
        \dtres+1&\ge \ell^{(i)}(P^*)\\
        &\ge w^{(i)}(P^*_H)-n\\
        &\ge \tilde{w}^{(i)}(P^*_H\circ(v^{(i)}_H))-n\\
        &\ge (1-\eps)\tilde{w}^{(i)}(P^{(i)}\circ(v^{(i)}_H))-n\\
    \end{align*}
        The inequalities are because $P^{(i)}$ is an $(1+\eps)$-approximate $(x,v^{(i)}_H)$-shortest path with respect to $\tilde{w}^{(i)}$, which differs from $w^{(i)}$ by only setting the weight of the ending vertex $v^{(i)}_H$ to be $0$ (remember that $P^{(i)}$ is the path from $x$ to $v^{(i)}_H$ excluding $v^{(i)}_H$ and $P^*_H$ is also a path from $x$ to $v^{(i)}_H$ excluding $v^{(i)}_H$). 
    \end{proof}
    Now we continue to analyze the length of $\hP$. There are three cases, depending on whether $V(\hP)$ excluding the last vertex is totally inside $V^{(i)}_{H,<\dtres}$ or not, and if not, what is the first vertex not in $V^{(i)}_{H,<\dtres}$.

    \paragraph{Case 1: $V(\hP)-\{v^{(i)}_H\}\subseteq V^{(i)}_{H,<\dtres}$.} Notice that 
    In this case, $V(\hP)$ must be totally inside $G^{(i)}_{H,\le\dtres}$ since all edges $E(\hP)$ are adjacent to a vertex in $V^{(i)}_{H,\le \dtres}$. According to the definition of $P^{(i)}$, we know that $P^{(i)}\circ(v^{(i)}_H)$ is an approximate shortest path in $G^{(i)}_{H,\le\dtres}$ with respect to $\tilde{w}^{(i)}$ where $\hP$ is also a path in $G^{(i)}_{H,\le\dtres}$, so

    \[\tilde{w}^{(i)}(\hP)\ge(1-\eps)\cdot \tilde{w}^{(i)}(P^{(i)}\circ(v^{(i)}_H))\ge(1-2\eps)\cdot \tilde{w}^{(i)}(P^{(i)}\circ(v^{(i)}_H))\]
    
    In the following cases, we assume there is a vertex $v_{\far}$ in $V(\hP)-\{v^{(i)}_H\}$ such that $v_{\far}\not\in V^{(i)}_{H,< \dtres}$ and $v_{\far}$ is the first such vertex. 

    \paragraph{Case 2: $v_\far\in V-V_{\inner}$.} In this case, remember that $v_{\far}$ is the first vertex on $\hP$ that is not in $v_{\far}\not\in V^{(i)}_{H,< \dtres}$, so the subpath $\hP[x,v_\far]$ is a path in $G^{(i)}_\aug$ which includes all edges adjacent to $V^{(i)}_{H,< \dtres}$. According to the fact that $v^{(i)}_*$ is the vertex in $V^{(i)}_\outer$ with the smallest $d^{(i)}_{\aug}(v^{(i)}_*)$, we have
    \[d^{(i)}_\aug(v_\far)\ge d^{(i)}_{\aug}(v^{(i)}_*)> \dtres\]
    
    The last inequality is the condition to enter Step 4. As $\hP[x,v_\far]$ is a path in $G^{(i)}_\aug$, we get
    \[w^{(i)}(\hP[x,v_\far])\ge (1-\eps)\cdot d^{(i)}_\aug(v_\far)> (1-\eps)\cdot \dtres\]

    According to \Cref{lem:dtresPi}, we finally get

    \begin{align*}
        \tilde{w}^{(i)}(\hP)&\ge w^{(i)}(\hP[x,v_\far])\\
        &\ge(1-\eps)\cdot \dtres\\
        &\ge (1-1.9\eps)w^{(i)}(P^{(i)})-n-1\\
        &\ge (1-2\eps)w^{(i)}(P^{(i)})\\
        &= (1-2\eps)\tilde{w}^{(i)}(P^{(i)}\circ(v^{(i)}_H))
    \end{align*}

    The last inequality is because \Cref{cla:basicw} that $w^{(i)}(P^{(i)})\ge n^3$.

    \paragraph{Case 3: $v_\far\in V_{\inner}$.} In this case, let $\id_\far$ be the identifier of the connected component stored in $\cD$ containing $v_\far$. We must have $\id_\far\not\in V^{(i)}_{<\dtres}$ as otherwise according to the definition of $V^{(i)}_{H,<\dtres}$ we would have $v_{\far}\in V^{(i)}_{H,<\dtres}$.

    According to the definition of $V^{(i)}_{<\dtres}$, we must have $\tilde{d}^{(i)}(\id_\far)\ge \dtres$. Thus, consider the path mapping from $\hP[x,v_\far]$ to $G^{(i)}$ and the fact that $\tilde{d}^{(i)}$ is an approximate distance function on $G^{(i)}$, we have

    \begin{align*}
        \tilde{w}^{(i)}(\hP)&\ge w^{(i)}(\hP[x,v_{\far}])\\
        &\ge \ell^{(i)}(\CT{\hP[x,v_{\far}]}{G^{(i)}})\\
        &\ge (1-\eps)\cdot \tilde{d}^{(i)}(\id_\far)\\
        &\ge (1-\eps)\cdot \dtres
    \end{align*}

    The first inequality is because $\tilde{w}^{(i)}$ compared to $w^{(i)}$ only set the weight of $v^{(i)}_H$ to $0$ and $\hP$ does not contain $v^{(i)}_H$ except the last vertex. The second inequality is from \Cref{lem:contract} (4). The third inequality is because $\tilde{d}$ is an approximate distance function on $G^{(i)}$.

    According to \Cref{lem:dtresPi}, we get

    \begin{align*}
        \tilde{w}^{(i)}(\hP)&\ge (1-\eps)\cdot \dtres\\
        &\ge (1-1.9\eps)\cdot \tilde{w}^{(i)}(P^{(i)}\circ(v^{(i)}_H))-(1-\eps)\cdot (n-1)\\
        &\ge (1-2\eps)\cdot \tilde{w}^{(i)}(P^{(i)}\circ(v^{(i)}_H))
    \end{align*}
    The last inequality is because \Cref{cla:basicw} that $w^{(i)}(P^{(i)})\ge n^3$. 
\end{proof}
\paragraph{Correctness: output a fractional cut with probability at least $0.01$.} Let us assume there is a $(x,\mu)$-local cut $(L,S,R)$ of size at most $k$ such that $L\subseteq V_{\inner}$. 

We first prove that the algorithm will not return $\bot$ in the initialization phase by meeting \Cref{eq:deg_condition}. Notice that since $x\in L$ and $(L,S,R)$ is a cut, we have

\[\deg_G(N_H(x))=\deg_G(N_H(x)\cap L)+\deg_G (N_H(x)\cap S)\]

Since $\deg_G(L)\le 2\mu$ according to the definition of $(x,\mu)$-local cut, the first term is bounded by

\[\deg_G(N_H(x)\cap L)\le 2\mu\]

Since $\abs{S} < k$ and the degree of every node in $V_{\inner}$ is at most $5\mu$, the second term is bounded by

\[\deg_G(N_H(x)\cap S)\le 5(k-1)\mu.\]

Therefore, the inequality \Cref{eq:deg_condition} is always false if there is a local cut. Now, assuming the algorithm does not terminate in the initialization phase, we proceed to analyze the rest the the algorithm.

We will show that the following \emph{good event} happens with probability at least $0.01$, then show that if it happens, the algorithm always returns a fractional vertex cut $C$ of size at most $k-0.6$.

\begin{description}
    \item[Good event $\cG$.] For every $i$, in Step 4, if $v^{(i)}_{H}$ is computed, then $v^{(i)}_H\in R$. Besides, all the SSF operations are correct.
\end{description}

\begin{lemma} \label{lem:good_happens}
    $\cG$ happens with probability at least $0.01$.
\end{lemma}
\begin{proof}
    According to \Cref{lem:dynamicspanningforest,rmk:oblivious}, all the SSF operations are correct with high probability. We will be conditioned on this and calculate the probability that $v^{(i)}_H\in R$ for every $i\in[r]$ executing Step 4. For an iteration $i$ executes Step 4, \Cref{lem:ViHledtres} gives
    \[2{k r \mu}\le \deg_G\left(V^{(i)}_{H,\le\dtres}\right)\le 20{k r \mu}.\]

    Remember that $v^{(i)}_H$ is defined as the random endpoint of a random edge from $E^{(i)}_{H,\le\dtres}$. Also recall that $E^{(i)}_{H,\le\dtres}$ contains all edge adjacent to $V^{(i)}_{H,\le\dtres}$. Thus, we have that

    \[2{k r \mu}\le |E^{(i)}_{H,\le\dtres}|\le 20{k r \mu}.\]

    According to the definition of $(x,\mu)$-local cut, we have that $\deg_G(L)\le \mu$. Moreover, we have $|S|\le k$. Remember that $V(H)=V_{\inner}$ where vertices in $V_{\inner}$ all have degree at most $5\mu$ according to the input assumption. Thus, we get

    \[\abs{\deg_G(L)+\deg_G(S\cap V_{\inner})}\le \mu+(k-1)\cdot 5\mu\le 2{k r \mu}/(0.4r).\]

    In other words, at most $\frac{1}{0.4r}$ fraction of the edges in $E^{(i)}_{H,\le\dtres}$ can either be adjacent to $L$, or adjacent to $S\cap V_{\inner}$, while all other edges must have at least one endpoint in $R$ and another endpoint either in $R$ or in $S-V_{\inner}$. In the case when an edge has one endpoint in $R$ and another endpoint in $S-V_{\inner}$, the endpoint in $S-V_{\inner}$ is not in $V(H)$, which is trivially not in $V^{(i)}_{H,\le\dtres}$. According to the definition of $v^{(i)}_H$, if the sampled edge satisfies that one endpoint in $R$ and another endpoint either in $R$ or in $S-V_{\inner}$, then $v^{(i)}_H$ is from choosing the endpoint inside $V^{(i)}_{H,\le\dtres}$, which can only be in $R$. The probability is at least $1-\frac{1}{0.4r}$. Notice that this probability is independent of all other random choices of the algorithm. Thus, throughout all the $r$-iterations, the probability that we always have $v^{(i)}_H\in R$ is at least
    \[\left(1-\frac{1}{0.4 r}\right)^r\ge 0.05\]
    for any $r \ge 9$.
\end{proof}

Now we assume $\cG$ happens. The correctness relies on the following lemma.

\begin{lemma}[Multiplicative Weight Updates]\label{lem:mwu}
    Let $G=(V,E)$ be an undirected graph. Suppose $(L,S,R)$ is a vertex cut of size at most $k$ and $x\in L$. Suppose
    \[\eps=\frac{1}{10k}\qquad \qquad r=\left\lceil\frac{40k^{2} \ln n}{\eps}\right\rceil\]
    
    If we initialize a weight function $w^{(1)}$ by
    \[
        w^{(1)}(u) =
        \begin{cases}
        0 & \text{if } u=x \\
        n^3 & \text{if } u\in N(x) \\
        1 & \text{if } u\in V-N[x] \\
        \end{cases}
    \]
    and we repeat for $r$ iterations, where in the $i$-iteration we choose an arbitrary vertex $y\in R$, and an arbitrary $(x, y)$-path $p^{(i)} \circ (y)$, update the weights by
    \[
        w^{(i+1)}(u) =
        \begin{cases}
        (1+\eps)\cdot w^{(i)}(u) & \text{if } u\in V(p^{(i)}) \\
        w^{(i)}(u) & \text{otherwise } \\
        \end{cases}
    \]
    then there must be an iteration such that 
    \[w^{(i)}(p^{(i)}) \ge \frac{1}{k-0.6}\cdot \sum_{\substack{v\in V-\{y\}\\ w^{(i)}(v)\not=1}}w^{(i)}(v)\]
\end{lemma}

Notice that \Cref{lem:mwu} captures exactly what our algorithm is doing:
\begin{itemize}
    \item The $\eps,r,w^{(1)}(u)$ is initialized as in the algorithm.
    \item Assuming the good event $\cG$ happens, then either Step 3 is executed, in which case $t\in R$ and we let $y=t$, or Step 4 is executed, in which case $v^{(i)}_H\in R$ and we let $y=v^{(i)}_H$. 
    \item The algorithm then finds a path $P^{(i)}$ starting at $x$, ending at a neighborhood of $y$. We let $p^{(i)}$ in \Cref{lem:mwu} to be $P^{(i)}$. Then the weight update of \Cref{lem:mwu} is exactly the same as Step 5 in our algorithm.
    \item By the definition of $\tilde{w}^{(i)}$ in the algorithm, we have that 
    \[w^{(i)}(p^{(i)})=\tilde{w}^{(i)}(P^{(i)})\qquad\qquad \sum_{\substack{v\in V-\{y\}\\ w^{(i)}(v)\not=1}}w^{(i)}(v)=\tW^{(i)}\]
\end{itemize}

Thus, according to \Cref{lem:mwu}, there must be an iteration such that \[w^{(i)}(P^{(i)}) \ge \frac{1}{k-0.6}\tW^{(i)},\]
meeting our return condition.

\begin{proof}[Proof of \Cref{lem:mwu}]
For each iteration \(i\), denote the terminal point of $p^{(i)}$ by $y^{(i)}$. define the potential
\[
\Phi^{(i)} := \sum_{\substack{v\in V \\ w^{(i)}(v)\neq 1}} w^{(i)}(v) \ge \sum_{\substack{v\in V-\{y^{(i)}\}\\ w^{(i)}(v)\not=1}}w^{(i)}(v).
\]

Now it suffices to show that there must be an iteration $i$ such that \begin{equation}\label{eq:terminate_condition} w^{(i)}(p^{(i)}) \ge \frac{1}{k-0.6}\Phi^{(i)}.\end{equation}

Suppose the contrary, that $w^{(i)}(p^{(i)}) < \Phi^{(i)}/(k-0.6)$ for in every iteration $i$. Let's compute increment of the potential, $\Phi^{(i+1)} - \Phi^{(i)}$. Note that $w^{(i)}(v) \neq w^{(i+1)}(v)$ only if $v \in V(p^{(i)})$. For a vertex $v \in V(p^{(i)})$, it contributes by $\eps  w^{(i)}(v)$ if $w^{(i)}(v) > 1$, otherwise by $(1 + \eps)w^{(i)}(v)=1+\eps$. As there are at most $n$ vertices of the second case, the increment is at most \[\Phi^{(i+1)} - \Phi^{(i)} \le \eps w^{(i)}(p^{(i)}) + n \le \left(1 + \frac{1}{n}\right)\cdot \eps w^{(i)}(p^{(i)}) \]

The second inequality comes from $n < n^{2}\eps \le \frac{\eps w^{(i)}(p^{(i)})}{n}$, since $w^{(i)}(p^{(i)}) \ge n^3$. The assumption gives \[\Phi^{(i+1)} \le \left(1 + \frac{n+1}{n}\frac{\eps}{k-0.6}\right)\Phi^{(i)}.\]

Iterating over \(r\) rounds yields

\[
\Phi^{(r+1)} \le \Phi^{(1)}\left(1+\frac{n+1}{n}\frac{\eps}{k-0.6}\right)^r \le \Phi^{(1)} \exp\left(\frac{n+1}{n}\frac{r\eps}{k-0.6}\right)
\]
from $(1 + \eps)^{r} \le e^{\eps r}$ for any $\eps > 0$. Since $w^{(1)}(y^{(1)})=1$ and all the other $v \notin N[x]$ has $w^{(1)}(v) = 1$, the initial potential $\Phi^{(1)}$ is

\[
\Phi^{(1)} = \sum_{v\in N(x)} n^3 \le n^{4}.
\]
Thus,
\[
\Phi^{(r+1)} \le n^{4} \exp\left(\frac{n+1}{n}\frac{r\eps}{k-0.6}\right) = \exp\left(4\ln n + \frac{n+1}{n}\frac{r \eps}{k - 0.6}\right) \le \exp\left(4\ln n + \frac{r\eps}{k-0.7}\right).
\]

The last inequality is from $n \ge 10k-7$.

Next, note that every \((x,y)\)–path must cross the vertex separator \(S\). Hence, in every iteration there is at least one vertex in \(S\) that lies on the path \(p^{(i)}\). Since \(|S|< k\), there exists some vertex \(s^*\in S\) that is updated in at least \(\lceil r / (k-1) \rceil\) iterations. It follows that
\[
w^{(r+1)}(s^*) \ge (1+\eps)^{r/(k-1)}.
\]
From $\ln(1 + x) \ge \frac{k-1}{k-0.8}x$ for any $0 < x < \frac{0.2}{k-1}$, we deduce that
\[
w^{(r+1)}(s^*) \ge \exp\Bigl(\frac{r\eps}{k-0.8}\Bigr).
\]
Since \(s^*\) contributes to the potential \(\Phi^{(r+1)}\), we must have
\[
\exp\Bigl(\frac{r\eps}{k-0.8}\Bigr) \le w^{(r+1)}(s^{\ast}) \le \Phi^{(r+1)} \le \exp\Bigl(4\ln n + \frac{r\eps}{k - 0.7}\Bigr).
\]
This implies \[r \le \frac{4(k-0.8)(k-0.7)\ln n}{0.1\eps} < \frac{40k^{2}\ln n}{\eps}\]

contradicting to the value of $\eps, r$. Therefore, there is an iteration $i$ satisfying \eqref{eq:terminate_condition}.
\end{proof}

\subsection{Complexity Analysis}\label{subsec:complexity}

In this section we analyze the complexity of the algorithm for \Cref{lem:localvertexcut}.

\paragraph{Initialization.} For the initialization, evaluating $\deg_{G}(x)$ and $\deg_{G}(N_{H}(x))$ takes $O(\mu)$ work and $\tO{1}$ given $\abs{N_{H}(x)} \le 2\mu$ by the input assumption of \Cref{lem:localvertexcut}. Initializing weights for $N_{H}[x]$ can be done in the same complexity.

The following lemma shows the complexity of an iteration.

\begin{lemma} \label{lem:loop_complexity}
    Assume all the SSF operations are correct. For every $i = 1, \cdots, r$, the iteration $i$ runs in $\tO{ikr\mu/\eps^2}$ work and $\tO{1}$ depth with high probability. Also, when it returns a fractional cut $(C, \tau)$, $C$ has at most $\tO{ikr\mu}$ nonzero entries.
\end{lemma}
Given \Cref{lem:loop_complexity}, notice that $\eps=O(1/k)$ and $r=\tO{k^2/\eps}=\tO{k^3}$, one can prove the complexity of \Cref{lem:localvertexcut} as the total work is $\tO{kr^{3}\mu / \eps^2} = \tO{k^{12}\mu}$, and depth is $\tO{r} = \tO{k^3}$ with high probability, and the returned fractional cut $C$ contains $\tO{kr^{2}\mu}=\tO{k^7\mu}$ nonzero entries.

\begin{proof}[Proof of \Cref{lem:loop_complexity}]
    First, we define the set of non-trivial vertices in the iteration $i$ as \[V_{\ne 1}^{(i)} := \{v \in V_{\inner} \mid w^{(i)}(v) \ne 1\}\]. We prove by induction on $i$ that for every $i\in[r]$, 
    \[\deg_{H}(V_{\neq 1}^{(i)}) \le 20ikr\mu.\] 
    
    It is true when $i = 1$ from the initialization constraint because only $N_H(x)$ contains non-trivial vertices, and $N_H(x)\le 2\mu+k\cdot (5\mu)$. The vertices newly getting nontrivial weight at iteration $i+1$ is a subset of \(V(P^{(i)})\). Hence the increment on the left-hand-side is at most $\deg_{H}(V(P^{(i)}))$. Note that $V(P^{(i)} )\subseteq V_{H, <\dtres}^{(i)}$ if it's from Step 3, otherwise $V(P^{(i)}) \subseteq V_{H, \le\dtres}^{(i)}$. Both of them have degree at most $20kr\mu$ according to \Cref{lem:ViHdtress} and \Cref{lem:ViHledtres}, respectively. Now we analyze the complexity step by step.

    \begin{description}
        \item[Step 1.] First we analyze the complexity of     \[\left(G^{(i)}=(V^{(i)},E^{(i)}),\ell^{(i)}\right)\leftarrow \alg{Contract}(G,V_{\inner},\cD^{(i)},w^{(i)}).\]
        From definitions of $V_{\ne 1}, E_{\ne 1}$ inside $\alg{Contract}$,
        \[
            \abs{V_{\ne 1}}, \abs{E_{\ne 1}} \le \deg_{H}(V_{\ne 1}^{(i)}) \le 20ikr\mu.
        \]
        According to \Cref{lem:dynamicspanningforest}, it takes $\tO{ikr\mu}$ work and $\tO{1}$ depth to execute \(\cD.\alg{Fail}(E_{\neq 1})\) and \(\cD.\alg{ID}(V_{\neq 1})\). Note that $\abs{V^{(i)}}, \abs{E^{(i)}} \le \deg_{H}(V_{\ne 1}^{(i)}) \le 20ikr\mu$ from the definition. Thus, the overall complexity for $\alg{Contract}$ is $\tO{ikr\mu}$ work and $\tO{1}$ depth with high probability. 
        
        In the next line, $\id_{x}$ can be fetched in $\tO{1}$ work and depth, again from \Cref{lem:dynamicspanningforest}. To execute     \[(T^{(i)},\tilde{d}^{(i)})\leftarrow \alg{SSSP}(G^{(i)},\ell^{(i)},\id_x,\eps),\]it takes  $\tO{\abs{E^{(i)}}/\eps^{2}} = \tO{ikr\mu/\eps^2}$ work and $\tO{1}$ depth. Putting everything together, Step 1 can be performed in $\tO{ikr\mu/\eps^2}$ work and $\tO{1}$ depth with high probability.

        \item[Step 2.] The degree sum \[(\sigma^{\deg}(\id))_{\id\in V^{(i)}}\leftarrow \cD^{(i)}.\alg{Sum}(V^{(i)})\] can be evaluated in $\tO{\abs{V^{(i)}}} = \tO{ikr\mu}$ work and $\tO{1}$ depth.
        
        To sort the vertices $V^{(i)}$, it also takes $\tO{\abs{V^{(i)}}}$ work and $\tO{1}$ depth. Similarly, to find $\dtres$ doing binary search on the prefix sum of $\sigma^{\deg}$ takes $\tO{V^{(i)}}$ work and $\tO{1}$ depth. In sum, Step 2 takes $\tO{ikr\mu}$ work and $\tO{1}$ depth with high probability.

        \item[Step 3.] Now we analyze \[V^{(i)}_{H,<\dtres}\leftarrow \bigcup \cD^{(i)}.\alg{Components}\left(V^{(i)}_{<\dtres}\right).\] According to \Cref{lem:ViHdtress}, $\deg_{G}(V_{H, <\dtres}^{(i)}) \le10kr\mu$. Hence it takes $\tO{kr\mu}$ work and $\tO{1}$ depth from \Cref{lem:dynamicspanningforest}. $V_{\outer}^{(i)}$ can be computed in $\tO{\deg_{G}(V_{H, <\dtres}^{(i)})}=\tO{kr\mu}$ work and $\tO{1}$ by looking up all the vertices in $N_{G}(V_{H, <\dtres}^{(i)})$ in $V_{\inner}$. Also, $V_{\aug}^{(i)}, E_{\aug}^{(i)}$ can be computed in $\tO{\deg_{G}(V_{H, <\dtres}^{(i)})}$ work and $\tO{1}$ depth. Note that $\abs{V_{\aug}^{(i)}}, \abs{E^{(i)}_{\aug}} \le \deg_{G}(V_{H, <\dtres}^{(i)})$.
        
        \[(T^{(i)}_{\aug},\tilde{d}^{(i)}_\aug)\leftarrow \alg{SSSP}(G^{(i)}_\aug,x,w^{(i)}|_{V^{(i)}_{\aug}},\eps)\]
        takes $\tO{\abs{E_{\aug}^{(i)}}/{\eps^2}} = \tO{kr\mu/\eps^2}$ work and $\tO{1}$ depth. 
        
        $v_{*}^{(i)}$ can be determined in $\abs{V_{\outer}^{(i)}}$ work and $\tO{1}$ depth. If $\tilde{d}_{\aug}^{(i)}(v_{i}^{\ast}) \le \dtres$, $P^{(i)}$ can be computed in $\tO{\abs{V_{\aug}^{(i)}}}$ work and $\tO{1}$ depth. In sum, Step 3 takes $\tO{kr\mu/\eps^2}$ work and $\tO{1}$ depth with high probability.

        \item[Step 4.] From \Cref{lem:ViHledtres}, we know that $\deg_{G}(V_{\le\dtres}^{(i)}) \le 20kr\mu$. Denote $K^{(i)} := T^{(i)}[V^{(i)}_{\le\dtres}]$. Inside the $\alg{Tree}$ subroutine \eqref{eq:hatT}, it executes \[\tilde{\cD}.\alg{Init}(K^{(i)}),\tilde{\cD}.\alg{Tree}(\cdot, \chi_{x}, 5kr\mu)\]. From \Cref{lem:dynamicspanningforest}, the first takes $\tO{\abs{E(K_{i})}}$ work and $\tO{1}$ depth, and the second takes $\tO{kr\mu}$ work and $\tO{1}$ depth. In sum, \eqref{eq:hatT} takes $\tO{kr\mu}$ work and $\tO{1}$ depth with high probability.
        
        In Case 1, $\deg_{G}(\hat{T}^{(i)}) \le 10kr\mu$. Thus \eqref{eq:step4_c1} takes $\tO{kr\mu}$ work and $\tO{1}$ depth.

        In Case 2, \eqref{eq:step4_c2} takes $\tO{kr\mu}$ work and $\tO{1}$ depth. From \Cref{lem:ViHledtres}, \eqref{eq:step4_c2_s2} can be also done in $\tO{\deg_{G}(V_{H, \le\dtres}^{(i)})} = \tO{kr\mu}$ work and $\tO{1}$ depth.
        
        Note that $\abs{V_{H, \le\dtres}^{(i)}}, \abs{E_{H, \le\dtres}^{(i)}} \le \deg_{G}(V_{H, \le\dtres}^{(i)})$. One can sample the vertex $v_{H}^{(i)}$ in $\tO{1}$ work and depth. Also, computing the restricted weight function $\tilde{w}^{(i)}$ can be done in $\tO{\abs{V_{H, \le\dtres}^{(i)}}}$ and $\tO{1}$ depth.

        Finally,
        \[(T^{(i)}_{H,\le\dtres},\tilde{d}_{H,\le\dtres})\leftarrow \alg{SSSP}(G^{(i)}_{H,\le\dtres},x,\tilde{w}^{(i)},\eps)\] takes $\tO{\abs{E_{\aug}^{(i)}}/{\eps^2}} = \tO{kr\mu/\eps^2}$ work and $\tO{1}$ depth. Computing $P^{(i)}$ from $T_{H,\le\dtres}^{(i)}$ takes $\tO{\abs{V_{H, \le\dtres}^{(i)}}}$ work and $\tO{1}$ depth. In sum, Step 4 takes $\tO{kr\mu/\eps^2}$ work and $\tO{1}$ depth with high probability.

        \item[Step 5.] Note that $\abs{V(P^{(i)}}, \abs{E(P^{(i)})} \le 20kr\mu$ from \Cref{lem:ViHdtress} and \Cref{lem:ViHledtres}. $\tW^{(i)}$ can be computed naively in $\tO{\abs{V_{\ne 1}^{(i)}}} = \tO{ikr\mu}$ work and $\tO{1}$ depth. In the \textbf{return} case, $C(u)$ also can be computed in $\tO{\abs{V_{\ne 1}^{(i)}}}$ work and $\tO{1}$ depth. Otherwise, $w^{(i+1)}$ can be computed in $\tO{\abs{V(P^{(i)})}}$ work and $\tO{1}$ depth. In sum, Step 5 can be computed in $\tO{ikr\mu}$ work and $\tO{1}$ depth. Note that $C(u)$ contains $\tO{ikr\mu}$ nonzero entries.
    \end{description}

    Summing up the complexity, the loop runs in $\tO{ikr\mu/\eps^2}$ work and $\tO{1}$ depth with high probability.
\end{proof}

%% file: rounding.tex
\section{Finding Integral $(s,t)$ Cuts (Proof of \Cref{lem:rounding})}\label{sec:rounding}
In this section, we prove \Cref{lem:rounding}. As an efficient routine, we first give an algorithm \Cref{lem:stvertexcut} to compute fractional $(s, t)$-cut in near-linear work and polylogarithmic depth. \Cref{lem:stvertexcut} will be proven at the end of the section. In fact, what it does is an easier variant of the local cut algorithm.

\begin{restatable}[Fractional $(s,t)$-cutm, proved in \Cref{subsec:rounding_deferred}]{lemma}{fractionalstcut}\label{lem:stvertexcut}
    Given an undirected graph $G$, and a positive integer $k$, and two distinct vertices $s$ and $t$. There exists a deterministic \pram{} algorithm $\alg{FractionalSTCut}(G, k, s, t)$ returns a function $C : V \to \bbR_{\ge 0}$ or $\bot$. The algorithm satisfies the following.

    \begin{itemize}
        \item (Correctness) If there is an integral $(s, t)$ vertex separator $S$ of $\abs{S} < k$, the algorithm returns a fractional cut $C$ of size at most $k - 0.5$. Otherwise, the algorithm returns $\bot$.
        \item (Complexity) The algorithm takes $\tO{mk^5}$ work and $\tO{k^3}$ depth.
    \end{itemize}
\end{restatable}

Back to the main algorithm, we start from showing a folklore lemma \Cref{lem:fractional_to_integral} which introduces the notion also appears in the proof.

\fractionalcut*

\begin{proof}
    Given a vertex separator $S$ of $\abs{S} < k$, a fractional cut $C_{S}(v) = \mathbf{1}[v \in S]$ gives a fractional cut of size at most $\abs{S} < k - 0.5$.
    
    For the converse, assume a fractional $(s, t)$-cut $C$ of size at most $k - 0.5$ is given. For any $\theta \in (0, 1)$, define $d(v) := \dist_{G, C_x}(s, v)$, with the sets
    \begin{align} 
    \label{eq:theta_cut}
    \begin{split}
    L_{\theta} &:= \{v \in V \mid d(v) < \theta\},\\\quad S_{\theta} &:=\{v \in V \mid \theta \le d(v) \le \theta + C(v)\},\\R_{\theta} &:= \{v \in V \mid d(v) - C(v) >\theta\}.
    \end{split}
    \end{align}
    Note that $s \in L_{\theta}$ and $t \in R_{\theta}$ for any $\theta \in (0, 1)$, as $C(s)=C(t)=0$ and $d(s)=0, d(t)\ge 1$. Also, $(L_{\theta}, S_{\theta}, R_{\theta})$ is an integral $(s, t)$-cut. If there is an edge $\{l, r\} \in E$ such that $l \in L_{\theta}, r \in R_{\theta}$. As $d$ gives the shortest distance from $s$ to $r$, $d(r) \le d(l) + C(r)$. Yet \[\theta + C(r) < d(r) \le d(l) + C(r) < \theta + C(r)\] contradicts to the definitions of $L_{\theta}, R_{\theta}$. If we sample $\theta$ uniformly at random, $\Pr[v \in S_{\theta}] = C(v)$, which gives $\bbE[\abs{S_{\theta}}] = \sum_{v} C(v) \le k - 0.5$. Thus, there must be a positive number $\phi \in (0, 1)$ such that $\abs{S_{\phi}} < k$.
\end{proof}

\rounding*

\begin{proof} First of all, we assume $k$ is actually the size of minimum cut. To confirm that, find a minimum integer $k'$ such that \[C \gets \alg{FractionalSTCut}(G, k', s, t)\] returns a fractional $(s, t)$ cut. Then, now we reset $k \gets k'$ and set $C$ as such fractional $(s, t)$-cut of size at most $k + 0.5$.

Recall the definition \eqref{eq:theta_cut}. As we have $\abs{S_{\theta}} \ge k-1$, we deduce $\Pr_{\theta \sim U(0, 1)}[\abs{S_{\theta}} = k-1] \ge 1/2$.

    In parallel setup, we replace the exact distance function $d$ with a $(1+\eps)$-approximate shortest distance function $\tilde{d}$, with a constant $\eps > 0$ which will be fixed later.
    \[(\tilde{T}, \tilde{d}) \gets \alg{SSSP}(G, C, s, \eps)\]
    Similar to \eqref{eq:theta_cut}, we extend the definition in the context of approximate distance.
    \begin{align} \label{eq:theta_ext}
    \begin{split}
    \wt{L_{\theta}} &:= \{v \in V \mid \tilde{d}(v) < \theta\},\\\quad \wt{S_{\theta}} &:=\{v \in V \mid \theta \le \tilde{d}(v) \le (1+\eps)(\theta + C(v))\},\\\wt{R_{\theta}} &:= \{v \in V \mid \tilde{d}(v) >(1+\eps)(\theta + C(v)).\}.
    \end{split}
    \end{align}

    We justify this construction by following lemma. The proof will be given in the end of this section.
    \begin{lemma} \label{lem:theta_ext}
        For any $\theta \in (0, 1/(1+\eps))$, $\wt{S_{\theta}}$ is an $(s, t)$-separator. 
    \end{lemma}

    Observe that $S_{\theta} \subseteq \wt{S_{\theta}}$, since \[\theta \le d(v) \le \tilde{d}(v),\quad \tilde{d}(v)\le(1+\eps)d(v) \le (1 + \eps)(\theta + C(v)).\]
    As $\Pr_{\theta \sim U(0, 1)}[\abs{S_{\theta}}=k-1] \ge 1/2$, we get $\Pr_{\theta \sim U(0, 1/(1+\eps))}[\abs{S_{\theta}}=k-1] \ge (1-\eps)/2 > 1/3$ for any $\eps < 1/3$. This proves the observation below.

    \begin{observation} \label{obs:rounding_success}
    Define $G_{\theta} := (V_{\theta}, E_{\theta})$ to be a graph obtained from $\wt{L_{\theta}}$ into a new source $s'$, and $\wt{R_{\theta}}$ into a new sink $t'$. In other words, \[V_{\theta} = S_{\theta} \cup \{s', t'\}, \quad E_{\theta} = E(G[\wt{S_{\theta}}]) \cup \{\{s',v\}\mid v \in \wt{S_{\theta}} \cap N(\wt{L_{\theta}})\} \cup \{\{t',v\}\mid v \in \wt{S_{\theta}} \cap N(\wt{R_{\theta}})\}. \]

    Given $\kappa(s', t') \ge k-1$ clearly from \Cref{lem:theta_ext}, If $\theta$ is sampled uniformly at random from $(0, 1/(1+\eps))$, $\kappa(s', t')=k-1$ with probability at least $1/3$.
    \end{observation}

    Now, we're ready to describe our iterative algoithm. Also, we choose the parameter \[\eps = 1/200.\]

    \begin{description}
    \item[Initialization.] Before starting the algorithm Find the least positive integer $k'$ such that \[C_{k'} \gets \alg{FractionalSTCut}(G, k', s, t)\] returns a cut. Then initialize with \[k \gets k', \quad C^{(1)} \gets C_{k'}, \quad G^{(1)} \gets G,\quad (s^{(1)},t^{(1)}) \gets(s, t)\]

    \item[Loop Iteration.]  For each iteration $h = 1, \cdots, 4\log_{10}(n/k)$, Compute the approximate distance function $\tilde{d}$ with \[(\tilde{T}, \tilde{d}) \gets \alg{SSSP}(G^{(h)}, C^{(h)}, s^{(h)}, \eps).\]
    Sample $K := \lceil 100\log_{10/9} n \rceil$ values of $\theta_{1}, \cdots, \theta_{K}$. Compute

    Compute the contracted graph $G_{\theta_i}$, and denote the new source and the new sink of $G_{\theta_i}$ as $s_{i}', t_{i}'$. Compute
    \[
    C_{i} \gets \alg{FractionalSTCut}(G_{\theta_{i}}, k, s_{i}', t_{i}')
    \]
    If $\alg{FractionalSTCut}(G_{\theta_{i}}, s_{i}', t_{i}', k)$ returns a valid fractional cut $C_{i}$ and satisfies the inequalities of \eqref{eq:constraint_rounding}, we call there is a \emph{good event} for $i$.
    \begin{align} \label{eq:constraint_rounding}
    \begin{split}
    \abs{V_{\theta_i}} &\le 10(2\eps\abs{V} + 4k)\\
    \abs{E_{\theta_i}} &\le 10(6\eps\abs{E} + \frac{k}{\eps}\abs{V})
    \end{split}
    \end{align}
    If any good event is present for $i$, set
    \[G^{(h+1)} \gets G_{\theta_i},\quad C^{(h+1)} \gets C_{i},\quad(s^{(h+1)},t^{(h+1)}) \gets (s_{i}',t_{i}')\] and proceed to the next iteration. Otherwise terminate return $\bot$.
    \item[Termination.] If $h \ge  4\log_{10}(n/k)$, compute exact distance function $d(v) = \dist_{G,C}(s, v)$ using Dijkstra's algorithm. \textbf{Return} $S_{\phi}$ such that $\abs{S_{\phi}} = k-1$.
    \end{description}

    \paragraph{Correctness of $\alg{IntegralSTCut}$.} If we sample $\theta \in (0, 1/(1+\eps))$ uniformly at random,\[\Pr[v \in \wt{S_{\theta}}] \le \frac{\eps}{1+\eps}\tilde{d}(v) + C(v) \le \eps \tilde{d}(v) + C(v) < 1.9\eps + C(v).\]
    The last inequality comes from assuming $\tilde{d}(t) \le 1.9$. It is presumable since if $\tilde{d}(t) > 1.9$, the re-scaled cut $C' := \frac{1.9}{\tilde{d}(t)}C$ still defines a fractional $(s, t)$ cut. From the inequality, now one can pose the size bound of $V_{\theta}$ in expectation.
    \begin{align} \label{eq:rounding_v_bound}
    \begin{split}
    \bbE[\abs{V_{\theta}}] &\le 2 + \bbE[\abs{S_{\theta}}] \le 2 + 1.1\eps \abs{V} + \sum_{v \in V} C(v) \\&\le 1.1\eps \abs{V} + k + 2.5\\&\le 2\eps\abs{V} + 4k.
    \end{split}
    \end{align}
    To bound the size of $E_{\theta}$, observe that there are at most $k/\eps$ vertices such that $C(v)\ge2\eps$.
    \begin{align} \label{eq:rounding_e_bound}
    \begin{split}
    \bbE[\abs{E_{\theta}}] &\le \sum_{v \in S_{\theta}}\deg_{G}(v)\Pr[v \in \wt{S_{\theta}}] \\&\le 1.1\eps \abs{E} + \sum_{\substack{v \in V\\C(v)\ge2\eps}} \deg_{G}(v)C(v) + \sum_{\substack{v \in V\\C(v)<2\eps}} \deg_{G}(v)C(v)\\
    &\le 1.1\eps\abs{E} + \frac{k}{\eps}\abs{V} + 4\epsilon\abs{E}\\
    &\le 6\eps\abs{E} + \frac{k}{\eps}\abs{V}.
    \end{split}
    \end{align}
    From Markov's inequality, \eqref{eq:constraint_rounding} is violated with probability at most $0.1$, respectively. Along with \Cref{obs:rounding_success}, the \emph{good event} happens with probability at least $1 - \frac{2}{3}-0.1-0.1 > 0.1$. Thus a \emph{good event} occurs within $100\log_{10/9} n$ independent $\theta_{i}$'s with probability $1 - 0.9^{100 \log_{10/9} n} = 1-n^{-100}$. Given a good event, It is clear that $G^{(i)}$ preserves the vertex cut of size $k-1$. This proves the correctness of the algorithm.
    \paragraph{Complexity of $\alg{IntegralSTCut}$.} The initialization involves binary search on $k'$. Each binary search invokes $\alg{FractionalSTCut}$, hence it takes $\tO{mk^3}$ work and $\tO{k^3}$ depth from \Cref{lem:stvertexcut}. For the recursion, denote the graph in $i$-th recursion as $G_i = (V_i, E_i)$. $\alg{SSSP}(G, C, s, \eps)$ takes $\tO{\abs{E_i}/\eps^2} = \tO{\abs{E_i}}$ work with $\tO{1} = \tO{1}$ depth.
    
    Computing $G_{\theta_i}$ and $C_{i}$ for $K = \tO{1}$ candidates takes $\tO{\abs{E_i}k^5}$ work and $\tO{k^3}$ depth by \Cref{lem:stvertexcut}. Hence, for the total depth of the recursion is $\tO{k^5}$. The total work is \begin{equation*} \tO{k^5 \sum_{i=0}^{\lceil 4\log_{10}(n/k)\rceil -1} \abs{E_i}} = \tO{mk^5}.\end{equation*}
    
    From \eqref{eq:constraint_rounding},
    \begin{align*}
        \abs{V_{i}} &\le 0.1\abs{V_{i-1}} + 40k\\
        \abs{E_{i}} &\le 0.3\abs{E_{i-1}} + 2000k\abs{V_{i-1}}.
    \end{align*}
    Thus for any $h \ge \log_{10}(n/k)$, \[\abs{V_h} \le 0.1^{h}n + O(k) = O(k),\] thus for any $i > h$, \[\abs{E_{i}} \le 0.3\abs{E_{i-1}} + O(k^2).\] Therefore, from $\abs{E_{h}} \le m$, for any $h' \ge 4\log_{10}(n/k)$, $\abs{E_{h'}} \le 0.3^{3\log_{10}(n/k)}\abs{E_{h}} + O(k^2) = O(k^2)$.

    As a result, when the algorithm reach the base case the graph has only $O(k^2)$ edges. Using the Dijkstra's algorithm takes $\tO{k^2}$ work and $\tO{k^2}$ depth. In sum, the algorithm takes $\tO{mk^5}$ work and $\tO{k^3}$ depth.
\end{proof}

\subsection{Deferred proofs} \label{subsec:rounding_deferred}

\begin{proof}[Proof of \Cref{lem:stvertexcut}]
We describe an algorithm which resembles the algorithm for \Cref{lem:localvertexcut}, but is way more simple.

\begin{description}
    \item[Initialization.] Initialize the weight $w^{(1)}$ as \[w^{(1)}(v)=\begin{cases} 0 & v \in \{s, t\}\\1 & \text{otherwise} \end{cases}\]

    \item[Loop.] Now we loop for $r = \lceil 320k^{3} \ln n \rceil$ iteration. For each $i = 1, \cdots, r$, \[(T^{(i)}, d^{(i)}) \gets \alg{SSSP}(G, w^{(i)}, s, \frac{1}{16(k-1)}).\]
    Take a path $P^{(i)}$ as the path from $s$ to $t$ in $T^{(i)}$. Let $W^{(i)} := \sum_{v \in V} w^{(i)}(v)$. If $w^{(i)}(P^{(i)})>W^{(i)}/(k-0.6)$, return a fractional cut $C(v) = (k-0.5)w^{(i)}(v)/W^{(i)}$. Otherwise, update $w^{(i+1)}(v) \gets w^{(i)}(v)(1 + \frac{1}{16(k-1)})$ for each $v \in P^{(i)}$.

    \item[Correctness.] If the algorithm output a fractional cut $C$, since $P^{(i)}$ is a $(1 + \frac{1}{16(k-1)})$-shortest $(s, t)$-path,
    \[\dist_{G, C}(s, t) \ge \frac{16k-16}{16k-15} \cdot \frac{k-0.5}{k-0.6} \ge 1.\]
    Thus $C$ is a fractional vertex cut with size $\le k - 0.5$. Now we prove assuming all the calls for $\alg{SSSP}$ is correct and there is an $(s, t)$-separator $S$ of size at most $k$, there is an iteration $i$ such that $w^{(i)}(P^{(i)}) > W^{(i)}/(k-0.6)$.

    Similar to \Cref{lem:mwu}, suppose the contrary that $w^{(i)}(P^{(i)}) \le W^{(i)}/(k+0.4).$ Then
    \[W^{(i+1)} \le \left(1 + \frac{1}{16(k-1)} \cdot \frac{1}{k-0.6}\right)W^{(i)},\]
    hence
    \[W^{(r+1)} \le \left(1 + \frac{1}{16(k-1)} \cdot \frac{1}{k-0.6}\right)^{r}W^{(1)} \le n \exp\left(\frac{r}{16(k-1)(k-0.6)}\right).\]
    Meanwhile, any $(s, t)$-path $P^{(i)}$ crosses the separator $S$. Thus there is a vertex $s \in S$ updated more than $r/(k-1)$ times. Hence
    \[w^{(r+1)}(s) \ge \left(1 + \frac{1}{16(k-1)}\right)^{r/(k-1)} \ge \exp\left(\frac{r}{16(k-1)(k-0.8)}\right).\]
    The second inequality comes from $\ln(1+x) \ge \frac{k-1}{k-0.8}x$ for any $0 < x < \frac{0.2}{k-1}$. Since $w^{(r+1)}(s)$ is at most $\le W^{(r+1)}$, one should have
    \[
    n \ge \exp\left(\frac{r}{16(k-1)(k-0.6)} - \frac{r}{16(k-1)(k-0.8)}\right) \ge 
    \exp\left(\frac{r}{80(k-1)(k-0.6)(k-0.8)}\right).
    \]
    which contradicts to our choice of $r = \lceil 80k^{3} \ln n\rceil$.

    \item[Complexity.] The initialization works in $O(n)$ work and the constant depth. Single call for $\alg{SSSP}$ takes $\tO{mk^2}$ work and $\tO{1}$ depth. And the update stage works in $O(n)$ work and the $\tO{1}$ depth. Summing up for $r = \tO{k^3}$ rounds, the overall work is $\tO{mk^5}$ and the depth is $\tO{k^5}$.
\end{description}
\end{proof}
\begin{proof}[Proof of \Cref{lem:theta_ext}]
    $s \in \wt{L_{\theta}}$ for any $\theta \in (0, 1)$, and $t \in \wt{R_{\theta}}$ if $\theta < 1/(1+\eps)$. Also, $(\wt{L_{\theta}}, \wt{S_{\theta}}, \wt{R_{\theta}})$ is an integral vertex cut, otherwise there is an edge $\{l, r\} \in E$ such that $l \in \wt{L_{\theta}}, r \in \wt{R_{\theta}}$. Given $\tilde{d}$ is an $(1+\eps)$ approximate shortest path length to $r$, $\tilde{d}(r) \le (1+\eps)(\tilde{d}(l) + C(r))$. Yet
    \[(1 + \eps)(\theta + C(v)) < \tilde{d}(r) \le (1+\eps)(\tilde{d}(l) + C(r)) < (1+\eps)(\theta + C(r))\]
    the contradiction is deduced.
\end{proof}

%% file: missingproofs.tex
\section{Missing Proofs}\label{sec:misingproofs}

\sssp*

To prove \Cref{thm:sssp}, we utilize the known \pram{} algorithm on undirected graphs with an edge-length. Here the edge-length is a function $w : E \to \bbR_{\ge 0}$, and the weight of a path $P$ in $G$ is denoted by $w(P) = \sum_{e \in E(P)} w(e)$.

\begin{theorem}[\cite{RozhonGHZL22}, Theorem 1.1] There is a deterministic \pram{} algorithm that, given an undirected graph $G = (V, E)$, an edge length function $w$, a vertex $s \in V$ and an approximation factor $\epsilon$, outputs a $(1 + \epsilon)$-approximate $s$-source shortest path tree in $\tO{m/\epsilon^2}$ work and $\tO{1}$ depth. \label{thm:sssp_edge}
\end{theorem}

Now, given a vertex length $\ell$, we assign a new edge-length $w_{\ell}$ by $w_{\ell}(\{u, v\}) = w(u) + w(v)$ for each $\{u, v\} \in E$. \Cref{lem:esptovsp} shows this edge length preserves (approximate) shortest paths.

\begin{lemma} \label{lem:esptovsp}
    Given an undirected graph $G = (V, E)$ and a vertex length $\ell$, If a path $P$ in $G$ is a $(1+\epsilon)$-approximate shortest $(s, t)$-path with respect to the edge length $w_{\ell}$, then $P$ is also a $(1 + \epsilon)$-approximate shortest path with respect to the vertex length $\ell$.
\end{lemma}

\begin{proof}
    Note that for any $(s, t)$-path $P'$, $w_{\ell}(P') = 2\ell(P') - \ell(s) - \ell(t) \le 2\ell(P')$. Let $P^{\ast}$ be a shortest $(s, t)$-path with respect to the vertex length $\ell$. It is clear that $P^{\ast}$ is also a shortest $(s, t)$-path with respect to the edge length $w_{\ell}$.

    \begin{align*}
        \ell(P) - \ell(P^{\ast}) &= \frac{1}{2}(w_{\ell}(P) - w_{\ell}(P^{\ast}))\\
        &\le \frac{\epsilon}{2}w_{\ell}(P^{\ast})\\
        &\le \epsilon \ell(P^{\ast})
    \end{align*}

    Thus $\ell(P) \le (1 + \epsilon)\ell(P^{\ast})$.
\end{proof}

\begin{proof}[Proof of \Cref{thm:sssp}]
    By \Cref{thm:sssp_edge}, one can assign edge weight $w_{\ell}$ and construct $(1+\epsilon)$-approximate $s$-source shortest path tree with respect to the edge length $w_{\ell}$. By \Cref{lem:esptovsp}, it is guaranteed it is also a $(1+\epsilon)$-approximate $s$-source shortest path tree with respect to the vertex length $\ell$.
\end{proof}

%% file: datastructure.tex
\section{The Data Structure}\label{sec:datastructure}

In this section, we prove \Cref{lem:dynamicspanningforest}. A full batch-dynamic graph connectivity algorithm, described previously by Acar et al.~\cite{acar2019parallel}, supports all the required operations in sensitivity spanning forest with amortized complexity. However, their approach does not suit our setting: it might require as much as $\Omega(m)$ work for a component with $m$ edges even when we only delete a small batch of $m' \ll m$ edges. Moreover, since we may never revisit the component again, amortized analysis becomes ineffective. Instead, we adopt the sequential technique introduced in \cite{kapron2013dynamic}, which maintains a dynamic spanning forest with high probability in $\tO{1}$ worst-case complexity, parallelized with batch-parallel euler tour tree from \cite{tseng2019batch}.

\subsection{Preliminaries}

\paragraph{Aggregation operator.}
An \emph{aggregation operator} is a commutative and associative binary operator \(\agg\), mapping two \(B=\tO{1}\)-bit inputs \(x, y\) to a single \(B\)-bit output \(z = x \agg y\). Given a finite set of \(B\)-bit inputs \(S = \{x_{1}, x_{2}, \dots, x_{n}\}\), the aggregated result \(\bigagg_{x \in S} x = x_{1} \agg x_{2} \agg \dots \agg x_{n}\) is well-defined. The \emph{aggregation work-depth}, denoted by \((W_{\agg}, D_{\agg})\), refers to the worst-case work and depth required to compute the aggregation operation \((x, y) \mapsto x \agg y\).

\paragraph{Basic parallel algorithms and data structures.}
A \emph{skip list}, introduced by Pugh~\cite{pugh1990skip}, is a randomized data structure used to efficiently manage linear or circular arrays, supporting fast split and merge operations. Tseng et al.~\cite{tseng2019batch} developed a batch-parallel variant of skip lists with near-linear work and polylogarithmic depth in terms of batch size. A \emph{parallel dictionary} supports batch insertions, deletions, and lookups; \cite{gil1991towards} showed that such a dictionary can achieve linear space, \(O(k)\) work, and \(O(\lg^{\ast} k) = \tO{1}\) depth for batch size \(k\). A \emph{parallel MSF (maximal spanning forest)} is to find a maximal spanning forest of a given undirected graph. This can be run in $\tO{m}$ work and $\tO{1}$ depth, in \cite{awerbuch1987new} or randomized \cite{gazit1991optimal}.

\paragraph{Circular list.}
A \emph{circular (or cyclic) list} of \(n\) elements is a linked list structure where the \(i\)-th element has the \((i+1)\)-th element as its right successor for \(1 \le i \le n - 1\), and the \(n\)-th element points back to the first element. An interval \([s, e]\) in a circular list is the consecutive subsection of elements starting from head \(s\) and ending at tail \(e\). The \emph{first right descendant} of an element \(s\) satisfying a particular property is the closest element \(y\) with that property, minimizing the interval length \(\abs{[s, y]}\).

\paragraph{Euler Tour.}
Given an undirected tree \(T = (V, E)\), we construct a directed graph \(T' = (V, E')\) by replacing each undirected edge in \(T\) with two directed edges (one in each direction), thus forming a bidirectional cycle. Additionally, each vertex \(v\) is assigned a self-loop \((v, v)\). Formally, we define:
\[
E' = \{(v, v) \mid v \in V\} \cup \{(u, v), (v, u) \mid \{u, v\} \in E\}.
\]
An \emph{Euler tour representation} of the tree \(T\) is a cyclic sequence of edges in \(E'\) corresponding precisely to an Eulerian circuit of \(T'\). Explicitly, an Euler tour representation is a permutation \((c_0, c_1, \dots, c_{|E'|-1})\) of edges in \(E'\) where, for each edge \(c_i = (a_i, b_i)\), we have \(b_i = a_j\), with \(j - i \equiv 1 \pmod{|E'|}\).

\paragraph{Euler tour-based data structures.}
An \emph{Euler tour tree} (ET-tree), first introduced by \cite{henzinger1999randomized}, is a dynamic data structure maintaining an Euler tour representation of a forest, supporting efficient updates such as edge insertions and deletions. \cite{tseng2019batch} subsequently developed a batch-parallel version of Euler tour trees. We formalize the definition of their \emph{Batch-parallel Euler Tour Tree} in \Cref{thm:batch_et}.

\subsection{Batch-Parallel Euler Tour Tree}

We recall the batch-parallel data structures from Tseng et al.~\cite{tseng2019batch}, here we list the operations that we need.

\begin{definition}[Batch-Parallel Interval Tree] \label{def:btit}
Given a circular list of \(n\) elements \((x_{1}, \sigma_{1}), \dots, (x_{n}, \sigma_{n})\), where each \(x_{i}\) is a \(B\)-bit input to an aggregation operator \(\agg\) with aggregation work-depth \((W_{\agg}, D_{\agg})\), and each \(\sigma_{i} \in \mathbb{R}^{+}\). A \emph{batch-parallel interval tree} is a data structure supporting the following queries.

\begin{itemize}
    \item \(\alg{Init}(x, \sigma)\): Initializes the structure with data \(x\) and weights \(\sigma\), in \(\tO{nW_{\agg}}\) work and \(\tO{D_{\agg}}\) depth with high probability.
    \item \(\alg{Update}(((j_{i}, x_{i}')_{i=1}^{q}))\): Updates data elements \(x_{j_i}\) to \(x_{i}'\) for each \(i\), in \(\tO{qW_{\agg}}\) work and \(\tO{D_{\agg}}\) depth with high probability.
    \item \(\alg{Agg}(([s_{i}, e_{i}])_{i=1}^{q})\): Returns \(\left(\bigagg_{j \in [s_{i}, e_{i}]} x_{j}\right)_{i=1}^{q}\) in \(\tO{qW_{\agg}}\) work and \(\tO{D_{\agg}}\) depth with high probability.
    \item \(\alg{Sum}(([s_{i}, e_{i}])_{i=1}^{q})\): Returns \(\left(\sum_{j \in [s_{i}, e_{i}]} \sigma_{j}\right)_{i=1}^{q}\) in \(\tO{q}\) work and \(\tO{1}\) depth with high probability.
    \item \(\alg{Search}(b, \theta)\): Given an element \(b\) and \(\theta \in \mathbb{R}^{+}\), returns the first right descendant \(e\) of \(b\) with \(\sum_{j \in [b, e]} \sigma_{j} \ge \theta\) in \(\tO{1}\) work and \(\tO{1}\) depth with high probability. It should be guaranteed that such $e$ exists.
\end{itemize}
\end{definition}

\begin{theorem}[Batch-Parallel Euler Tour Tree~\cite{tseng2019batch}] \label{thm:batch_et}
Given an undirected forest \(F = (V, E)\) with vertex data \((x_{v})_{v \in V}\) and weights \((\sigma_{v})_{v \in V}\), where each \(x_{v}\) is input to an aggregation operator \(\agg\) with aggregation work-depth \((W_{\agg}, D_{\agg})\), and $\sigma_{v} \ge 1$ for any $v \in V$. There exists a \emph{Batch-parallel Euler tour tree} $\cE$ in \pram{} model, which is a randomized data structure supporting:

\begin{itemize}
    \item \(\alg{Init}(F, x, \sigma)\): Initializes the structure on \(F\) in \(\tO{|V|W_{\agg}}\) work and \(\tO{D_{\agg}}\) depth with high probability.
    \item \(\alg{Link}(E')\): Adds edges \(E'\) disjoint from \(E\) and $E \cup E'$ is acyclic, in \(\tO{|E'|W_{\agg}}\) work and \(\tO{D_{\agg}}\) depth with high probability.
    \item \(\alg{Cut}(E')\): Deletes edges \(E' \subseteq E\) in \(\tO{|E'|W_{\agg}}\) work and \(\tO{D_{\agg}}\) depth with high probability.
\end{itemize}

Each connected component \(K \subseteq V\) is assigned a unique identifier \(\id(K)\), forming a set of IDs \(\mathcal{I}_{\mathcal{E}}\). For each \(\id \in \mathcal{I}_{\mathcal{E}}\), the component is denoted \(\mathcal{E}(\id)\). The structure maintains a batch-parallel interval tree (\Cref{def:btit}) \(\cS_{\id}\) representing an Euler tour of each component, associating each vertex node \((v, v)\) with data \(x_{v}\) and weight \(\sigma_{v}\). Any other edge nodes $(u, v)$ receive the data $x$ as identity, and $\sigma = 0$.
\end{theorem}

\begin{corollary} \label{cor:query_batch_et}
    Given a batch-parallel euler tour tree $\cE$, it supports:
       
    \begin{itemize}
        \item $\alg{ID}((v_{i})_{i=1}^{q})$: Returns \((\id_{v_i})_{i=1}^{q}\) where $v_{i} \in \cE(\id_{i})$ in $\tO{q}$ work and $\tO{1}$ depth with high probability.
        \item \(\alg{Update}(((v_{i}, x_{i}')_{i=1}^{q}))\): Updates data elements \(x_{v_i}\) to \(x_{i}'\) for each \(i\), in \(\tO{qW_{\agg}}\) work and \(\tO{D_{\agg}}\) depth with high probability.
        \item $\alg{Agg}((\id_{i})_{i=1}^{q})$: Returns \((\bigagg_{u \in \cE(\id_{i})} x_{u})_{i=1}^{q}\) in \(\tO{qW_{\agg}}\) work and $\tO{D_{\agg}}$ depth with high probability.
        \item $\alg{Sum}((\id_{i})_{i=1}^{q})$: Returns \((\sum_{u \in \cE(\id_{i})} \sigma_{u})_{i=1}^{q}\) in \(\tO{q}\) work and $\tO{1}$ depth with high probability.
        \item $\alg{Components}((\id_{i})_{i=1}^{q})$: Returns a list of connected components, \((\cE(\id_{i}))_{i=1}^{q}\) in $\tO{\sum_{\id \in I'} \abs{\cE(\id)} }$ work and $\tO{1}$ depth with high probability.
        \item $\alg{Tree}(\id, x, q)$: Given an ID $\id \in \cI_{\cE}$, satisfying $\sigma(\cE(\id)) \ge 2q$, returns a tree $T$ in $\tO{q}$ work and $\tO{1}$ depth with high probability., such that $V(T) \subseteq \cE(\id), x \in V(T)$, and one of the following two cases happens.
        \begin{enumerate}
            \item $q\le \sigma(V(T))\le 2q$.
            \item $\sigma(V(T))<q$ and there is a vertex $v\in \cC(\id)$ with $\sigma(v)>q$ such that $v$ is adjacent to $V(T)$ in $E$. 
        \end{enumerate}
    \end{itemize}
\end{corollary}

\begin{proof}[Proof of \Cref{cor:query_batch_et}]
The correctness and complexity of the operations \(\alg{ID}, \alg{Update}, \alg{Agg}, \alg{Sum}\), and \(\alg{Components}\) follow immediately from the properties provided by the Batch-parallel Euler Tour Tree (\Cref{thm:batch_et}) and the underlying Batch-parallel interval tree (\Cref{def:btit}). Thus, we focus on proving the correctness and complexity of the \(\alg{Tree}\) operation.

\paragraph{$\alg{Tree}(\chi, v, q)$.} Given an identifier \(\id \in \cI_{\cE}\), a vertex \(x \in \cE(\id)\), and a weight parameter \(q\), our goal is to identify a subtree \(T\) with \(V(T)\subseteq \cE(\id)\), containing vertex \(x\), satisfying precisely one of the following conditions:

\begin{enumerate}
    \item \(q \le \sigma(V(T)) \le 2q\), or
    \item \(\sigma(V(T)) < q\), and there exists a vertex \(v\in \cE(\id)\), adjacent to \(V(T)\) in \(E\), with \(\sigma(v) > q\).
\end{enumerate}

\paragraph{Simplifying the condition.} It suffices to find a tree $\hat{T}$ with a vertex \(y\) such that \[\sigma(V(\hat{T}) - \{y\}) < q \le \sigma(V(\hat{T})).\] We call this vertex $y$ as a \emph{key vertex}. Given such \(\hat{T}, y\), one can find another tree \(T\) satisfying one of the given conditions. if \(\sigma(V(\hat{T})) > 2q\), the vertex \(y\) satisfies \(\sigma(y)>q\), and we return the connected component of \(\hat{T}_{I}-y\) containing vertex \(x\). Otherwise, $\hat{T}$ is our desired tree satisfying $q \le \sigma(V(\hat{T})) \le 2q$.

\paragraph{Step 1 (Interval Search).}
Perform an interval search by executing:
\[
e \gets \cS_{\id}.\alg{Search}((x,x), q).
\]

Since only vertex nodes have positive weights, without loss of generality, we assume the resulting node \(e\) is a vertex node, denoted by \(e=(y,y)\) for some vertex \(y\in \cE(\id)\).

Define the interval:
\[
I := [(x,x), (y,y)].
\]

We introduce the following notation for an arbitrary interval \(J\) of \(\cS_{\id}\).
\[E_{J} := \{\{u,v\}\in E \mid (u,v)\in J\},\quad V_{J} := \{v\in V \mid (v,v)\in J\}\]
Additionally, denote the set of endpoints in $E_{J}$ as \(V(E_{J}) \subseteq V\).

By the construction of the search query, we immediately have:
\[
\sigma(V_{I}-\{y\}) < q \le \sigma(V_{I}).
\]

The subgraph \(T_{I}:=(V(E_{I}), E_{I})\) is connected and acyclic, hence it forms a tree. However, \(T_{I}\) might contain vertices that are endpoints of edges in the interval but do not appear explicitly as vertex nodes. Define such vertices as \emph{absent vertices}:
\[
A_{I} := V(E_{I}) - V_{I}.
\]

If the set of absent vertices \(A_{I}\) is empty, then the tree $\hat{T}_{I}$ with the key vertex $y$ satisfies desired condition.

\paragraph{Step 2 (Resolving Absent Vertices)}
If \(A_{I}\neq \emptyset\), we proceed carefully. Note that any interval $J$ of Euler tour $\cS_{\chi}$ satisfies
\[
|V(E_{J})|-1\le|E_{J}|\le3|V(E_{J})|-2.
\]

Thus, truncating the interval \(I\) to its prefix \(I'\) of size at most \(3\lceil q\rceil -2\) ensures that \(
|V(E_{I'})|\ge \lceil q\rceil
\), implying that \(\sigma(V(E_{I'}))\ge q\). Next, define the discovery index \(d(v)\) of a vertex \(v\in V(E_{I'})\) as the smallest index of an edge \((a,b)\) within the interval \(I'\) for which \(v\in\{a,b\}\). Let
\[V(E_{I'})=\{v_{1},v_{2},\dots,v_{t}\}\] be vertices ordered by increasing discovery index, i.e.:
\(
1=d(v_{1})\le d(v_{2})\le\cdots\le d(v_{t}).
\) Identify the smallest index \(j\le t\) such that
\(\sigma(\{v_{1},\dots,v_{j}\})\ge q,\) and truncate the interval \(I'\) to the prefix \(I''\) of length \(d(v_{j})\). The resulting subgraph \(\hat{T}_{I''}:=(V(E_{I''}),E_{I''})\) with the key vertex $v_{j}$ meets our desired condition.

\paragraph{Complexity Analysis.}
We now confirm the complexity. In \textbf{Step 1}, Finding the interval \(I\) via \(\alg{Search}\) takes \(\tO{1}\) work and \(\tO{1}\) depth. (from \Cref{def:btit}). Identifying $\abs{A_{I}}$ and determining the cases takes $\tO{1}$ work and $\tO{1}$ depth. In sum, Step 1 consumes $\tO{1}$ work and depth with high probability.

In \textbf{Step 2}, Interval truncations (\(I\to I'\to I''\)) require \(\tO{q}\) work and \(\tO{1}\) depth with high probability., due to the bounded interval size. Determining the absent vertices \(A_{I'}\) and obtaining their weights requires at most \(O(q)\) queries, thus \(\tO{q}\) work and \(\tO{1}\) depth with high probability. Computing discovery indices similarly takes \(\tO{q}\) work and \(\tO{1}\) depth with high probability. Constructing the tree \(\hat{T}_{I''}\) from interval \(I''\) can be done in \(\tO{q}\) work and \(\tO{1}\) depth with high probability using parallel MSF.

Finally, given a tree $\hat{T}$ and the key vertex $y$, returning a tree with desired condition takes $\tO{\abs{V(\hat{T})}} = \tO{q}$ work and $\tO{1}$ depth with high probability using parallel MSF.

Combining these steps yields a total complexity of:
\(\tO{q}\) work, and $\tO{1}$ depth with high probability.
\end{proof}

\begin{remark}
    In original source \cite{tseng2019batch}, \Cref{thm:batch_et,cor:query_batch_et} are written in terms of a slightly different parallelism, named \mtram{}. It allows individually processors to dynamically spawn child processors. However, it is known that any \mtram{} algorithm with work $W$ and depth $D$ can be simulated in $\tO{W/P + D}$ expected time in \pram{} with $P$ processors \cite{10.1145/321812.321815}. i.e., the work and depth from both models are equivalent up to polylogarithmic factors. Applying Markov's inequality to the expected work and depth gives w.h.p. bounds, with additional $\tO{1}$ factors in work and depth.
\end{remark}

\subsection{Cutset Data Structure}

Next, we reiterate the construction of the randomized cutset data structure originated from \cite{kapron2013dynamic}, designed to support dynamic graph connectivity in polylogarithmic worst-case complexity.

\begin{definition}
    Given a simple undirected graph $G = (V, E)$ with $\abs{V} = n$. An \emph{edge classifier} is a randomized variable $c : \{1, \cdots, \lceil 2\lg n \rceil - 1\} \times E \to \{0, 1\}$ such that $\Pr[c(i, e) = 1] = 2^{-i}$ independently for any $i = 1, \cdots, \lceil 2\lg n \rceil  -1$ and $e \in E$.
\end{definition}

Given an edge classifier $c$, we define the set of edges $L_{c}(i) := \{e \in E \mid c(i, e) = 1\}$.

\begin{lemma}[\cite{kapron2013dynamic}, Lemma 3.2] \label{lem:kkm13}
    For any nonempty edge set $W \subseteq E$, with probability at least $1/9$, there exists an integer $i$ such that $\abs{L_{c}(i) \cap W} = 1$.
\end{lemma}

To utilize \Cref{lem:kkm13}, we fix a label $l : E \to \{1, \cdots, n(n-1)/2\}$. Note that $l(e)$ is a $\lceil 2\lg n\rceil-1$-bit string. Also, we denote the bitwise-XOR operator by $\oplus$.

\begin{corollary} \label{cor:kkm_hash}
    Given a vertex set $S \subseteq V$, an edge classifier $c$, and a positive integer $i$. define $H_{c, i}(S)$ to be a $\lceil 2\lg n \rceil - 1$ bit string such that
    \begin{align*}
    H_{c, i}(S) := \bigoplus_{v \in S} \bigoplus_{e \in \delta_{G}(v) \cap L_{c}(i)} l(e).
    \end{align*}

    Assume $E[S, V-S]$ is nonempty. With probability at least $1/9$, there exists an integer $j$ such that $H_{c, j}(S)=l(e)$ for some edge $e \in E[S, V-S]$.
\end{corollary}

\begin{proof}
    The corollary is immediate from \Cref{lem:kkm13} if $H_{c, i}(S) = \bigoplus_{e \in E[S, V-S]} l(e)$. For an edge $e \in L_{c}(i)$, $l(e)$ is counted as a summand as many as the number of endpoints of $e$ included in $S$. It is odd if and only if $e \in E[S, V-S]$.
\end{proof}

\subsection{Proof of \Cref{lem:dynamicspanningforest}}

Now the proof of \Cref{lem:dynamicspanningforest} follows.

\decrementalspanningforest*

\begin{proof}[Proof of \Cref{lem:dynamicspanningforest}]
We first analyze the correctness and complexity of the algorithms for the operations \(\alg{Init}(G,\sigma)\) and \(\alg{Fail}(E')\), which initialize and maintain the spanning forest data structure \(\mathcal{E}\). Given these operations are correctly maintaining the spanning forest structure with high probability, the whole lemma is proven by \Cref{cor:query_batch_et}.

\paragraph{Correctness and complexity of \(\alg{Init}(G,\sigma)\).}  
To initialize the structure, we first find a maximal spanning forest \(F\) of \(G\). This can be performed deterministically in \(\tO{m}\) work and \(\tO{1}\) depth, using a parallel MSF algorithm \cite{awerbuch1987new}. Next, we independently sample \(\Theta(\log^2 n)\) edge classifiers \(c^{(j,d)}\) for \(j = 1, \dots, \lceil 100\log_{9/8} n \rceil\) and \(d=1, \cdots, \lceil \log_{3/2} n \rceil\). For each vertex \(v\), we compute initial vertex data:
\[
x_v := (x_{v,ijd}),\quad x_{v,ijd}:=H_{c^{(j,d)}, i}(\{v\}).
\]

We then define the aggregation operator \(\agg\) as elementwise bitwise XOR, which clearly satisfies \(\tO{1}\) aggregation work and depth. We initialize a batch-parallel Euler tour tree \(\mathcal{E}\) by calling:
\[
\mathcal{E}.\alg{Init}((V,F), x, \sigma),
\]
which by \Cref{thm:batch_et} requires \(\tO{\abs{V}}\) work and \(\tO{1}\) depth with high probability. Thus, the initialization step has the claimed complexity.

\paragraph{Description of \(\alg{Fail}(E')\).}  
We now describe the procedure \(\alg{Fail}(E')\), designed to remove a set of edges \(E'\) from the graph and maintain correctness of the spanning forest data structure \(\mathcal{E}\).

Let \(V'\) be the set of endpoints of edges in \(E'\). We begin by updating vertex data for each \(v \in V'\) as follows:
\[
x_{v,ijd} \gets x_{v,ijd}\oplus c_i^{(j,d)}(e)l(e),
\]
where \(e\) is the incident edge in \(E'\). Then, we delete the affected edges from the spanning forest by executing:
\[
\mathcal{E}.\alg{Cut}(F\cap E'),
\]
which requires \(\tO{\abs{E'}}\) work and \(\tO{1}\) depth with high probability. by \Cref{thm:batch_et}. Then we iteratively recover connectivity through the following procedure, executed for \(r = \lceil \log_{3/2} n \rceil\) rounds. Initialize \[I^{(1)} \gets \cE.\alg{ID}(V'), \quad F^{(1)}\gets F\cap \tE.\]

\paragraph{Recovery Loop:}  
Within each iteration \(d=1,\dots,r\), we execute:

\begin{enumerate}
    \item Compute aggregated hashes:
    \[
    (x_{\id})_{\id\in I^{(d)}} = \mathcal{E}.\alg{Agg}(I^{(d)}),
    \]
    with \(x_{\id,ij}=H_{c^{(j)},i}(\mathcal{E}(\id))\).

    \item For each \(\id\in I^{(d)}\), we attempt to identify an edge \(e_{\id}\in\tE[\mathcal{E}(\id),V-\mathcal{E}(\id)]\) such that:
    \[
    x_{\id,ijd}=l(e_{\id})
    \]
    for any indices \(i,j\), if such edge exists.

    \item Let \(R^{(d)}\subseteq\{e_{\id}\mid \id\in I^{(d)}\}\) be a maximum cardinality subset maintaining the acyclicity of the graph \(F^{(d)}\cup R^{(d)}\). Update $\cE$ and parameters by following order.
    \[
    \mathcal{E}.\alg{Link}(R^{(d)}), \quad I^{(d+1)}\gets \mathcal{E}.\alg{ID}(V'),\quad F^{(d+1)}\gets F^{(d)}\cup R^{(d)}.
    \]
\end{enumerate}

\paragraph{Correctness of $\alg{Fail}(E')$.}  
To show correctness, define an \emph{exposed ID} as an identifier \(\chi\in I^{(d)}\) at round \(d\), such that \(\tE[\mathcal{E}^{(d)}(\chi),V-\mathcal{E}^{(d)}(\chi)]\) is nonempty, where \(\mathcal{E}^{(d)}\) denotes the Euler tour tree at the start of round \(d\). Let \(t^{(d)}\) be the number of exposed IDs at round \(d\). It suffices to show that \(t^{(r+1)}=0\) with high probability.

By \Cref{cor:kkm_hash}, each exposed ID finds a reconnecting edge with high probability, say $1 - n^{-100}$. Note that we sampled $e_{\id}$ only from $x_{\chi, ijd}$ to guarantee independence from previous recovery loops, therefore for any $j, d$ that\[\Pr[L_{c^{(j,d)}}(i) \cap W = 1 \mid R^{(1)}, \cdots, R^{(d-1)}] = \Pr[L_{c^{(j, d)}}(i) \cap W = 1].\]. We call an exposed ID that fails to find such an edge \emph{isolated}. Using Markov's inequality, the number \(f^{(d)}\) of isolated IDs in round \(d\) satisfies $f^{(d)} \le 0.1t^{(d)}$ with high probability.

Non-isolated IDs reconnect into larger components via edges in \(R^{(d)}\), and hence their number reduces by at least half. Thus, the expected number of exposed IDs in round \(d+1\) satisfies:
\[
t^{(d+1)} \le \frac{t^{(d)}-f^{(d)}}{2}+f^{(d)} \le t^{(d)}\frac{1+0.1}{2}<\frac{2}{3}t^{(d)},
\]
with high probability. Conditioned for \(r=\lceil\log_{3/2}n\rceil\) rounds, it implies that
\[
t^{(r+1)}<t^{(1)}(2/3)^{r}<1,
\]
with high probability. Thus, all components of $\cF_{\cD}$ become maximal connected components of \(\tilde{G}\).

\paragraph{Complexity of \(\alg{Fail}(E')\).}  
Before the recovery loop, the updates and initial queries require \(\tO{|E'|}\) work and \(\tO{1}\) depth, as \(|V'|\le 2|E'|\). Inside the loop, \(\alg{Agg}, \alg{Link}\), and \(\alg{ID}\) each require at most \(\tO{\abs{V'}}\) work and \(\tO{1}\) depth with high probability. The maximal subset \(R^{(d)}\) can be computed using parallel MSF in \(\tO{\abs{V'}}\) work and \(\tO{1}\) depth with high probability.

Summing over all \(r=\tO{1}\) rounds yields total complexity in $\tO{\abs{E'}}$ work and $\tO{1}$ depth with high probability.

\end{proof}

%% file: reductiontoBMM.tex
\section{A Parallel Reduction to Maximum Bipartite Matching}\label{sec:reduction}

In this section, we explore the relationship between vertex connectivity and the dense case of the maximum bipartite matching problem (implying dense reachability), which we denote as {\sf D-MBM} (Dense Maximum Bipartite Matching). Here, the input is a dense bipartite graph \(G = (U \cup V, E)\) with \(|E| = \Theta(n^2)\).

\begin{definition}[Dense Maximum Bipartite Matching ({\sf D-MBM})]
Given a dense bipartite graph \(G = (U \cup V, E)\) with \(|E| = \Theta(n^2)\), the {\sf D-MBM} problem is to compute the value of a maximum matching in \(G\).
\end{definition}

The best-known algorithm for {\sf D-MBM} in subpolynomial depth is based on a standard reduction to matrix rank computation, running in \(n^{\omega}\) work, where \(\omega < 2.371553\) denotes the exponent of matrix multiplication \cite{williams2024new}. Moreover, it is known that {\sf D-MBM} implies dense reachability from a folklore reduction. For completeness, we define dense reachability below.

\begin{definition}[Dense Reachability]
Given a directed graph \(G = (V, E)\) and two designated vertices \(s, t \in V\), the \emph{Dense Reachability} problem asks whether there exists a directed path from \(s\) to \(t\) in \(G\).
\end{definition}

Despite decades of research, no subpolynomial depth algorithm has been found for dense reachability that improves upon the \(n^{\omega}\) work bound—which is easily achieved by repeatedly squaring the adjacency matrix. This observation suggests that achieving a work bound significantly better than \(n^{\omega}\) for {\sf D-MBM} in subpolynomial depth presents a natural barrier. Notice that while \(n^{\omega}\) is conjectured to reach \(n^2\) (i.e., linear work), the current best bound for \(\omega\) is \(2.371553.\) 

We now introduce a direct-sum version of {\sf D-MBM}.

\begin{definition}[\(t\)-Min Dense Maximum Bipartite Matching ({\sf t-DMBM\(^{\min}\)})]
Given a collection of \(t\) dense bipartite graphs \(\{G_1, G_2, \ldots, G_t\}\) (each with $n$ vertices and \(\Theta(n^2)\) edges), the {\sf t-DMBM\(^{\min}\)} problem asks to compute the minimum, over all instances, of the sizes of the maximum matchings.
\end{definition}

It would be surprising if the direct-sum of {\sf D-MBM} instances could be solved faster than processing each instance independently. Our main reduction is captured in the following theorem. By “almost linear work” we mean \(L^{1+o(1)}\) where \(L\) is the input size, and by “subpolynomial depth” we mean a depth of \(L^{o(1)}\).

\begin{lemma}
\label{lem:reduction}
Suppose that the \(k\)-vertex connectivity problem on \(n\)-vertex undirected graphs can be solved in almost linear work and subpolynomial depth for some \(k = \Omega(n^{\epsilon})\). Then, the {\sf t-DMBM\(^{\min}\)} problem on \(t = O(n^{1-\epsilon})\) graphs, each with \(O(n^{\epsilon})\) vertices, can also be solved in almost linear work and subpolynomial depth.
\end{lemma}

\Cref{lem:reduction} shows that solving \(k\)-vertex connectivity in almost linear work and subpolynomial depth when \(k\) is polynomial in \(n\) is beyond reach given the current understanding.

\begin{proof}[Proof of \Cref{lem:reduction}]
We prove the lemma in two steps. In the first step, we reduce the {\sf t-DMBM\(^{\min}\)} problem to the \emph{\(t\) Dense Perfect Bipartite Matching} problem (abbreviated as {\sf t-DPBM}). In the second step, we reduce {\sf t-DPBM} to the \(k\)-vertex connectivity problem.

For notational convenience, assume that each instance of the {\sf t-DMBM\(^{\min}\)} problem is a bipartite graph \(G_i = (U_i,V_i,E_i)\) with \(|U_i| = |V_i| = a\).

\paragraph{Step 1: Reduction from {\sf t-DMBM\(^{\min}\)} to {\sf t-DPBM}.}
We first define the \emph{\(t\) Dense Perfect Bipartite Matching} problem.

\begin{definition}[\(t\) Dense Perfect Bipartite Matching ({\sf t-DPBM})]
Given a collection of \(t\) dense bipartite graphs \(\{G_1, G_2, \ldots, G_t\}\), the {\sf t-DPBM} problem asks to decide whether there exists at least one graph among the \(t\) instances that does \emph{not} have a perfect matching.
\end{definition}

To reduce {\sf t-DMBM\(^{\min}\)} to {\sf t-DPBM}, we perform a binary search over a parameter \(x\) that represents the amount of augmentation applied to each graph. For each graph \(G_i = (U_i,V_i,E_i)\), we construct an augmented graph \(G'_i(x)\) by:
\begin{enumerate}
    \item Adding \(x\) new vertices to \(U_i\) (denote this set by \(U'_i\)) and connecting each new vertex to every vertex in \(V_i\).
    \item Adding \(x\) new vertices to \(V_i\) (denote this set by \(V'_i\)) and connecting each new vertex to every vertex in \(U_i\).
\end{enumerate}
In the augmented graph \(G'_i(x)\), the bipartition now has \(a+x\) vertices on each side. A perfect matching in \(G'_i(x)\) would have size \(a+x\).

The key observation is as follows. Suppose the maximum matching in \(G_i\) has size \(m_i\). Then, in \(G'_i(x)\) one can extend any matching in \(G_i\) by matching at most \(x\) newly added vertices (since they are adjacent to all vertices on the opposite side). Thus, if \(m_i \ge a - x\), one can obtain a matching of size at least \(a+x\) in $G'_i(x)$. Conversely, if \(m_i \le a - x - 1\), then even after augmentation, \(G'_i(x)\) cannot have a perfect matching, as otherwise deleting newly added vertices results in a matching in $G_i$ of size at least $a-x$.

Let \(x^*\) be the largest value of \(x\) for which the {\sf t-DPBM} algorithm outputs `yes'. It follows that the minimum over the maximum matching sizes of the original graphs is exactly \(a - x^* - 1\). 

\paragraph{Step 2: Reduction from {\sf t-DPBM} to \(k\)-vertex connectivity.}
We now describe how to reduce the {\sf t-DPBM} problem to an instance of the \(k\)-vertex connectivity problem. Given the \(t\) bipartite graphs \(G_i = (U_i,V_i,E_i)\), we construct a single undirected graph \(H\) as follows:
\begin{enumerate}
    \item For each \(i=1,\dots,t\), include all vertices and edges of \(G_i\) in \(H\).
    \item For every \(i=1,\dots,t-1\), add additional edges to make the set \(V_i \cup U_{i+1}\) a clique.
\end{enumerate}
This construction “concatenates” the \(t\) graphs into one global graph. Observe that \(H\) has \(\Theta(at)\) vertices and \(\Theta(a^2 t)\) edges in total.

We choose the parameters so that \(a = \Theta(n^{\epsilon})\) and \(t = a^{\frac{1}{\epsilon}-1} = \Theta(n^{1-\epsilon})\). We then set the connectivity parameter \(k = a = \Omega(n^{\epsilon})\). The key claim is that \(H\) is \(k\)-vertex connected if and only if every one of the original graphs \(G_i\) has a perfect matching.

\paragraph{Correctness.} It suffices to show the claim in Step 2. Define
\begin{align*}
L_{i} &:= U_{i} \cup \bigcup_{j=0}^{i-1}V(G_{j}),\\
R_{i} &:= V_{i} \cup \bigcup_{j=i+1}^tV(G_{j}).
\end{align*}
Note that $E[L_{i}, R_{i}] = E[U_{i}, V_{i}]$. Suppose there exists a bipartite graph $G_{i}$ that does not have a perfect matching. Let $C$ be a minimum vertex cover of $G_{i}$. By K\"{o}nig's theorem, $\abs{C} < k$. Notice that there are no edges between $U_{i} - C$ and $V_{i} - C$, therefore not in between $L_{i} - C$ and $R_{i} - C$. Thus, $C$ is a vertex cut of $H$.

For the converse, suppose there is a vertex cut $S$ of size less than $k = a$. As $\abs{U_{i}} = \abs{V_{i}} = a$ for all $1 = 1, \cdots, t$, $U_{i} - S, V_{i} - S \neq \emptyset$ for all $i$. If $E[U_{i} - S, V_{i} - S] \neq \emptyset$ for all $i$, the graph is connected. Thus, there is an integer $j$ such that $E[U_{j} - S, V_{j} - S] = \emptyset$. $N_{G_j}(U_{j} - S) = \emptyset$ implies $G_{j}$ does not have a perfect matching.

\paragraph{Complexity.} In Step 1, it involves $O(\log a)$ instances of {\sf t-DPBM} where the number of vertices in each graph is $O(a)$. Constructing $G_{i}'(x)$ introduces $O(a^2 t) = O(m)$ new edges, thus the almost-linear work and subpolynomial depth solution for {\sf t-DPBM} gives the same for {\sf t-DMBM}$^{\min}$. In Step 2, there are $m = O(a^2t) = O(n^{1+\eps})$ edges. Given a vertex cut $S$ is returned in $\tO{m^{1+o(1)}} = n^{1+\eps+o(1)}$ work and $n^{o(1)}$ depth, one can find such $j$ that $N_{G_j}(U_{j} - S) = \emptyset$ in $O(a^2t)$ work and $\tO{1}$ depth. In sum, {\sf t-DMBM}$^{\min}$ is reduced to $k$-vertex connectivity with near-linear additional work and subpolynomial depth.
\end{proof}